\documentclass[12pt]{article}
\usepackage{amsfonts}
\usepackage[dvips]{color}
\usepackage{xcolor}

\usepackage{enumerate}
\usepackage{amssymb}
\usepackage{graphicx}
\usepackage{tikz}
\usepackage{tikz-qtree}
\usepackage{amsmath}
\usepackage{amsthm}
\usepackage{mathrsfs}
\usepackage{natbib}
\usepackage{booktabs}
\usepackage{rotating}
\usepackage{afterpage}
\usepackage{cancel}
\usepackage{mathtools}
\usepackage{comment}
\usepackage{xfrac}
\usepackage{setspace}
\AtBeginDocument{%
  \setlength{\abovedisplayskip}{5pt plus 2pt minus 2pt}%
  \setlength{\belowdisplayskip}{5pt plus 2pt minus 2pt}%
  \setlength{\abovedisplayshortskip}{2pt plus 2pt}%
  \setlength{\belowdisplayshortskip}{3pt plus 2pt minus 2pt}%
  \setlength{\textfloatsep}{8pt plus 2pt minus 2pt}%
  \setlength{\intextsep}{8pt plus 2pt minus 2pt}%
}

\newcommand*{\bs}[1]{ \boldsymbol{#1} }
\newcommand{\indep}{\perp \!\!\! \perp}

\newcounter{pkt}
{\begin{list}{(\roman{pkt})}{\usecounter{pkt}\parskip0ex\parsep0ex\itemsep0ex\topsep0ex}}%
	{\end{list}}

\usepackage{url}       %
\usepackage{xcolor}
\usepackage[margin=1.25in]{geometry}

\usepackage{graphicx}

\definecolor{jblue}{RGB}{0,191,255}
\definecolor{jgold}{RGB}{212,202,58}
\definecolor{jmagenta}{RGB}{255,109,174}
\definecolor{jgreen}{RGB}{0,183,141}
\definecolor{jgreen2}{RGB}{23,139,23}
\definecolor{jorange}{RGB}{255,165,0}

\theoremstyle{plain}
\newtheorem{theorem}{Theorem}
\newtheorem*{theorem*}{Theorem}
\newtheorem{lemma}{Lemma}
\newtheorem*{lemma*}{Lemma}
\newtheorem{assumption}{Assumption}
\theoremstyle{definition}
\newtheorem{remark}{Remark}
\newtheorem*{remark*}{Remark}
\newtheorem{example}{Example}
\newtheorem{extension}{Extension}

\numberwithin{theorem}{section}
\numberwithin{lemma}{section}
\numberwithin{remark}{section}
\newcommand{\blackqed}{\hfill \ensuremath{\blacksquare}}
\newcommand{\argmin}{\operatorname*{argmin}}

\newcounter{saveeqn}

\usepackage{hyperref}
\usepackage{enumitem}

\graphicspath{{./figures/}}

\begin{document}

    \title{\vspace*{-2cm} \bf 
        Large-Scale Estimation under Unknown Heteroskedasticity\footnote{I thank my doctoral advisors Xu Cheng, Frank Schorfheide and Petra Todd for their guidance. I also thank participants at numerous seminars, in particular the seminar series at University of Pennsylvania and SMU, as well as referees, for their valuable comments that have improved the paper. All errors are my own.}
    }

    \author{Sheng Chao Ho\thanks{Correspondence: School of Economics, Singapore Management University, 90 Stamford Road, Singapore, S178903. scho@smu.edu.sg. }}

    \date{This Version: \today}

    \maketitle

    \begin{abstract} 	
        This paper studies nonparametric empirical Bayes methods in a heterogeneous parameters framework that features unknown means and variances. We provide extended Tweedie’s formulae that express the (infeasible) optimal estimators of heterogeneous parameters, such as unit-specific means or quantiles, in terms of the density of certain sufficient statistics. These are used to propose feasible versions with nearly parametric regret bounds of the order of $(\log n)^\kappa / n$. The results rely on a distributional assumption, and thus a misspecification analysis is also presented. The estimators are employed in a study of teachers' value-added, where we find that allowing for heterogeneous variances across teachers is crucial for delivering optimal estimates of teacher quality and detecting low-performing teachers.
    \end{abstract}

    \section{Introduction}
\label{sec:intro}

The central task of this paper is the large-scale, optimal estimation of unit-specific parameters under unknown heteroskedasticity. 
Motivating examples\footnote{Examples of the following three applications are \cite{Chetty2018}, \cite{Gilraine2020}, and \cite{Liu2020} respectively.} include estimating intergenerational mobility of neighborhoods to guide families toward areas with the highest potential for upward mobility, or teacher value-added for use as inputs in high-stakes labor decisions. 
We work within the following model in which, or variations of it, the estimation problem of these applications are framed\footnote{We focus on \eqref{eq:disagg} for common $J_i=J$ and discuss variations such as heterogeneous $J_i$ subsequently in Section~\ref{sec:ext}. This is because the central feature of our analysis is the unknown heteroskedasticity.}:
\begin{equation}
	y_{ij} = \mu_i + \sigma_i \cdot \epsilon_{ij}, 
	\quad 
	\epsilon_{ij} \mid \mu_i,\sigma_i
	\sim_{\text{iid}} p_\epsilon
	\label{eq:disagg}
\end{equation}
for $j=1,\dots,J$, where $p_\epsilon$ has zero mean and unit variance and the unobserved $(\mu_i,\sigma_i) \sim_{\text{iid}} G_0$ for $i=1,\dots,n$.
For now, let us consider estimating teacher value-added as the running example: the outcome $y_{ij}$ represents the test score of the $j^{\text{th}}$ student in teacher $i$'s class, often after partialling out student-specific controls such that the independence of $\epsilon_{ij}$ is reasonable. 
A common objective is to estimate $\left\{\mu_i\right\}_{i\in[n]}$, interpretable as how the average student (i.e., $\epsilon_{ij}=0$) will perform under each teacher and commonly termed as the teachers' mean value-added.

Relative to existing large-scale estimation methods that impose homoskedasticity, we develop optimal estimators allowing for unknown heteroskedasticity, which matters in two respects. First, $\sigma_i$ is a key component of the optimal estimator of $\mu_i$: restricting $\sigma_i=\sigma$ inflates risk and can systematically under- or over-estimate certain subsets of $\{\mu_i\}_{i\in[n]}$. Furthermore, $\sigma_i$ is important for the estimation of unit-specific quantiles
\begin{equation}
	q_{\alpha,i}:= \inf \left\{ \epsilon\in\mathbb{R} : \alpha \leq F_\epsilon(\epsilon \mid \mu_i,\sigma_i) \right\} 
	\text{ for some $\alpha\in(0,1)$}
	\label{eq:quantiles}
\end{equation}
that are of policy-relevance, and a generalization of the commonly targeted mean value-added $\mu_i$. 
For example, in the teacher value-added context, this means estimating and comparing how the bottom (say) 10\% student performs under each teacher rather than only the average student -- a more relevant object when policymakers care about the lower tail of student performance.

In many applications $J \ll n$, with two implications. First, the few observations per unit make estimates of $\mu_i$ or $q_{\alpha,i}$ based on unit $i$ alone imprecise, and hence high-risk. Second, the large $n$ lets us borrow information across units to sharpen individual estimates and reduce average risk.    
More precisely, it is known that the estimator of $\mu_i$ or $q_{\alpha,i}$ minimizing the average (across $i$) risk under quadratic loss is its posterior mean that takes $G_0$, the distribution of $(\mu_i,\sigma_i)$, as the prior. 
This forms the oracle estimator\footnote{See the end of this Section for notational details.} $\mathbb{E}_{G_0}[\, \cdot \mid \mathcal{Y}_i ]$, where $\mathcal{Y}_i$ is the set of observations associated with unit $i$.

Under a sampling-distributional assumption, the oracle is still infeasible because $G_0$ is unknown in practice, and feasible empirical Bayes versions that exploit the cross-sectional information to mimic the oracle $\mathbb{E}_{G_0}[\, \cdot \mid \mathcal{Y}_i ]$ have to be proposed.
However, nonparametric estimation of $G_0$ using the cross-section is known to have slow rates of convergence because we are essentially estimating a density of unobservables. This suggests that emulating the oracle will be a difficult task. In response, this paper proves 
an extension of the Tweedie's formula\footnote{We review this in Section~\ref{sec:lit}.} to the setting of unknown heteroskedasticity, where the oracle is shown to depend on $G_0$ only through the density of certain sufficient statistics. 
Thus the problem of emulating the oracle is dramatically simpler, since the performance of feasible versions now only hinge on their ability to estimate a density of observables.
We exploit this insight to provide regret bounds for the proposed estimators that are of the order $(\log n)^\kappa /n$ for some $\kappa>0$. 
The upshot is that there is little cost to practitioners adopting a nonparametric empirical Bayes approach where both the means and variances are unknown and heterogeneous, whereas the benefits in terms of risk reduction is often substantial relative to the commonly used parametric modeling of $G_0$.

The remaining paper is organized as follows. Section~\ref{sec:lit} reviews the literature and Section~\ref{sec:model} defines the model, risk and optimal estimators. Section~\ref{sec:theory} introduces feasible versions of the oracles and proves their regret convergence rates, and Section~\ref{sec:ext} discusses extensions including heterogeneous $J_i$ and optimal forecasting in large-$n$ short-$T$ panels. Section~\ref{sec:missp} studies misspecification of the sampling distribution, Section~\ref{sec:simul} evaluates finite-sample performance via numerical experiments, and Section~\ref{sec:appl} demonstrates the importance of unknown heteroskedasticity in an application to teacher quality. Section~\ref{sec:concl} concludes.

\noindent\textit{Notations.} 
The symbol $[n]$ refers to the set $\{1,\dots,n\}$. The symbol
$\mathcal{Y}_i$ refers to the set of observations associated with unit $i$ (i.e., $\{y_{ij}\}_{j\in[J]}$) and $\mathcal{Y}$ is taken to refer to the full set of observations (i.e., $\{y_{ij}\}_{i\in[n],j\in[J]}$).
A subscript on an expectation denotes the (prior) distribution of the latent variables, while a superscript makes explicit the object being integrated out; for instance $\mathbb{E}_{G}^{\mu}[\mu \mid \mathcal{Y}_i]$ integrates out the unknown $\mu$.
Finally let $\partial_y f(y,x)$ denote the partial derivative of $f(y,x)$ with respect to $y$.
    \section{Related Work}
\label{sec:lit}
Our paper belongs to the empirical Bayes methodology, which goes back at least to \cite{Efron1973}'s interpretation of the James-Stein estimator as the posterior mean under a normal-normal hierarchical model. This is an instance of parametric empirical Bayes, where a class of priors is posited and feasible estimators emulate the optimal Bayes estimate within the implied class; recent advances include \cite{Xie2012} and \cite{Kwon2026}, with applications in \cite{Chetty2014}, \cite{Finkelstein2016} and others. Such methods can be restrictive, however, since the choice of prior may induce undesirable behavior -- for instance, the normal-normal model imposes equal shrinkage on all units with the same number of observations. Nonparametric empirical Bayes instead targets the posterior mean without a parametric prior, affording more flexible shrinkage, but raises finite-sample concerns because nonparametric estimation of $G_0$, a deconvolution problem, has slow convergence rates (as slow as $(\log n)^{-1}$; see \cite{Fan1991}).

Under homoskedasticity ($\sigma_i=\sigma_0$ for all $i$), progress on the nonparametric route was enabled by the Tweedie representation of the posterior mean\footnote{In the homoskedastic case $G_0$ denotes the univariate distribution of $\mu_i$.} under a normal sampling model (see \cite{Efron2011}):
\begin{equation}
	\mathbb{E}_{G_0}[\mu_i\mid \sigma_0,\mathcal{Y}_i]
	=
	y_i + \frac{\sigma_0^2}{J} \cdot \partial_y\log f_{G_0}(y_i \mid \sigma_0)
	\label{eq:tweed.hom}
\end{equation}
where $y_i$ is the sample mean and $f_{G_0}(y_i \mid \sigma_0)$ its mixture density. The key insight is that the Bayes correction requires only $f_{G_0}(y_i \mid \sigma_0)$, not the prior $G_0$ itself. This was instrumental for \cite{Jiang2009}, who use a nonparametric plug-in $\hat{G}$ to establish fast regret convergence rates for $\mathbb{E}_{\hat{G}}[\mu_i\mid \sigma_0,\mathcal{Y}_i]$ in estimating $\{\mu_i\}_{i\in[n]}$. This was extended by \cite{Jiang2020} to known heteroskedasticity with the heteroskedasticity independent of $\mu_i$, and by \cite{Chen2026} to allow dependence between $\mu_i$ and the known heteroskedasticity through the first two moments.\footnote{Remark~\ref{rmk:het} elaborates on our setting and contributions relative to these papers.}

Our paper also takes the nonparametric route but in the more general setting of unknown heteroskedasticity. This substantially complicates estimation of $\mu_i$, since $\sigma_i$ -- which regulates the Bayes correction in \eqref{eq:tweed.hom} -- is now unknown, and the insight above no longer applies because the form of $\mathbb{E}_{G_0}^\sigma [\sigma_i^2 \cdot \partial_y\log f_{G_0}(y_i | \sigma_i) \mid \mathcal{Y}_i]$, and whether it depends on $G_0$ only through a density of observables, is not known. The papers closest in framework are \cite{Gu2017}, \cite{Gu2017a} and \cite{Banerjee2023}. \cite{Gu2017} and \cite{Gu2017a} study $\mathbb{E}_{\hat{G}}[\mu_i\mid \mathcal{Y}_i]$ and $\mathbb{E}_{\hat{G}}[\sigma_i^2\mid\mathcal{Y}_i]$ for a nonparametric $\hat{G}$ via simulations with promising results. We bolster these through a substantial extension of Tweedie's formula accommodating unknown heteroskedasticity and use it to study the theoretical properties of $\mathbb{E}_{\hat{G}}[\mu_i \mid \mathcal{Y}_i]$, crucially without restricting the dependence within $(\mu_i,\sigma_i)$.

The same model is considered by \cite{Banerjee2023}, who study optimal estimation of $\{\mu_i\}_{i\in[n]}$ but under a precision-weighted loss that, roughly speaking, transforms the heteroskedastic problem back into a homoskedastic one -- less appropriate for our applications, where each unit's loss is equally important. \cite{Ignatiadis2025} considers an identical setting in a partially empirical Bayes framework (means fixed, variances random with a common distribution), targeting multiple testing on the means with guarantees such as FDR control; we instead adopt a fully empirical Bayes framework and study optimal estimation.

More generally, our problem relates to learning the conditional expectation of a latent variable given observables. \cite{Ignatiadis2023} estimate means in a similar setting but with possibly non-normal sampling models and available microdata, using the $J$ replicates and sample-splitting to recast the problem as a regression problem. Our empirical Bayes approach requires only the sufficient statistics, which as aggregates are more readily available with confidential microdata (see e.g. \cite{Chetty2014}); and our near-parametric regret results assume a Gaussian sampling model while restricting the prior little, whereas their parametric rates require assumptions on the joint distribution of observed and latent variables (e.g. exponential families with conjugate priors), so the two are non-nested and complementary.

Finally, beyond the means $\{\mu_i\}_{i\in[n]}$, we study estimators for objectives such as the quantiles $\{q_{\alpha,i}\}_{i\in[n]}$ -- a policy-relevant generalization of the means -- and the variances $\{\sigma_i^2\}_{i\in[n]}$ or standard deviations $\{\sigma_i\}_{i\in[n]}$ that may matter in other applications.

    \section{Model and Risk}
\label{sec:model}

\subsection{Model}
\label{subsec:model}

We maintain the following model throughout: 
\begin{equation}
	\label{eq:disagg.model}
	y_{ij} = \mu_i + \sigma_i \cdot \epsilon_{ij}, 
	\quad 
	\epsilon_{ij} \mid \mu_i,\sigma_i
	\sim_{\text{iid}} p_\epsilon
	\quad\text{ and }\quad
	(\mu_i,\sigma_i) \sim_{\text{iid}} G_0
\end{equation}
for $j=1,\dots,J$, and $i=1,\dots,n$, and where $p_\epsilon$ has zero mean and unit variance. Extensions are discussed in Section~\ref{sec:ext}.
Matched datasets, such as the student-teacher data of our empirical application, fit into \eqref{eq:disagg.model} as follows. Suppose that there are $\bar{J}$ students in total, with each student $j$ represented by $\tilde{\epsilon}_j$ drawn from a common distribution, who are matched to $n$ teachers. Let $i(j): [\bar{J}] \mapsto [n]$ be a function that returns the teacher identity of student $j$ resulting from this matching process. The outcomes of the students are realized after the matching process as 
$\tilde{y}_j = \mu_{i(j)} + \sigma_{i(j)} \tilde{\epsilon}_j$.
This fits into our framework \eqref{eq:disagg.model} by defining $J := | \{j\in[\bar{J}]:i(j)=i\} | $, and
$\{\epsilon_{ij}\}_{j\in[J]} := \{\tilde{\epsilon}_{j}\}_{i(j)=i}$ and similarly for the $\{y_{ij}\}_{j\in[J]}$. The conditional independence then requires that the matching did not depend on $\{\tilde{\epsilon}_{j}\}_{j=1}^{\bar{J}}$.
More realistically, $\tilde{y}_{j}$ is taken to be the raw test score with covariates partialled out, $\tilde{y}_{j} = \tilde{y}_{j}^* - x_{j}^\prime \beta_0$, where in the teacher value-added literature $x_{j}$ typically includes student demographics or lagged test scores such that the independence assumption is reasonable. We abstract from estimation of the common vector $\beta_0$ as is typical in the literature.
% , see e.g., \cite{Gilraine2020} and \cite{Kwon2023b}. 

An essential difference between \eqref{eq:disagg.model} and models commonly studied in the empirical Bayes literature is the heterogeneous and unknown variances $\{\sigma_i^2\}_{i\in[n]}$. Beyond being empirically realistic, it also allows for meaningful comparisons of units beyond their mean parameters. We provide two motivating examples of this and Section~\ref{sec:ext} further motivates the unknown heteroskedasticity in optimal forecasting. 
%Furthermore, a common policy study is to analyze the effects of removing the bottom $p$ percent of teachers on marginal and average student outcomes, such as in \cite{Gilraine2020}. The standard of measure has been the mean outcome of teachers i.e., $\mu_i$. However, if the interest is in the policy benefits of removing the tails of the teacher distribution, it may be of more relevance to focus on the teacher-specific quantiles for low $\alpha$.

\begin{example}[Teacher Value-added]
	A policymaker concerned about lower-ability students may want to compare teachers by the lower end of their value-added rather than the mean, which only reflects the average student. The relevant measure is then $q_{\alpha,i}$ of \eqref{eq:quantiles} for a low $\alpha$, which represents the performance of the $\alpha$-quantile student (in terms of $\tilde{\epsilon}_j$) under teacher $i$.
	\blackqed
\end{example}
\begin{example}[Neighborhood Effects]
	There is growing interest in measuring the mean causal outcome of growing up in county $i$ on adulthood income.  
	A motivation given in \cite{Chetty2018} is to ``construct forecasts of the causal effect of growing up in each county that can be used to guide families seeking to move to better areas''.
	However, even if a neighborhood produces good outcomes on average (i.e. $\mu_i$ is large), this may be an imprecise measure for a given child -- for instance in neighborhoods with highly unequal outcomes. The $\{q_{\alpha,i}\}_{i\in[n]}$, which convey the likelihood of different future income levels, are thus useful complements to $\{\mu_i\}_{i\in[n]}$.
	\blackqed
\end{example}
Therefore, we also derive optimal estimators of $\{q_{\alpha,i}\}_{i\in[n]}$, and also of $\{\sigma_i\}_{i\in[n]}$ and $\{\sigma_i^2\}_{i\in[n]}$ that may be of separate interest.
To make progress, we augment \eqref{eq:disagg.model} with a distributional restriction and summarize our modeling choice in Assumption~\ref{ass:disagg.model}.
\begin{assumption}
	\label{ass:disagg.model}
	Suppose that (i) equation \eqref{eq:disagg.model} holds and (ii)
	$p_\epsilon(x) = \phi(x)$ where $\phi(\cdot)$ denotes the standard normal density.
\end{assumption} 
Several remarks now contrast this with the model under which existing optimal estimation results are typically derived (summarized in Assumption~\ref{ass:agg.model}) and also reveal their complementary nature. 
\begin{assumption}
	\label{ass:agg.model}
	Suppose that (i) equation \eqref{eq:disagg.model} holds with $\sigma_i=\sigma_0$ for a constant $\sigma_0$ 
	and (ii) $p_{z}(x) = \phi(x)$ where $p_{z}(\cdot)$ denotes the density of $\frac{1}{\sqrt{J}}\sum_{j=1}^J \epsilon_{ij}$.
\end{assumption}
\begin{remark}[Microdata normality]
	\label{rmk:agg.model}
	Assumption~\ref{ass:agg.model}(ii) holds approximately by the CLT. While we impose normality at the microdata level (Assumption~\ref{ass:disagg.model}(ii)), these are conceptually equivalent given i.i.d. microdata. In practice neither is likely to hold exactly, and Section~\ref{sec:missp} studies the effect of violations on optimal estimation. Briefly, the optimal estimators under Assumption~\ref{ass:disagg.model} exploit normality more fully than those under Assumption~\ref{ass:agg.model} and are thus more sensitive to it, but both incur misspecification errors scaling as $\frac{1}{J}$ and $\frac{1}{J^2}$ respectively and so are relatively benign for moderate $J$. Such violations is of course application-dependent; in our empirical illustration we find no evidence against microdata normality (Section~\ref{sec:appl}).
	\blackqed
\end{remark}
\begin{remark}[Microdata heteroskedasticity]
	\label{rmk:hom}
	Existing empirical Bayes analyses commonly add a homoskedasticity restriction (Assumption~\ref{ass:agg.model}(i)), whereas our model allows heteroskedasticity and is substantially richer: integrating $\sigma_i$ out of \eqref{eq:disagg.model} using $G_0$ yields scale mixtures of normals -- including the student-$t$ with arbitrary degrees of freedom and any symmetric stable distribution -- a larger class than the normal family obtained by \textit{a priori} restricting $\sigma_i$ to a point mass at $\sigma_0$. An alternative approach to the heteroskedasticity is simply plugging in sample variances as if they were the true variances; our simulations and empirical illustration show that doing so, or assuming homoskedasticity, can be substantially sub-optimal. Our results thus complement the literature by delivering optimality where heteroskedasticity is warranted at the expense of taking microdata normality more seriously. This ultimately enlarges the relevance of the nonparametric empirical Bayes toolkit to practitioners.
	\blackqed
\end{remark}
\begin{remark}[Other channels of heteroskedasticity]
	\label{rmk:het}
	Starting from an aggregated model (modeling the $i$-specific sample mean rather than the microdata) under Assumption~\ref{ass:agg.model}(ii), \cite{Jiang2020} and \cite{Chen2026} provide regret bounds where the heteroskedasticity of the sample means operates through heterogeneous $J_i$, with $\sigma_i$ homogeneous or known and varying dependence allowed within $(J_i,\mu_i,\sigma_i)$. We instead let heteroskedasticity operate through the unknown $\sigma_i$ with homogeneous $J_i=J$ and with no restriction on the dependence within $(\mu_i,\sigma_i)$. When the heterogeneity in $J_i$ is limited, as in our teacher value-added illustration, one can bin units by class size and estimate within bins; the results from Section~\ref{sec:theory} suggest this retains good finite-sample performance. Alternatively, if $J_i$ is independent of $(\mu_i,\sigma_i)$, our results apply directly at this most general level of heterogeneity.
	\blackqed
\end{remark}

\subsection{Risk}
\label{subsec:compoundrisk}

The discussion in this section will be framed in terms of the mean value-added $\{\mu_i\}_{i\in[n]}$, because these are the common object of interest. Our theoretical results also apply to other objects of interest including $\{q_{\alpha,i}\}_{i\in[n]}$ as we show in Section~\ref{sec:theory}.

An estimator $\hat{\mu}$, which maps the observations $\mathcal{Y}$ to $\mathbb{R}^n$, is evaluated under the loss function
$l(\hat{\mu},\mu) := \frac{1}{n}\sum_{i=1}^{n} [\hat{\mu}_i-\mu_i]^2$. 
Define the (integrated) risk
\begin{equation}
	R_{G_0}(\hat{\mu},\mu)
	:=
	\mathbb{E}_{G_0}^{(\mu,\sigma)}
	\mathbb{E}_{(\mu,\sigma)}^{\mathcal{Y}}
	l(\hat{\mu},\mu).
	\label{eq:risk.mse}
\end{equation}
Assumption~\ref{ass:disagg.model} together with the Rao-Blackwell Theorem imply that $\hat\mu^*$, the minimizer of $R_{G_0}(\hat{\mu},\mu)$ over all borel mappings of the data $\mathcal{Y}$ to $\mathbb{R}^n$, satisfies
\begin{equation}
		\hat{\mu}_i^* 
		= 
		\hat{\mu}_{G_0}(y_i,s_i^2) 
		:= 
		\mathbb{E}_{G_0}\left[\mu_i \mid y_i,s_i^2\right]
	\label{eq:oracle.est}
\end{equation}
and $(y_i,s_i^2)$ are the unit-specific sample mean and variance.
Similar statements also hold for the decision problems involving $\{\sigma_i\}_{i\in[n]}$ or $\{\sigma_i^2\}_{i\in[n]}$, and we denote their corresponding optimal estimators as $\hat{\sigma}^{*}$ and $\hat{\sigma}^{2*}$.

Finally, even though our interest in $\hat{\mu}^*$ stems from the decision problem \eqref{eq:oracle.est}, there are other reasons why $\hat{\mu}^*$, and thus our proposed estimator, may be of interest. One such setting is a decision problem where we wish to identify units with $\mu_i$ below a certain threshold $c$ under an absolute loss function. 
\begin{lemma}[Identifying units below a fixed threshold]
	\label{lm:threshold.detect}
	Suppose that Assumption~\ref{ass:disagg.model} holds. 
	Let $c$ be a fixed constant and define 
	\begin{equation}
		l_{\text{d}}(\hat{\mu},\mu,c) 
		:=
		\frac{1}{n}\sum_{i=1}^n
		\bigg(
		\left\{\hat{\mu}_i \leq c\right\}
		\left\{\mu_i > c\right\}
		+
		\left\{\hat{\mu}_i > c\right\}
		\left\{\mu_i \leq c\right\}
		\bigg)
		\cdot 
		\lvert \mu_i-c \rvert,
		\label{eq:loss.detect}
	\end{equation}
	and $R_{G_0,\text{d}}(\hat{\mu},\mu,c) := \mathbb{E}_{G_0}^{(\mu,\sigma)}
	\mathbb{E}_{(\mu,\sigma)}^{\mathcal{Y}} l_{\text{d}}(\hat{\mu},\mu,c) $.
	Then $\hat{\mu}^*$ 
	solves the following problem: 
	\begin{equation}
		\argmin_{\hat{\mu}}
		R_{G_0,\text{d}}(\hat{\mu},\mu,c)
		\label{eq:oracle.est.aux}
	\end{equation}
	where the minimization is over all borel mappings of the data $\mathcal{Y}$ to $\mathbb{R}^n$.
	Similarly, $\hat{q}_{\alpha}^*$ minimizes $R_{G_0,\text{d}}(\hat{q}_{\alpha},q_\alpha,c)$ over all borel mappings $\hat{q}_{\alpha}:\mathcal{Y}\mapsto\mathbb{R}^n$. 
\end{lemma}
Thus $\hat{\mu}^*$ is also of interest for identifying teachers with poor performance relative to an absolute threshold (e.g., mean value-added below a fixed value), where the cost of misidentification is symmetric and increasing in the absolute loss. This complements existing results such as \cite{Gu2023}, which identify units within a relative threshold (say the bottom $10\%$), and \cite{Kline2024}, which assign ranks to units under alternative loss functions.

    \section{Theoretical Results}
\label{sec:theory}
\subsection{Tweedie Representation of Optimal Estimators}
\label{subsec:tweed}

The optimal estimators are the Bayes estimators $\mathbb{E}_{G_0}[\,\cdot\, | y_i,s_i]$ taking $G_0$ as the prior. These are oracles in the sense that $G_0$ is unknown in practice. As a first step to studying the regrets of the feasible versions, we prove a Tweedie-type representation of the oracles. Its proof is in Appendix~\ref{subsec:app.tweed}.
\begin{theorem}[Tweedie's formula under unknown heteroskedasticity]
	\label{thm:main}
	Suppose that Assumption~\ref{ass:disagg.model} holds and $J>3$. Then,
	\begin{alignat*}{2}
		(i)&\qquad
		\mathbb{E}_{G_0} [\mu_i \mid y_i,s_i] 
		&&=
		y 
		+ 
		\frac{\mathbb{E}_{G_0} [\sigma_i^2 \mid y_i,s_i]}{J} \cdot \partial_y \log f_{G_0}(y_i \mid s_i) 
		+ 
		\frac{1}{J}\partial_y \mathbb{E}_{G_0} [\sigma_i^2 \mid y_i,s_i],
		%\label{eq:Bayesmu}
		\\
		(ii)&\qquad
		\mathbb{E}_{G_0} [\sigma_i^2 \mid y_i,s_i]
		&&=
		k\cdot
		\int_{s_i^2}^\infty
		\left[\frac{s_i^2}{t}\right]^{k-1}\frac{f_{G_0}(t \mid y_i)}{f_{G_0}(s_i^2 \mid y_i)}dt,
		%\label{eq:Bayessigmasq}
		\\
		(iii)&\qquad
		\mathbb{E}_{G_0} [\sigma_i \mid y_i,s_i]
		&&=
		\frac{k}{\Gamma(0.5)}\cdot
		\int_{s_i^2}^\infty
		\frac{1}{\sqrt{k(t-s_i^2)}}
		\left[\frac{s_i^2}{t}\right]^{k-1}\frac{f_{G_0}(t \mid y_i)}{f_{G_0}(s_i^2 \mid y_i)}dt,
		%\label{eq:Bayessigma}
	\end{alignat*}
	where $\Gamma(\cdot)$ is the gamma function, $k:=\frac{1}{2}(J-1)$, and $(y_i,s_i^2)$ are the sufficient statistics for $(\mu_i,\sigma_i^2)$, and $f_{G_0}(y,s^2)$ is the mixture density of $(y_i,s_i^2)$ under $G_0$.
\end{theorem}
\begin{remark}[Proof Outline]
	\cite{Efron2011} provides a formula, originally due to \cite{Robbins1956} for deriving expressions similar to Theorem~\ref{thm:main} when the sampling model lies within the exponential family. The derivation essentially recognizes the posterior density as a member of the exponential family, and moments of the sufficient statistic(s) can then be obtained by repeatedly differentiating the cumulant generating function.
	The extension to the heteroskedastic case is non-trivial, because the sufficient statistics are sufficient for $(\mu_i,\mu_i/\sigma_i^2)$ and not $\sigma_i$ nor $\sigma_i^2$.
	To handle this, we take a more direct route to the first expression (see Appendix~\ref{subsec:app.tweed}) by differentiating under the integral and applying the law of iterated expectations, and the latter two expressions are obtained because by \cite{Cressie1986}, the moment generating function (MGF) contains information on the fractional moments -- and the posterior MGF of $\sigma_i^{-2}$ is known by following \cite{Efron2011} described above.
	\blackqed
\end{remark}
Theorem~\ref{thm:main} provides representations of the oracle estimators that depend on $G_0$ only through $f_{G_0}(y,s^2)$, the mixture density of sufficient statistics. 
In particular, (i) is an extension of the well-known Tweedies formula
%  in the homoskedastic or known variance case:
\begin{equation}
	\mathbb{E}_{G_0}[\mu_i \mid y_i,\sigma_i]
	=
	y_i + \frac{\sigma_i^2}{J}\partial_y \log f_{G_0}(y_i \mid \sigma_i).
	\label{eq:tweed.hom.2} 
\end{equation}
where now the unknown $\sigma_i^2$ is replaced by its Bayes estimator together with an additional Bayes correction term $\frac{1}{J}\partial_y \mathbb{E}_{G_0}[\sigma_i^2 \mid y_i,s_i]$. 
That the oracles depend on $f_{G_0}$ rather than $G_0$ \textit{per se} is important: the deconvolution error is far less consequential for learning the super-smooth $f_{G_0}$ than for learning $G_0$, which has convergence rate as slow as $\frac{1}{\log n}$ for super-smooth mixtures (see \cite{Fan1991}). This insight, and the particular forms of the representations, underlie the near-parametric regret convergence rates subsequently established for an estimator $f_{\hat{G}}$.

% The main insight of Theorem~\ref{thm:main} and Lemma~\ref{lm:main.quant} is that the Bayes estimators depend on $G_0$ only through $f_{G_0}$. As a result, emulating the oracles requires only estimating $f_{G_0}$ and not $G_0$. This avoids the deconvolution step which can be ill-posed in practice. The representations are also central for expressing the regret of using a feasible $f_{\hat{G}}$ in terms of the Hellinger distance between $f_{\hat{G}}$ and $f_{G_0}$ -- super-smooth densities in our model -- which is then key to our results on the near-parametric regret convergence rates.  

A separate application of Theorem~\ref{thm:main} is in the $f$-modeling approach, where the representations motivate directly estimating $f_{G_0}$ with a nonparametric estimator $\hat{f}$. This is computationally simple but is not guaranteed properties of the oracles, such as monotonicity, because it ignores the structure of $f_{G_0}$ as a mixture density.
However, the simplicity of this approach may become attractive in more complex settings, such as covariate-assisted estimation discussed in Section~\ref{sec:ext}.

The representation of $\mathbb{E}_{G_0} [\sigma_i \mid y_i,s_i]$ in Theorem~\ref{thm:main} is also relevant for the unit-specific quantiles \eqref{eq:quantiles}, which in our model (Assumption~\ref{ass:disagg.model}) take the simple form $q_{\alpha,i}= \mu_i + \sigma_i\cdot \Phi^{-1}(\alpha)$, with $\Phi(\cdot)$ the standard normal CDF -- which is linear in $(\mu_i,\sigma_i)$.
Theorem~\ref{thm:main} then provides us with a Tweedie-type representation of the optimal estimator of $\{q_{\alpha,i}\}_{i\in[n]}$,
\begin{equation}
	\hat{q}_\alpha^*
	:=
	\argmin_{\hat{q}} R_{G_0}(\hat{q},q_\alpha),
	\label{eq:oracle.est.q}
\end{equation}
where the minimization is over all borel mappings of the observations $\mathcal{Y}$ to $\mathbb{R}^n$.
\begin{lemma}[Optimal Quantile Estimators]
	\label{lm:main.quant}
	Suppose Assumption~\ref{ass:disagg.model} hold and $J>3$. Let $\alpha\in(0,1)$. Then, 
	\begin{align*}
		\hat{q}_{\alpha,i}^*
		&= 
		\mathbb{E}_{G_0} [\mu_i \mid y_i,s_i]
		+
		\mathbb{E}_{G_0} [\sigma_i \mid y_i,s_i] \cdot \Phi^{-1}(\alpha)
	\end{align*}
	\label{eq:compoundopt.quant}
	which depends on $G_0$ only through the density $f_{G_0}(y,s^2)$.
\end{lemma}

\begin{comment}
	\begin{remark}[Alternative Tweedie Expressions]
		An expression for $\mathbb{E}_{G_0}[\sigma^2 \mid s]$ -- note the integrating out of $y$ -- is given in Proposition 3 of \cite{Gu2017a} credited to \cite{Robbins1982}, and is not in general equivalent to $\mathbb{E}_{G_0}[\sigma^2 \mid y,s]$ of Theorem~\ref{thm:main}.
		Similar expressions of optimal estimators for $\sigma^2$ under alternative loss functions are given in \cite{Kwon2023}, which were then exploited for constructing empirical Bayes estimates of $\sigma^2$.
		\blackqed
	\end{remark}
\end{comment}

\subsection{Regret Convergence Rates}
\label{subsec:rates}

The oracles' representations from the previous section are key to establishing the fast regret convergence results of their feasible counterparts. In particular, we consider $(\hat{\mu},\hat{\sigma},\hat{\sigma}^2,\hat{q}_{\alpha})$, taken to denote feasible versions of the oracles $(\hat{\mu}^*,\hat{\sigma}^*,\hat{\sigma}^{*2},\hat{q}_{\alpha}^*)$, where $G_0$ is replaced by $\hat{G}_n$ of Assumption~\ref{ass:npmle} below.  
In preparation for what follows, we clarify notations and state the assumptions underlying the results. 
The notation $G(A \times B)$ for a  distribution $G$ and sets $A,B\subset\mathbb{R}$ denotes the probability that $G$ places on the event $(\mu\in A,\sigma\in B)$.
The notation $X\lesssim Y$ means there exists a constant $C$, possibly dependent on the constants of Assumption~\ref{ass:subexp}, such that $X \leq C Y $.

The first assumption is on $G_0$, the distribution of latent effects $\{(\mu_i,\sigma_i)\}_{i\in[n]}$.
\begin{assumption}[Class of Distributions]
	\label{ass:subexp}
	$G_0\in\mathcal{G}:=\{ G: G(\mathbb{R} \times [\underline{\sigma},\infty)) =1 \} $
	where the marginals of every $G\in\mathcal{G}$ have 
	sub-exponential tails: for every $G\in\mathcal{G}$ and $t > 0$,
	\begin{equation}
		P_{G}(|\mu| > t) \leq C_1 \exp\{-\lambda_1 t^{1/\gamma_1}\},
		\quad
		P_{G}(\sigma > t) \leq C_2 \exp\{-\lambda_2 t^{1/\gamma_2}\}
		\label{eq:subexp.tails}
	\end{equation}
	where $C_1, C_2, \lambda_1, \lambda_2 , \gamma_1 , \gamma_2 > 0$ are constants for the class $\mathcal{G}$.
\end{assumption}
Assumption~\ref{ass:subexp} accommodates distributions with sub-exponential tails such as Gaussian mixtures for $\mu_i$ and Gamma mixtures (suitably truncated away from 0) for $\sigma_i^2$.
For example with $\gamma_1=0.5$, we obtain distributions for $\{\mu_i\}_{i\in[n]}$ with sub-Gaussian tails. 
The restrictions are on the tails of the marginals and no restrictions are imposed on the dependence within $(\mu,\sigma)$.
The lower-bound $\underline{\sigma}$ on the support of $\sigma$ is crucial for our proofs, which is used to guarantee sufficient smoothness of the mixture density $f_{G_{0}}$.

\begin{assumption}[Approximate MLE]
	\label{ass:npmle}
	$\hat{G}_n$ satisfies
	\begin{equation}
		\prod_{i=1}^n f_{\hat{G}_n}(y_i,s_i^2)
		\geq 
		\sup_{D\in\mathscr{D}}\prod_{i=1}^n f_{G}(y_i,s_i^2) - \eta_n
	\end{equation}
	where $\mathscr{D}$ is the set of all distributions, and $\eta_n \asymp \tfrac{1}{n}$. Furthermore $\hat{G}_n$ also satisfies 
	\begin{equation}
		\hat{G}_n\left(\left[-\max_{i\leq n} |y_i| \,,\, \max_{i\leq n} |y_i|\right] \times \left[\underline{\sigma} \,,\, \max_{i\leq n} s_i\right]\right)=1.
	\end{equation}
%	Furthermore, $\hat{G^{(i)}}$ satisfies
%	\begin{equation}
%		\prod_{j\neq i}^n f_{\hat{G}_n^{(i)}}(y_j,s_j^2)
%		\geq 
%		\sup_{G}\prod_{j\neq i}^n f_{G}(y_j,s_j^2)
%		- \eta_n
%	\end{equation}
%	where $\eta_n \asymp \tfrac{1}{n}$ and the supremum is over all bivariate distributions on $\mathbb{R}\times\mathbb{R}_{++}$.
\end{assumption}
Assumption~\ref{ass:npmle} requires $\hat{G}_n$ to be an approximate MLE with support within the boundaries of the data.
\begin{remark}[Implementation]
	\label{rmk:implementation}
	Under a sampling density assumption (i.e. Assumption~\ref{ass:disagg.model}), $\hat{G}_n$ is typically computed using NPMLE, for instance with $n$ uniformly spaced support points within the boundaries of the data and then casting the problem as a convex optimization problem; see \cite{Koenker2014}. We use EM algorithm to estimate the sieve MLE $\hat{D}(k_n)\in\mathscr{D}(k_n)$ with $k_n \asymp \log n$: 
	\begin{equation}
		\prod_{i=1}^n f_{\hat{D}(k_n)}(y_i,s_i^2)
		\geq 
		\sup_{D \in \mathscr{D}(k_n)}\prod_{i=1}^n f_{G}(y_i,s_i^2),
	\end{equation}
	where $\mathscr{D}(k_n)$ is the set of distributions with at most $k_n$ support points. 
	We find little difference between these two implementations within the numerical experiments\footnote{Note the difference between our EM algorithm and that from, for instance \cite{Jiang2009}, which uses $n$ uniformly spaced support points: we use only $\log n$ support points that are adaptively selected.}. 
	\blackqed
\end{remark}

We now state the regret convergence results. The proofs are in Appendix~\ref{sec:app.reg}.
\begin{theorem}[Regret Bounds]
	\label{thm:eb.opt}
	Suppose Assumptions \ref{ass:disagg.model}, \ref{ass:subexp} and \ref{ass:npmle} hold, and $J>5$.
	% For any $\gamma_1, \gamma_2 > 1$, define
	% .
	Then
	\begin{align*}
		\sup_{G_0\in\mathcal{G}} \left[R_{G_0}(\hat{\mu},\mu) - R_{G_0}(\hat{\mu}^{*},\mu)\right]
		&\,\lesssim\,
		\tfrac{1}{n} \cdot (\log n)^{6\tilde{\gamma}_1+8\tilde{\gamma}_2+2}
		\\
		\sup_{G_0\in\mathcal{G}} \left[R_{G_0}(\hat{\sigma},\sigma) - R_{G_0}(\hat{\sigma}^{*},\sigma)\right]
		&\,\lesssim\,
		\tfrac{1}{n} \cdot (\log n)^{4\tilde{\gamma}_1+8\tilde{\gamma}_2+1}
		\\
		\sup_{G_0\in\mathcal{G}} \left[R_{G_0}(\hat{\sigma}^2,\sigma^2) - R_{G_0}(\hat{\sigma}^{2*},\sigma^2)\right]
		&\,\lesssim\,
		\tfrac{1}{n} \cdot (\log n)^{4\tilde{\gamma}_1+8\tilde{\gamma}_2+1}
	\end{align*}
	where $\tilde{\gamma}_1 := \gamma_1 \vee [\gamma_2 + \tfrac{1}{2}]$ and $\tilde{\gamma}_2 := \gamma_2 + \tfrac{1}{2}$.
\end{theorem}

\begin{remark}[Proof Outline]
	\label{rmk:proof.outline}
	We sketch the argument for $\{\mu_i\}_{i\in[n]}$. The choice of Bayes risk and the iid assumption across $i$ reduce the regret to
	\begin{equation}
		\mathbb{E}_{G_{0}}^{\mathcal{Y}} 
		\left[
		\hat{\mu}_{\hat{G}_n}(y_1,s_1^2 ) 
		-  
		\hat{\mu}_{G_{0}}(y_1,s_1^2 ) 
		\right]^2,
		\label{eq:eb.opt.outline.1}
	\end{equation}
	which is within $o(1/n)$ of
		\begin{equation}
		\mathbb{E}_{G_{0,n}}^{\mathcal{Y}} 
		\left[
		\hat{\mu}_{\hat{G}_n}(y_1,s_1^2 ) 
		-  
		\hat{\mu}_{G_{0,n}}(y_1,s_1^2 ) 
		\right]^2,
		\label{eq:eb.opt.outline.2}
	\end{equation}
	where $G_{0,n}$ truncates $G_{0}$ to a logarithmically growing support. Conditioning on the event $A_n$ that $f_{\hat{G}_n}$ is within $\varepsilon_n$ of $f_{G_{0,n}}$ in Hellinger distance (we later show $\mathbb{P}_{G_{0,n}}[A_n^c] = o(\tfrac{1}{n})$), \eqref{eq:eb.opt.outline.2} is in turn within $\delta_n\asymp\frac{1}{n}$ of
	\begin{equation}
		\max_{H_n\in\mathcal{H}_n} \mathbb{E}_{G_{0}}^{\mathcal{Y}_1} 
		\left[
		\hat{\mu}_{H_n}(y_1,s_1^2 ) 
		-  
		\hat{\mu}_{G_{0,n}}(y_1,s_1^2 ) 
		\right]^2,
		\label{eq:eb.opt.outline.3}
	\end{equation}
	where $\mathcal{H}_n\subseteq\mathcal{G}_n$ ($\mathcal{G}_n$ being the subset of $\mathcal{G}$ with logarithmically growing supports) is a deterministic class of densities within $\varepsilon_n$ of $f_{G_{0,n}}$ in Hellinger distance. Theorem~\ref{thm:main}, which extends Proposition~3 of \cite{Jiang2009} to unknown heteroskedasticity, is central in reducing \eqref{eq:eb.opt.outline.3} to the Hellinger distance between $f_{H_n}$ and $f_{G_{0,n}}$ which is bounded by $\varepsilon_n$. The main difficulty in the extension is relating the integrals in $\hat{\sigma}^{2*}$ and $\hat{\sigma}^2$ to this distance. Because $\{f_{G}:G\in\mathcal{G}_n\}$ is super-smooth, we obtain $\mathbb{P}_{G_{0,n}}[A_n^c] = o(\tfrac{1}{n})$ for $\varepsilon_n\asymp n^{-1/2}$ up to logarithmic factors. Finally, $\tilde{\gamma}_1$ and $\tilde{\gamma}_2$ appear instead of $\gamma_1$ and $\gamma_2$ because we do not assume knowledge of the latter and instead rely on Assumption~\ref{ass:npmle} to restrict the support of the MLE.
	\blackqed
\end{remark}

Theorem~\ref{thm:eb.opt} provides regret bounds for the estimation of $\{\mu_i\}_{i\in[n]}$, $\{\sigma_i\}_{i\in[n]}$ and $\{\sigma_i^2\}_{i\in[n]}$, which all converge at the parametric rate (up to logarithmic factors). 
This suggests that the cost of taking a nonparametric approach will be small in finite-samples. 
In contrast, the benefits in terms of risk reduction can be large relative to assuming a particular parametric form for $G_{0,n}$, as documented in \cite{Gilraine2020}. 
Similar results are well-known in the cases of homoskedasticity or known heteroskedasticity, and Theorem~\ref{thm:eb.opt} can thus be seen as an extension to the setting where the heteroskedasticity is unknown and has to be estimated. 

\begin{comment}
	\begin{remark}[Rate Comparisons]
		With $\gamma_1=\gamma_2=0$, such that $\mathcal{G}$ has bounded support, the regret convergence rate for $\{\hat{\mu}_G(y_i,s_i)\}_{i\in[n]}$ reduces to $\frac{1}{n}(\log n)^{8.5}$. This is in contrast to $\frac{1}{n}(\log n)^{5}$ of Theorem 5 of \cite{Jiang2009}, where the larger logarithmic factor here arises from the second half of Assumption~\ref{ass:npmle} to accommodate the unknown heteroskedasticity.  
		\blackqed
	\end{remark}
\end{comment}

Theorem~\ref{thm:eb.opt} also extends ratio-optimality type results, established by \cite{Jiang2009} for homoskedastic models\footnote{
	These results were similarly established by \cite{Brown2009} and \cite{Liu2020} for $f$-modeling estimators.
},
to accommodate unknown heteroskedasticity and more general estimation objectives. For example, supposing
\begin{equation}
	\inf_{G_0\in\mathcal{G}}
	\frac{
		n R_{G_{0}}(\hat{\mu}^{*},\mu)
	}{
 		(\log n)^{6\tilde{\gamma}_1+8\tilde{\gamma}_2+\frac{3}{2}}
	}
	\rightarrow \infty
\end{equation}
such that the oracle risk is not too small, then Theorem~\ref{thm:eb.opt} implies the following ratio-optimality result:
\begin{equation}
	\sup_{G_0\in\mathcal{G}}
	\left[\frac{
		R_{G_{0}}(\hat{\mu},\mu) 
	}{
		R_{G_{0}}(\hat{\mu}^{*},\mu)
	}\right]
	\leq
	1 + o(1).
	\label{eq:ratio.opt}
\end{equation}
Ratio-optimality, which targets relative risk, is a much stronger notion than the usual empirical Bayes optimality, which only requires that 
$
	R_{G_{0}}(\hat{\mu},\mu)  = 
	R_{G_{0}}(\hat{\mu}^{*},\mu) + o(1)
$.
 
Finally recall that by Lemma~\ref{lm:main.quant}$, \hat{q}_{\alpha,i}^*$ (i.e., the oracle estimator of $q_{\alpha,i}$) is a linear function of $\hat{\mu}_i^*$ and $\hat{\sigma}_i^*$. Theorem~\ref{thm:eb.opt} therefore also provides bounds on the regret of using $\{\hat{q}_{\alpha,i}\}_{i\in[n]}$ to estimate $\{q_{\alpha,i}\}_{i\in[n]}$.
\begin{lemma}[Regret Bounds for $\hat{q}$]
	Suppose that Assumptions \ref{ass:disagg.model}, \ref{ass:subexp} and \ref{ass:npmle} hold, and $J>5$. Then for any fixed $\alpha\in(0,1)$
	\begin{equation*}
		\sup_{G_0\in\mathcal G}\left[R_{G_{0}}(\hat{q}_\alpha,q_\alpha) - R_{G_{0}}(\hat{q}_\alpha^{*},q_\alpha)\right]
		\lesssim 
		\tfrac{1}{n} \cdot (\log n)^{6\tilde{\gamma}_1+8\tilde{\gamma}_2+2}.
		\\
	\end{equation*}
\end{lemma}

\section{Extensions}
\label{sec:ext}

The regret convergence rates of Theorem~\ref{thm:eb.opt} demonstrate that feasible versions of the oracles incur little cost in terms of regret even under unknown heteroskedasticity. 
This section considers three extensions to the model of Assumption~\ref{ass:disagg.model}, and we focus on the estimation of $\left\{\mu_i\right\}_{i\in[n]}$ to keep the discussion succinct.

\begin{extension}[Heterogeneous sample sizes]
	The first extension consists of settings where the number of observations per unit, $J_i\in \mathbb{N}$, is heterogeneous and furthermore informative about the underlying parameters $(\mu_i,\sigma_i)$. For instance, teachers have differing class sizes and a teacher with higher value-added $\mu_i$ may be systemically assigned smaller classes. Failure to take this into account will result in under- or over-estimation of $\left\{\mu_i\right\}_{i\in[n]}$; see \cite{Chen2026} for fuller arguments and especially in the setting of estimating neighborhood effects. 
	In this case, we augment the model to be 
	\begin{equation}
		y_{ij} \mid \mu_i,\sigma_i,J_i\sim \mathcal{N}\left[\mu_i,\sigma_i^2\right]
		\label{eq:disagg.model.hetj}
	\end{equation}
	for $j=1,\dots,J_i$ and $(\mu_i,\sigma_i,J_i) \sim_{\text{iid}} H_0$ for $i=1,\dots,n$, and modify the decision problem to 
	\begin{equation}
		\check{\mu}^*
		:=
		\argmin_{ \check{\mu} }
		\mathbb{E}_{H_0}^{(\mu,\sigma)}
		\mathbb{E}_{(\mu,\sigma)}^{(\mathcal{Y},\mathcal{J})}
		\frac{1}{n}\sum_{i=1}^n
		\left( \check{\mu} - \mu_i \right)^2
	\end{equation}
	%where we integrate out $\mathcal{J}:=\{J_i\}_{i\in[n]}$ in addition to the observations $\mathcal{Y}$, and 
	where $\mathcal{J}:=\{J_i\}_{i\in[n]}$ and the minimization is over all borel mappings of $(\mathcal{Y},\mathcal{J})$ to $\mathbb{R}^n$. 
	The oracle $\check{\mu}^*$ takes a similar form as Theorem~\ref{thm:main} but conditioning on $J_i$ everywhere:
	\begin{lemma}
		\label{lem:main.hetj}
		Under \eqref{eq:disagg.model.hetj} where $(\mu_i,\sigma_i,J_i) \sim_{\text{iid}} H_0$ and $\mathbb{P}_{H_0}\{J_i > 3\} = 1$, we have 
		\begin{equation*}
			\check{\mu}_i^* 
			=
			y_i 
			+ 
			\frac{\mathbb{E}_{H_0}\left[\sigma^2\mid y_i,s_i,J_i\right]}{J_i}
			\partial_y \log f_{H_0}\left(y_i,s_i^2 \mid J_i\right)
			+
			\frac{1}{J_i}
			\partial_y \mathbb{E}_{H_0}\left[\sigma^2\mid y_i,s_i,J_i\right].
		\end{equation*}
	\end{lemma}
	Thus, when the heterogeneity in the discrete $J_i$ is limited such that there are sufficient observations for each value $J$, one may pool observations of teachers with identical $J_i$ to estimate the density $f_{H_0}(y,s | J_i=J)$ for each $J$. The fast regret convergence rates of Theorem~\ref{thm:eb.opt} suggest that this will work well in practice. 
	\blackqed
\end{extension}
\begin{extension}[Covariate-assisted estimation]
	The discussion above applies equally to the setting where the researcher has access to covariate(s) $\bs{x}_i\in\mathbb{R}^m$ that is potentially informative about the parameters $(\mu_i,\sigma_i)$ and that will help sharpen estimation. For instance in the teacher value-added literature, the years of experience and education level of teacher $i$ may be positively correlated with her value-added $\mu_i$. In such settings, by suitably modifying the model and decision problem as above, we can show a Tweedie-representation: the oracles depend on $G_0$, now the distribution of $(\mu_i,\sigma_i,\bs{x}_i)$, only through the conditional density $f_{G_0}(y,s^2 \mid \bs{x})$. 

	With the latent distribution now involving $\bs{x}_i$, taking the $g$-modeling approach (i.e. plugging in an estimate of $G_0( \cdot \mid \bs{x}_i)$ into $f_{G_0}(y,s^2 \mid \bs{x})$ for each of $\{\bs{x}_i\}_{i\in[n]}$) can be difficult to implement. In contrast, the Tweedie representations imply that $f$-modeling estimators are straightforward, where off-the-shelf estimators can be plugged in for the unknown $f(y,s^2 \mid \bs{x})$, with existing results on density estimation available to  combine with Theorem~\ref{thm:eb.opt} (see proof logic in Remark~\ref{rmk:proof.outline}) for a regret analysis. \blackqed
\end{extension}
\begin{extension}[Dynamic panel forecasting]
	Consider the problem of constructing optimal forecasts for a large collection of short time series, taking the basic dynamic panel model studied in \cite{Liu2020} as an example:
	\begin{equation}
		y_{it} = \rho \cdot y_{it-1} + \mu_i + \sigma \cdot \epsilon_{it}, 
		\quad \epsilon_{it}\mid \mu_i,y_{i,0} \sim \mathcal{N}[0,1],
		\quad (\mu_i,y_{i,0})\sim G_0
	\end{equation} 
	for $i=1,\dots,n$ and $t=1,\dots,T$.
	The motivating application is regulatory stress testing of bank-holding companies, where one must first forecast bank balance-sheet variables under given macroeconomic conditions; as detailed in \cite{Liu2020}, frequent mergers and post-2008 regulatory changes make a large-$n$ small-$T$ framework suitable here.
	In such a model, the optimal one-step-ahead forecast is 
	$\mathbb{E}_{G_0}[y_{it} \mid \mathcal{Y}_{i,t-1}] = \rho y_{it-1} + \mathbb{E}_{G_0}[\mu_i \mid \mathcal{Y}_{i,t-1}]$, where $\mathcal{Y}_{i,t-1}$ is the vector of observations up to time period $t-1$ associated with unit $i$. 
	The conditional mean $\mathbb{E}_{G_0}[\mu_i \mid \mathcal{Y}_{i,t-1}]$ takes the form \eqref{eq:tweed.hom.2} with $f_{G_0}(y_i | \sigma)$ replaced by $f_{G_0}(y_i| \sigma,y_{i,0})$ where $y_i$ is the sufficient statistic for $\mu_i$, which reveals that the oracle forecasts' reliability hinges upon the homoskedasticity assumption, because $\sigma^2$ regulates the Bayes correction in estimating each and every $\mu_i$. 
	
	Within this context, our model adds a layer of heterogeneity through the unit-specific $\sigma_i$, replacing the common $\sigma$ above with $\sigma_i$ and letting $(\mu_i,\sigma_i,y_{i,0})\sim G_0$.
	Indeed it is reasonable to expect in this application that the heterogeneity in $\{\sigma_i^2\}_{i\in[n]}$ to be comparable to that within $\{\mu_i\}_{i\in[n]}$. 
	As a result, the oracle forecasts that allow for unknown heteroskedasticity can substantially improve over those that impose homoskedasticity by allowing for unit-specific Bayes correction.  
	In such a model, emulating the oracles again reduces to emulating $\mathbb{E}_{G_0}[\mu_i \mid \mathcal{Y}_{i,t-1}]$, which now takes the form in Theorem~\ref{thm:main} with $f_{G_0}(y_i,s_i^2)$ replaced by $f_{G_0}(y_i,s_i^2| y_{i,0})$.
	If one is willing to accept independence between $y_{i,0}$ and $(\mu_i,\sigma_i^2)$, then Section~\ref{sec:theory} delivers fast regret\footnote{Regret that is relative to the oracle forecasts, i.e., that knows the distribution $G_0$.} convergence rates of using feasible forecasts that plug in approximate MLEs of the density of $(\mu_i,\sigma_i^2)$. 
	In the case of unrestricted dependence between $y_{i,0}$ and $(\mu_i,\sigma_i^2)$, the results developed here, leveraging on the expressions derived in Theorem~\ref{thm:main}, may be extended under suitable smoothness conditions on the density of $y_i,s_i^2 \mid y_{i,0}$.  
	\blackqed
\end{extension}

    \section{Misspecification of the Sampling Model}
\label{sec:missp}

The existing optimal estimation results under homoskedasticity, as well as those newly established in Theorem~\ref{thm:eb.opt} for unknown heteroskedasticity, rely on an assumption of normal sampling model;
see Assumptions~\ref{ass:agg.model}~and~\ref{ass:disagg.model} respectively. 
In particular the fast regret convergence rates are due to the optimal (i.e. Bayes) estimators depending on $G_0$ only through the super-smooth mixture density, and in a way where the regret of the feasible versions can be quantified in terms of the Hellinger distance between the true and estimated densities. When this assumption is inexact, i.e. maintaining \eqref{eq:disagg.model} but where neither $p_\epsilon$ nor $p_z$ is necessarily normal, then Bayes estimators derived under a normal sampling model are no longer the true Bayes estimators, and these optimal estimation results no longer hold. 

This section quantifies the error from such misspecification, specifically\footnote{Similar results can be derived for the estimation of $\{\sigma_i\}_{i\in[n]}$ or $\{\sigma^2_i\}_{i\in[n]}$, though not for $\{q_{\alpha,i}\}_{i\in[n]}$ which relies on the linearity of the $q_{\alpha,i}$ in $\mu_i$ and $\sigma_i$ in a Gaussian sampling model. This suggests that the optimality results for $\{q_{\alpha,i}\}_{i\in[n]}$ are more fragile to deviations from microdata normality. } for the estimation of $\{\mu_i\}_{i\in[n]}$ which is the most common objective. 
In a nutshell, provided the Gaussian mixture density approximates the true one well enough, this error is of order $O(\frac{1}{J^2})$ and $O(\frac{1}{J})$ in the homoskedastic and heteroskedastic settings respectively (improving to $O(\frac{1}{J^3})$ and $O(\frac{1}{J^2})$ when the microdata is symmetric). The results of Theorem~\ref{thm:eb.opt} are thus more sensitive to non-normality than those under homoskedasticity; but since both errors scale with $\frac{1}{J}$, nonparametric empirical Bayes remains justified for moderate $J$ even under mild misspecification.
% For convenience, we restate the model \eqref{eq:disagg.model}:
% \begin{equation}
% 	y_{ij} = \mu_i + \sigma_i \cdot \epsilon_{ij}, \quad \epsilon_{ij} \mid \mu_i,\sigma_i\sim_{\text{iid}} p_\epsilon
% 	\quad\text{ and }\quad
% 	(\mu_i,\sigma_i) \sim_{\text{iid}} G_0
% 	\label{eq:missp.model.0}
% \end{equation}
% for $j=1,\dots,J$ and $i=1,\dots,n$, and where $p_\epsilon$ has mean 0 and variance 1 is \emph{not} necessarily normal.  
Because the Bayes estimators are different in the homoskedastic versus the heteroskedastic settings, we split the analysis accordingly. 
This section's proofs are located in Appendix~\ref{subsec:app.missp}.
\subsection{The Homoskedastic Setting}

We first look at the error from misspecification in the homoskedastic setting, where $\sigma_i=\sigma_0$ a known constant. Then $G_0$ collapses to a univariate distribution on $\mu_i$. Let $f_p(y | \mu_i)$ denote the sampling density of $y_i$, where the subscript ``$p$'' emphasizes its dependence on $p_\epsilon$.
The model \eqref{eq:disagg.model} implies a distribution on $(y_i,\mu_i)$ as
\begin{equation}
	f_{p}( y \mid \mu_i) = \tfrac{\sqrt{J}}{\sigma_0} \cdot p_{z}\left(z_i\right) 
	\quad\text{ and }\quad \mu_i \sim_{\text{iid}} G_0
	\quad\text{ for }\quad i=1,\dots,n,
	\label{eq:missp.model.1}
\end{equation}
where $z_i:= \frac{y-\mu_i}{\sigma_0/\sqrt{J}}$. 
% and $p_{z}$ is the density of $z_i=\frac{1}{\sqrt{J}}\sum_{j=1}^J \epsilon_{ij}$.
The Bayes estimator of $\mu_i$ is $\mathbb{E}_{p,G_0}[\mu_i\mid y_i]$ with the subscript ``$p,G_0$'' denoting an expectation taken under \eqref{eq:missp.model.1}. 

The misspecified estimator we study is $\mathbb{E}_{\phi,B_0}[\mu \mid y]$, where the subscript ``$\phi,B_0$'' denote an expectation taken under the normal sampling model (let $f_\phi(y|\mu_i)$ denote the sampling density of $y_i$ under normality):
\begin{equation}
	f_\phi(y \mid \mu_i) = p^\phi(z_i) := \tfrac{\sqrt{J}}{\sigma_0}  \phi(z_i)
	\quad\text{ and }\quad \mu_i \sim_{\text{iid}} B_0
	\quad\text{ for }\quad i=1,\dots,n,
	\label{eq:missp.model.2}
\end{equation}
where $B_0$ is a distribution on $\mathbb{R}$.
A finite-sample perspective of this, which corresponds to how empirical Bayes estimators are implemented in practice, is where the $G_0$ is estimated under a normal sampling model (e.g., via NPMLE)
% the mixture density of the sufficient statistic $y_i$ is estimated using NPMLE\footnote{Equivalently, a finite (homoskedastic) Gaussian mixture model.}, and the implied prior for $\mu_i$ (call it $\hat{B}$) is 
and plugged into the posterior mean to form the feasible estimator. This is precisely $\mathbb{E}_{\phi,\hat{B}}[\mu_i \mid y_i]$, the finite-sample analogue to $\mathbb{E}_{\phi,B_0}[\mu_i \mid y_i]$.

The first half of \eqref{eq:missp.model.2} is only approximately correct by the CLT. Thus $\mathbb{E}_{\phi,B_0}[\mu_i \mid y_i]$ is not the same as the optimal estimator $\mathbb{E}_{p,G_0}[\mu_i\mid y_i]$ derived under the true model \eqref{eq:missp.model.1}.
The following lemma quantifies the error from such an approximation. 
\begin{lemma}
	\label{lem:missp.hom}
	Under homoskedasticity, the misspecification error is
	\[
	\mathbb{E}_{p,G_0}[\mu_i \mid y_i] - \mathbb{E}_{\phi,B_0}[\mu_i \mid y_i] = 
	\tfrac{\sigma_0^2}{J}\nabla_0
	- 
	\tfrac{\sigma_0^2}{J}\mathbb{E}_{p,G_0}[ \Delta_0 \mid y_i ]
	\]
	where $\Delta_0 :=  \tfrac{\sqrt{J}}{\sigma_0}[\partial_z\log p_z(z_i) + z_i]$, $\nabla_0 := \partial_y \log f_{p,G_0}(y_i)  - \partial_y \log f_{\phi,B_0}(y_i)$, and $f_{p,G_0}(y_i)$ and $f_{\phi,B_0}(y)$ are the mixture densities under \eqref{eq:missp.model.1} and \eqref{eq:missp.model.2} respectively.
\end{lemma}
An application of Jensen's inequality with Lemma~\ref{lem:missp.hom} yields 
\begin{equation}
	\mathbb{E}_{p,G_0}\big[ \mathbb{E}_{p,G_0}[\mu_i \mid y_i] - \mathbb{E}_{\phi,B_0}[\mu_i \mid y_i] \big]^2 
	\leq 
	2\tfrac{\sigma_0^4}{J^2} \cdot \mathbb{E}_{p,G_0}[ \nabla_0^2+\Delta_0^2 ] 
\end{equation}
where $\Delta_0$ as defined is a measure of the discrepancy between the true and Gaussian models in terms of the sampling densities' log-derivatives. 
An Edgeworth expansion of $p_z$ yields 
\begin{equation}
	\Delta_0 = \frac{\kappa_3}{2\sigma_0}[z_i^2 - 1] + O_p(J^{-1/2})
\end{equation}
where $\kappa_3:=\mathbb{E}_{p}[\epsilon_{ij}^3]$ measures the skewness of the microdata. As a result, $\Delta_0^2 = O_p(1)$ and, subject to control of $\nabla_0$, the misspecification error in the homoskedastic setting will thus be of order $O(\frac{1}{J^2})$. This improves to $O(\frac{1}{J^3})$ if the microdata were symmetric.
% Furthermore if the microdata were symmetric so that $\mathbb{E}_p[\epsilon_{ij}^3]=0$, then the misspecification error will fall by a further order to $O(\frac{1}{J^3})$.

$\nabla_0$ is a measure of the discrepancy between the mixture densities implied by the true and the Gaussian models, and will be small if the Gaussian mixture density is able to approximate the true mixture density reasonably well. This term will be exactly zero if the following holds.
\begin{assumption}
	\label{ass:missp.gmm}
	$B_0$ satisfies
	$f_{p,G_0}(y) = f_{\phi,B_0}(y) \text{ for all } y\in\mathbb{R}$.
\end{assumption}
As an example, suppose that $G_0$ takes the form of a Gaussian mixture: $G_0=\mathcal{N}[0,v]\ast Q$ for some $v > \sigma_0^2/J$ and a distribution $Q$. Then for any $p_\epsilon$, Assumption~\ref{ass:missp.gmm} will be met when $B_0$ is the distribution of $\tfrac{\sigma_0}{\sqrt{J}}Z \ast \mathcal{N}[0,v-\tfrac{\sigma_0^2}{J}] \ast Q$ for $Z\sim \mathcal{N}[0,1]$. 
More generally, Assumption~\ref{ass:missp.gmm} is a joint restriction on $(p_\epsilon,G_0)$ and allows for a reasonably broad class of models.

\subsection{The Heteroskedastic Setting}

We next turn to the heteroskedastic setting, where we recycle notations from the homoskedastic setting to highlight objects that are parallel in the two settings.
The model \eqref{eq:disagg.model} implies a distribution for $(y_i,s_i,\mu_i,\sigma_i)$ as  
\begin{equation}
	f_{p}(y,s^2 \mid \mu_i,\sigma_i) 
	= \tfrac{J}{\sigma_i^3 \sqrt{2}} p_{z,w}\left( z_i,w_i \right) 
	\quad\text{ and }\quad (\mu_i,\sigma_i) \sim_{\text{iid}} G_0
	\text{ for $i=1,\dots,n$ }
	\label{eq:missp.model.3}
\end{equation}
where $p_{z,w}$ is the joint density of $z_i $ and $  w_i := \sqrt{\tfrac{J}{2}}\left(\tfrac{s_i^2}{\sigma_i^2}-1\right)$.
The Bayes estimator under this model is denoted as $\mathbb{E}_{p,G_0}[\mu_i\mid y_i,s_i] $. 
% As before, we assume for this section that the mixture density $f_{p,G_0}(y,s^2)$ lies in a particular class of Gaussian mixture densities. 
% \begin{assumption}
% 	\label{ass:missp.het.gmm}
% 	There exists a distribution $B_0$ on $\mathbb{R}\times\mathbb{R}_{++}$ such that
% 	$$f_{p,G_0}(y,s^2) = f_{\phi,B_0}(y,s^2) \text{ for all } (y,s^2)\in\mathbb{R}\times\mathbb{R}_{++}$$
% 	where
% 	$f_{\phi,B_0}(y,s^2)$ is the mixture density of $(y,s^2)$ taken under a Gaussian microdata sampling model and $(\mu_i,\sigma_i)\sim_{\text{iid}}B_0$; i.e, \eqref{eq:missp.model.4} with $G_0$ replaced by $B_0$.
% \end{assumption}

As before, the misspecified estimator we study is $\mathbb{E}_{\phi,B_0}[\mu_i \mid y_i,s_i]$, the Bayes estimator derived under a model with microdata normality:
\begin{equation}
	f_\phi(y,s^2 \mid \mu_i,\sigma_i) = \tfrac{J}{\sigma_i^3 \sqrt{2}} p_{z,w}^\phi(z_i,w_i)
	\quad\text{ and }\quad (\mu_i,\sigma_i) \sim_{\text{iid}} B_0
	\text{ for } i=1,\dots,n.
	\label{eq:missp.model.4}
\end{equation}
where $p_{z,w}^\phi(\cdot)$ denotes the sampling density of $(z_i,w_i)$ under microdata normality. 
Where the normality is inexact, $\mathbb{E}_{\phi,B_0}[\mu_i\mid y_i,s_i] $ is not the same as $ \mathbb{E}_{p,G_0}[\mu_i\mid y_i,s_i] $. The next lemma quantifies the error from misspecification under heteroskdasticity.
\begin{lemma} 
	\label{lem:missp.het}
	The misspecification error
	$\mathbb{E}_{\phi,B_0}[\mu_i\mid y_i,s_i]  -
		\mathbb{E}_{p,G_0} [ \mu_i \mid y_i,s_i ]$ is expressible as
	\begin{align*}
		&
		\tfrac{1}{J}\nabla_1 +
		\tfrac{1}{J}\mathbb{E}_{p,G_0}
		\bigg[
			\sigma_i^2 \Delta_1 
				- 
				\Delta_2 \cdot 
				\tfrac{\sqrt{J}}{\sigma_i} \partial_z\log p_{z|w}(z_i|w_i)
				% \underbrace{\tfrac{\sqrt{J}}{\sigma} \partial_z\log p_{z|w}(z|w)}_{\partial_y \log f_p(y| \mu,\sigma,s )}
				- 
				\tfrac{\sqrt{J}}{\sigma_i} \partial_z \Delta_2
			\,\Bigm|\, y,s
		\bigg] \\
	    \text{where} \qquad 
		&\Delta_1 := \tfrac{\sqrt{J}}{\sigma_i}[ \partial_z \log p_{z|w}(z_i\mid w_i) + z_i ] \\
		&\Delta_2 := \int_0^\infty
			e^{-t/\sigma_i^2} \frac{r(z_i,w_i)-r(z_i,w_i+\delta_t)}{r(z_i,w_i)} dt \\
		&r(z,w) := \tfrac{p_{z,w}(z,w)}{p_{z,w}^\phi(z,w)} \\ 
		&\delta_t := \tfrac{t / k}{\sigma_i^2} \sqrt{\tfrac{J}{2}} = \tfrac{t}{\sigma_i^2}\tfrac{1}{\sqrt{J}} \cdot \sqrt{2}\tfrac{J}{J-1} \\
		&\nabla_1 := \int_0^\infty \left[\tfrac{s_i^2}{s_i^2+tk^{-1}}\right]^{k-1} 
		\left[\frac{\partial_yf_{p,G_0}(y_i,s_i^2+tk^{-1})}{f_{p,G_0}(y_i,s_i^2)} -  
		\frac{\partial_yf_{\phi,B_0}(y_i,s_i^2+tk^{-1})}{f_{\phi,B_0}(y_i,s_i^2)}\right]
		dt,
	\end{align*}
	and $f_{p,G_0}(y,s^2)$ and $f_{\phi,B_0}(y,s^2)$ are the mixture densities under \eqref{eq:missp.model.3} and \eqref{eq:missp.model.4}.
\end{lemma}
Lemma~\ref{lem:missp.het} implies the bounding of the misspecification error by four different terms:
\begin{align}
	\mathbb{E}_{p,G_0}
	\left[\mathbb{E}_{\phi,B_0}[\mu_i\mid y_i,s_i]  -
		\mathbb{E}_{p,G_0} [ \mu_i \mid y_i,s_i ] \right]^2 
	\leq&
	\tfrac{1}{J^2}\mathbb{E}_{p,G_0}\left[ \sigma_i^2\Delta_1 \right]^2 
	\nonumber\\+&
	\tfrac{1}{J^2}\mathbb{E}_{p,G_0}\left[ \Delta_2 \cdot \tfrac{\sqrt{J}}{\sigma_i} \partial_z \log p_{z|w}(z_i|w_i) \right]^2  
	\nonumber\\+&
	\tfrac{1}{J^2}\mathbb{E}_{p,G_0}\left[ \tfrac{\sqrt{J}}{\sigma_i} \partial_z \Delta_2\right]^2
	\nonumber\\+&
	\tfrac{1}{J^2}\mathbb{E}_{p,G_0}\left[\nabla_1\right]^2.
	\label{eq:missp.model.5}
\end{align}
For large $J$, $p_{z,w}(z_i,w_i)$ is close to $\phi(z_i,w_i)$, a bivariate normal density with variances equal $1$ and $1+\tfrac{1}{2}\mathbb{E}_p[\epsilon_{ij}^4]$ and covariance that is zero if and only if the microdata is symmetric.
As a result, $\partial_z \log p_{z|w}(z_i\mid w_i) + z_i$ is $O_p(1)$ and consequently $\Delta_1$ is $O_p(\sqrt{J})$. The first term of \eqref{eq:missp.model.5} is thus $O(\tfrac{1}{J})$. This implies the misspecified Bayes estimator incurs a larger error order in the heteroskedastic setting than under homoskedasticity, where the equivalent term with $\Delta_0$ is $O(\tfrac{1}{J^2})$. 
The intuition is both Lemmas~\ref{lem:missp.hom}~and~\ref{lem:missp.het} use the Tweedies' representations of the Bayes estimators, but \eqref{eq:tweed.hom.2}, the representation under homoskedasticity, requires only the normality of $y_i$ whereas Theorem~\ref{thm:main} additionally exploits the independence of $y_i$ and $s_i^2$ implied by microdata normality.  
In fact, $\Delta_1$ becomes the same order as $\Delta_0$ when the microdata is symmetric -- i.e. when $y_i$ and $s_i^2$ are asymptotically independent so $p_{z|w}(z_i\mid w_i) = p_{z}(z_i) $. 

We now briefly argue that the second and third terms of \eqref{eq:missp.model.5} are both $O(\tfrac{1}{J^2})$, and so conclude that the misspecification error in the heteroskedastic setting -- subject to control of the fourth term -- is generally of order $O(\tfrac{1}{J})$.  
Under regularity conditions\footnote{These conditions include a polynomial-bounded growth of $\partial_w^2 r(z,w+u)/r(z,w)$ in $u$. With the exponential weight $e^{-t/\sigma_i^2}$, we are then able to confine the region of integration in $\Delta_2$ to $\{t:\,t\leq C\log J\}$ for some large constant $C$ (the tail region is negligible to order $J^{-1}$). This is a region where $\delta_t \asymp t/\sqrt{J}$ is uniformly small enough for a Taylor expansion of $r(z,w+\delta_t)$ about $\delta_t=0$.} we can write 
\begin{equation}
\begin{aligned}[b]
	\Delta_2 =& -\int_0^\infty e^{-t / \sigma_i^2} \cdot \delta_t \cdot \partial_w\log r(z_i,w_i) dt + O_p(J^{-1}) 
	\\
	=&  \sqrt{2}\tfrac{J}{J-1} \cdot  [-\sigma_i^2 \tfrac{1}{\sqrt{J}}] \cdot \partial_w\log r(z_i,w_i) + O_p(J^{-1}). 
\end{aligned}
\end{equation}
For large $J$, $r(z_i,w_i)$ is approximately $\frac{\phi(z_i,w_i)}{\phi(z_i)\phi(w_i)}$  and thus $\Delta_2$ will be $O_p(\tfrac{1}{\sqrt{J}})$. Furthermore $\partial_z \log p_{z|w}(z_i|w_i) $ will be $O_p(1)$ and thus the second term of \eqref{eq:missp.model.5}  is $\tfrac{1}{J^2}\cdot O(1)$ = $O(\tfrac{1}{J^2})$.
By the same argument we can show 
\begin{equation}
	\partial_z \Delta_2 = -\sigma_i^2\sqrt{\tfrac{2}{J}}\cdot \partial_z\partial_w\log r(z_i,w_i) + O_p(J^{-1})
	= -\sigma_i \sqrt{\tfrac{2}{J}}\partial_z\partial_w\log r(z_i,w_i) + O_p(J^{-1}) 
\end{equation}
and $\partial_z\partial_w\log r(z_i,w_i)=O_p(1)$. As a result, the third term of \eqref{eq:missp.model.5} is  of order $O(\tfrac{1}{J^2})$. 

Now for the fourth term of \eqref{eq:missp.model.5}, this is exactly zero if the mixture density generated under \eqref{eq:missp.model.4} is able to match the mixture density of the true model \eqref{eq:missp.model.3} in the same spirit as Assumption~\ref{ass:missp.gmm}. More generally, the condition that $\mathbb{E}_{p,G_0}[\nabla_1]^2 = O(J)$ is sufficient for the misspecification bound to remain as $O(\frac{1}{J})$.

% The first involving $\Delta$ is similar to $\Delta_0$ of the homoskedastic setting, where its magnitude depends on the discrepancy between the model and the Gaussian approximation in terms of their densities' log-derivatives (w.r.t. $y$). The second and third terms are new and depend on $\Delta_2$. Roughly speaking and by inspecting the integrand within $\Delta_2$, these measure the misspecification in terms of the relative shapes of the density of the sample variance $s^2$ at different points, and how these differ from those implied by a Gaussian sampling model.  
% Thus in the heteroskedastic setting there are more directions in which the misspecification of a Gaussian model hurts the misspecified estimator $\mathbb{E}_{\phi,B_0}[\mu\mid y,s]$.
% Nonetheless, these three terms are fixed for a given model and crucially, scale as $\frac{1}{J^2}$. Thus the misspecification error in the heteroskedastic setting will be of order $O(\frac{1}{J^2})$ as in the homoskedastic setting, and relatively benign for a moderate $J$.

\section{Simulations}
\label{sec:simul}

% This section utilizes numerical experiments to evaluate the finite-sample performance of the proposed estimators and alternative sets of estimators. 

\subsection{Data Generating Process}
\label{subsec:simul.dgp}

The DGP for the simulations proceeds in two stages, where conditional on $\left\{(\mu_i,\sigma_i)\right\}_{i=1}^n$ the second stage generates observations $y_{ij}$ according to Assumption~\ref{ass:disagg.model} with $J=15$.
In the first stage, we generate $(\mu_i,\sigma_i)\sim_{\text{iid}}G_0$ by coupling Gaussian and Gamma random variables\footnote{$\gamma(\kappa,\lambda)$ denotes the Gamma distribution with shape $k$ and scale $\lambda$, and $F_{\gamma(\kappa,\lambda)}^{-1}(\cdot)$ denotes the quantile function of $\gamma(\kappa,\lambda)$. $\Phi(\cdot)$ denotes the standard normal CDF.}: 
$\mu_i = \alpha + \nu^{1/2} \cdot \xi_{i,1}$ and $\sigma_i^2 = F_{\gamma(\kappa,\lambda)}^{-1} \big(\Phi(\xi_{i,2})\big)$, where
\begin{equation}
	\begin{bmatrix}
		\xi_{i,1} \\ \xi_{i,2}
	\end{bmatrix} 
	\sim \,\mathcal{N}\left[
	\begin{bmatrix}
		0 \\ 0
	\end{bmatrix},
	\begin{bmatrix}
		1 & \rho \\ \rho & 1
	\end{bmatrix}
	\right].
	\label{eq:simul.dgp}
\end{equation}
Thus $\mu_i$ is a Gaussian variable with mean $\alpha$ and variance $\nu$, whereas $\sigma_i^2$ is a Gamma variable with shape $\kappa$ and scale $\lambda$. 
The parameter $\rho$ controls the correlation between $(\mu_i,\sigma_i^2)$.
In the calibrated DGP, we set $(\alpha,\nu,\kappa,\lambda,\rho)$ such that the first five moments of $(\mu_i,\sigma_i^2)$ matches those estimated from the dataset of our application: 
\begin{equation}
	\mathbb{E}[\mu_i] = 0 \, , \mathbb{V}[\mu_i] = 0.018\, , 
	\mathbb{E}[\sigma_i^2] = 0.26\, , \mathbb{V}[\sigma_i^2] = 0.003\, , \text{cor}(\mu_i,\sigma_i^2) = -0.38.
\end{equation}
% \begin{figure}[t!]
% 	\caption{DGP Calibrated to Teacher Value-added Application}
% 	\label{fig:simul.dens}	
% 	\begin{center}
% 		\begin{tabular}{ccc}
% 			\includegraphics[width=.3\textwidth]{simul_fig1.pdf} 
% 			&
% 			\includegraphics[width=.3\textwidth]{simul_fig2.pdf}
% 			&
% 			\includegraphics[width=.3\textwidth]{simul_fig3.pdf}
% 		\end{tabular} 
% 	\end{center}
% 	{\footnotesize {\em Notes:} The DGP calibrated to fit the moments of the teacher value-added application. From the left to the right panel: the density of $\mu_i$, the density of $\sigma_i^2$, and finally the bivariate density of $(\mu_i,\sigma_i^2)$ with $\mu_i$ on the $x$-axis.}\setlength{\baselineskip}{4mm}
% \end{figure}
% Figure~\ref{fig:simul.dens} displays the marginal and joint densities of $(\mu_i,\sigma_i^2)$ under this calibration. 
We then perturb the DGP along a single dimension (e.g., changing $\mathbb{V}[\mu_i]$ while maintaining the other four moments) and observe the performances of the estimators.

\subsection{Estimators}
\label{subsec:simul.est}
In implementing the feasible versions of the oracles, we adopt a $n^{1/2}$-order sieve MLE $f_{\hat{G}(n^{1/2})}$-- see Remark~\ref{rmk:implementation}.
We call this set of estimators \textsc{het}.
%two orders: $\log n$ and $n^{3/4}$. In terms of sieves, vdv2001's theoretical results suggest that a sieve MLE of order $\log n$ has the same asymptotic performance as the exact NPMLE in estimating $f_{G_0}$. On the other hand the order of $n^{3/4}$, being much closer to $n$ than $\log n$ is, is meant to be a closer approximation of the exact NPMLE.  
%This will provide insights on how the choice of sieves impacts the finite-sample performance of \textsc{het} estimators under differing DGPs. 
As alternatives, 
we include a `known-heteroskedasticity' version of \textsc{het} that treats $s_i^2$ as if it were $\sigma_i^2$, with the prior on $\mu_i$ estimated using a sieve MLE as above\footnote{Thus, such an estimator implicitly assumes independence between $\mu_i$ and $\sigma_i$.}. This we refer to as $\textsc{het-s}$. 
We also include a homoskedastic version of \textsc{het}, i.e., imposing homoskedasticity $\sigma_i^2 = \sigma^2$ estimated using $\tfrac{1}{n}\textstyle\sum_{i=1}^ns_i^2$, and where the prior on $\mu_i$ is similarly estimated as above. We call this $\textsc{hom}$. 
Finally, the set of naive estimators, which uses $y_i$ to directly estimate $\mu_i$ and $y_i + \Phi^{-1}(\alpha)\cdot s_i$ to estimate $q_{i,\alpha}$, is included and referenced as \textsc{naive}.

\subsection{Results} 

\subsubsection{Relative Regrets I}
\begin{table}[!ht]
	\centering
	\caption{Relative Regrets For MSE Loss Function}
	\small 
	\begin{spacing}{1}
	    \begin{tabular}{ llcccc }
		\toprule
%		(I) & (I) & (I) & (I) & (I) & (I) & (I) & (I) & (I)
%		\\
		$n$ & Estimator & \phantom{1}Calibrated\phantom{1} & 
		\phantom{1}$\mathbb{V}[\sigma_i^2]=$ \phantom{1} & 
		\phantom{1}$\text{cor}(\mu_i,\sigma_i^2)$ \phantom{1} &
		\phantom{1}$J=7$ \phantom{1} \\
		& & DGP &
		0.018 & = -0.76 & \\
		\cmidrule(lr){3-6}
		& & 
		\multicolumn{4}{c}{Estimation Target: $\{\mu_i\}_{i\in[n]}$} \\
		\midrule
		$200 $
		& $\textsc{het}$  				& .12 & .19 & .12 & .17 \\
		& $\textsc{het-s}$                & .13 & .17 & .20 & .20 \\
		& $\textsc{hom}$                & .12 & .25 & .19 & .13 \\
		& $\textsc{naive}$              & 1.0 & 1.2 & 1.2 & 2.1 \\
		[0.5ex]
		$5,000 $
		& $\textsc{het}$  			    & .0089 & .019 & .0087 & .011 \\
		& $\textsc{het-s}$                & .053 & .080 & .12 & .15 \\
		& $\textsc{hom}$                & .031 & .14 & .11 & .027 \\
		& $\textsc{naive}$              & 1.0 & 1.2 & 1.1 & 2.1 \\
		\cmidrule(lr){3-6}
		& & 
		\multicolumn{4}{c}{Estimation Target: $\{q_{0.1,i}\}_{i\in[n]}$} \\
		\midrule
		$200 $
		& $\textsc{het}$    			  & .13 & .15 & .11 & .18 \\
		& $\textsc{het-s}$                  & .70 & .21 & .50 & 1.5 \\
		& $\textsc{hom}$                  & .21 & 1.3 & .30 & .18 \\
		& $\textsc{naive}$                & 2.0 & 1.2 & 1.7 & 3.3 \\
		[0.5ex]
		$5,000 $
		& $\textsc{het}$    			  & .0097 & .016 & .0079 & .012 \\
		& $\textsc{het-s}$                  & .64 & .17 & .46 & 1.4 \\
		& $\textsc{hom}$                  & .16 & 1.2 & .25 & .076 \\
		& $\textsc{naive}$                & 2.0 & 1.2 & 1.7 & 3.3 \\
		\bottomrule
	\end{tabular}
	\end{spacing}
	\label{tab:simul.mse}
	
	{\footnotesize {\em Notes:} 
		Calibrated DGP satisfies $\mathbb{E}[\mu] = 0$, $\mathbb{V}[\mu] = 0.018$, $\mathbb{E}[\sigma^2] = 0.26$, $\mathbb{V}[\sigma^2] = 0.003$, and $\text{cor}(\mu,\sigma^2) = -0.38$, with $J=15$.
	}\setlength{\baselineskip}{4mm}
\end{table}

Table~\ref{tab:simul.mse} displays
% \footnote{Appendix~\ref{sec:simul.supp} explores further numerical experiments, when the estimation objective is $\{\sigma_i\}_{i\in[n]}$ and $\{\sigma_i^2\}_{i\in[n]}$ instead.} 
the relative regrets involving the MSE loss as defined in \eqref{eq:risk.mse}:
\begin{equation}
	\frac{
		\hat{\mathbb{E}}[l(\hat{\mu},\mu)] - 
		\hat{\mathbb{E}}[l(\hat{\mu}^*,\mu)]
	}{
		\hat{\mathbb{E}}[l(\hat{\mu}^*,\mu)]
	}
	\qquad 
	\text{ and }
	\qquad 
	\frac{
		\hat{\mathbb{E}}[l(\hat{q}_{0.1},q_{0.1})] - 
		\hat{\mathbb{E}}[l(\hat{q}_{0.1}^*,q_{0.1})]
	}{
		\hat{\mathbb{E}}[l(\hat{q}_{0.1}^*,q_{0.1})]
	}
\end{equation}
for varying estimators $\hat{\mu}$ or $\hat{q}_{0.1}$ and DGPs, and where the expectation $\hat{\mathbb{E}}$ is taken across 1,000 simulation rounds. 
The 3\textsuperscript{rd} column displays results for the calibrated DGP, whereas the DGP for the 4\textsuperscript{th} column onward each perturb the calibrated DGP along one dimension. 
For example, the 4\textsuperscript{th} column displays results for the DGP with the same moments, except that now $\mathbb{V}[\sigma_i^2] = 0.018$.
The top and bottom half of Table~\ref{tab:simul.mse} are for the estimation targets of $\{\mu_i\}_{i\in[n]}$ and $\{q_{0.1,i}\}_{i\in[n]}$ respectively.

In the calibrated DGP, $\{\sigma_i^2\}_{i\in[n]}$ is relatively homogeneous compared to $\{\mu_i\}_{i\in[n]}$, so \textsc{hom} applies a nearly optimal Bayes correction (see $(i)$ of Theorem~\ref{thm:main} and \eqref{eq:tweed.hom.2}) for $\{\mu_i\}_{i\in[n]}$ and is competitive with \textsc{het}. \textsc{het-s} is also comparable, so plugging in $s_i$ for $\sigma_i$ is only mildly sub-optimal here. As we show later, however, this similarity in MSE performance masks differences in detecting units with low $\mu_i$, a policy-relevant exercise. For large $n$ the relative regret of \textsc{het} vanishes, whereas \textsc{hom} retains a non-negligible regret due to the (limited) heterogeneity in $\{\sigma_i^2\}_{i\in[n]}$.

When the estimation target shifts to $\{q_{0.1,i}\}_{i\in[n]}$, the performance of \textsc{hom} worsens relative to \textsc{het}. This is because $\sigma_i^2$ not only matters for the Bayes correction in estimating $\mu_i$, but also directly as part of the estimation target: $q_{0.1,i}=\mu_i+\Phi^{-1}(0.1)\cdot\sigma_i$ where $\Phi^{-1}(0.1)=-1.28$. The sub-optimality induced by \textsc{hom}'s homoskedasticity assumption turns out to be large even for a limited heterogeneity in $\{\sigma_i^2\}_{i\in[n]}$. 
Even more striking is \textsc{het-s} where, for the same reason, simply plugging in the imprecise $s_i$ for $\sigma_i$ incurs a large regret. 

The relative-regret differential between \textsc{het} and both \textsc{hom} and \textsc{het-s} is also especially large in the case of heterogeneous $\{\sigma_i^2\}_{i\in[n]}$ (see 4\textsuperscript{th} column) or strong dependence between $\mu_i $ and $\sigma_i^2$ (see 5\textsuperscript{th} column). For instance recall that \textsc{hom}, which assumes homogeneous $\sigma_i=\sigma$, implicitly assumes no dependence between $(\mu_i,\sigma_i)$. When there is meaningful heterogeneity in $\{\sigma_i^2\}_{i\in[n]}$ and dependence between $(\mu_i,\sigma_i)$, \textsc{hom} is unable to take advantage of the information in $s_i^2$ regarding $\sigma_i^2$ to sharpen the estimation of $\mu_i$ and so incurs a larger relative regret. On the other hand, \textsc{het} optimally exploits this information for estimating $\{\mu_i\}_{i\in[n]}$ and its relative regret remains similar from the 3\textsuperscript{rd} column to 4\textsuperscript{th} or 5\textsuperscript{th} columns.

The final point of comparison is between columns 3 and 6, where $J$, the number of observations per unit, is reduced from 15 to 7. In this case, the relative regrets of \textsc{het-s} increase significantly more so than \textsc{het} and \textsc{hom}. This illustrates the point that plugging in $s_i$ for $\sigma_i$ is most sub-optimal when $J$ is small, which are the types of setting where empirical Bayes methods are commonly employed.   

In sum, we conclude that \textsc{het} dominates \textsc{hom}, \textsc{het-s} and \textsc{naive} for both targets $\{\mu_i\}_{i\in[n]}$ and $\{q_{0.1,i}\}_{i\in[n]}$ across all specifications  and in particular even for small sample sizes. Thus, there is little finite-sample cost to adopting \textsc{het} in practice, and much to be gained in terms of the relative regrets especially when $\{\sigma_i^2\}_{i\in[n]}$ is hetergeneous or when there exists strong dependence between $\mu_i$ and $\sigma_i^2$.

\subsubsection{Relative Regrets II}

\begin{figure}[t!]
	\caption{Relative Regrets For Detecting Low $\mu_i$ or $q_{0.1,i}$}
	\label{fig:simul.detect}	
	\begin{center}
		\begin{tabular}{cc}
			Target: $\{\mu_i\}_{i=1}^n$
			&
			Target: $\{q_{0.1,i}\}_{i=1}^n$
			\\ [1.5ex]
			\includegraphics[width=.30\textwidth]{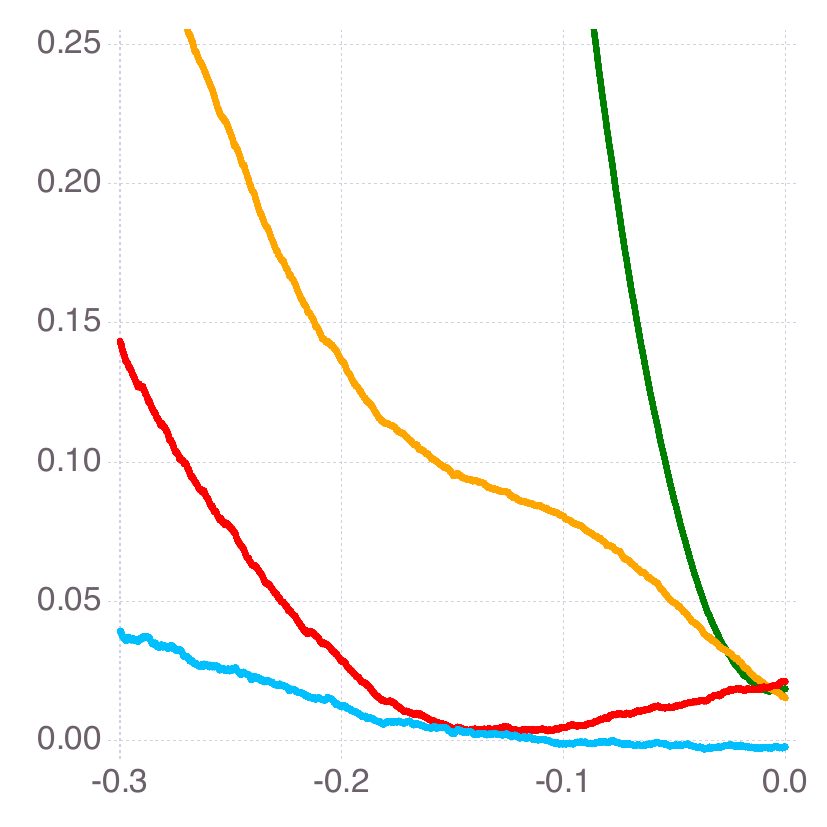} 
			&
			\includegraphics[width=.30\textwidth]{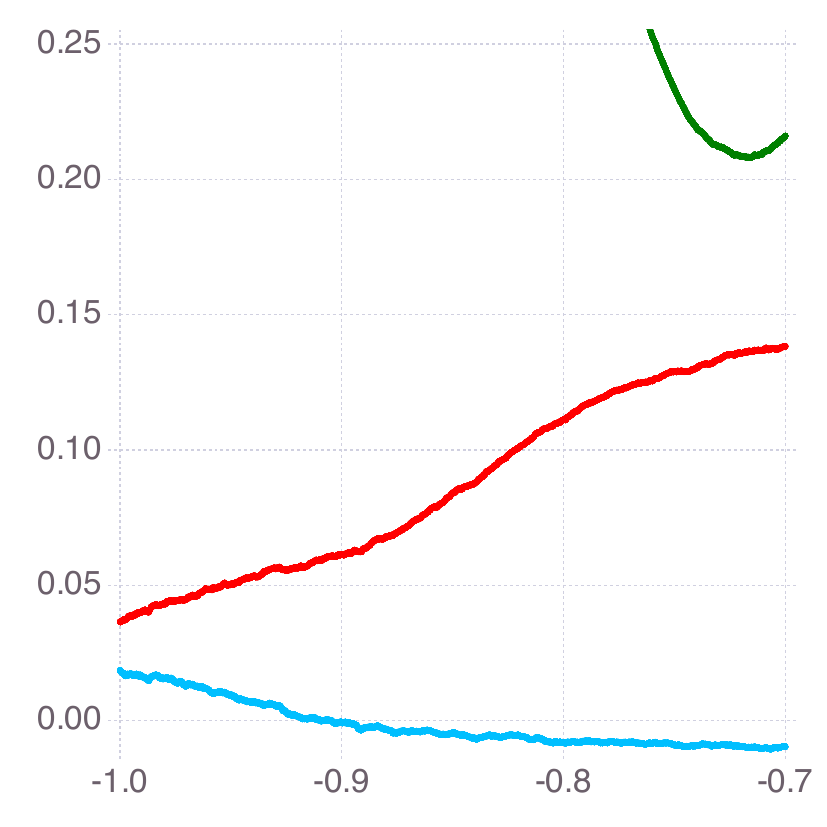}
		\end{tabular} 
	\end{center}
	{\footnotesize {\em Notes:} 
		The color coding is {\color{jblue} \textsc{het}}, {\color{red} \textsc{hom}}, {\color{jorange} \textsc{het-s}}, and {\color{jgreen2} \textsc{naive}} respectively. The underlying DGP is the calibrated DGP, with $n=5000$.
		The threshold value $c$ is varied on the $x$-axis. 
		The values $(-0.3, 0.0)$ and $(-1.0, -0.7)$ corresponds to the $(0.02, 0.5)$ population quantiles of $\mu_i$ and $q_{0.1,i}$ respectively.
	}\setlength{\baselineskip}{4mm}
\end{figure}

Figure~\ref{fig:simul.detect} displays the relative regrets involving the loss function $l_{\text{d}}(\cdot,\cdot,c)$ as defined in \eqref{eq:loss.detect}: 
\begin{equation}
	\frac{
		\hat{\mathbb{E}}[l_{\text{d}}(\hat{\mu},\mu,c)] - 
		\hat{\mathbb{E}}[l_{\text{d}}(\hat{\mu}^*,\mu,c)]
	}{
		\hat{\mathbb{E}}[l_{\text{d}}(\hat{\mu}^*,\mu,c)]
	}
	\qquad 
	\text{ and }
	\qquad 
	\frac{
		\hat{\mathbb{E}}[l_{\text{d}}(\hat{q}_{0.1},q_{0.1},c)] - 
		\hat{\mathbb{E}}[l_{\text{d}}(\hat{q}_{0.1}^*,q_{0.1},c)]
	}{
		\hat{\mathbb{E}}[l_{\text{d}}(\hat{q}_{0.1}^*,q_{0.1},c)]
	}
\end{equation}
for the left and the right panels of the figure respectively. 
Intuitively, the loss functions underlies a decision problem of identifying units with $\mu_i$ or $q_{0.1,i}$ below a fixed threshold $c$, where the loss from each unit increases symmetrically in the magnitude of misidentification. 
The DGP utilized is the calibrated DGP with $n=5000$, and
the threshold value $c$ is varied on the $x$-axis of Figure~\ref{fig:simul.detect}.

From the left panel, we observe that \textsc{het} dominates \textsc{hom} for all values of $c$. Furthermore, even though the heterogeneity in $\{\sigma_i^2\}_{i\in[n]}$ is limited under this DGP, the relative regret differential between \textsc{het} and \textsc{hom} widens considerably for small values of $c$. 
This is likely due to the negative dependence between $\mu_i$ and $\sigma_i^2$: units with the smallest $\mu_i$ tend to have the largest $\sigma_i^2$, and \textsc{hom}, which assumes a common $\sigma^2$ for all units, will suffer the largest regret at detecting units with small $\mu_i$ because it is precisely at these regions that the homoskedasticity assumption is most misplaced.

The story is reversed for detecting units with small $q_{i,0.1}$, in that the regret differential between \textsc{het} and \textsc{hom} decreases for smaller values of $c$. 
It turns out that for \textsc{hom} in this case, for small values of $\mu_i$, the (negative) bias in the estimation of $\mu_i$ arising from the homoskedasticity assumption cancels out the (negative) bias in the estimation of $\sigma_i^2$ through the quantile's expression: $q_{0.1,i}= \mu_i - 1.28\cdot \sigma_i$.\footnote{To elaborate in detail: in this DGP the units with small $\mu_i$ tend to have larger $\sigma_i$, and so the homoskedasticity assumption falsely imputes small $\sigma_i$s to these units; thus \textsc{hom} estimators of $\sigma_i$ for these units will be too small. As a result of this, the Bayes correction for the estimation of $\mu_i$ is smaller than optimal and the \textsc{hom} estimates of these $\mu_i$ will also be too small.}

Finally, \textsc{het-s} and \textsc{naive}, which simply plug in $s_i$ and $(y_i,s_i)$ for the unknown $\mu_i$ and $(\mu_i,\sigma_i)$  respectively, perform poorly at all values of $c$. In fact, they have relative regret that are uniformly above $0.2$ in the right panel. This is because they identify many units as having $\mu_i$ or $q_{0.1,i}$ below $c$, but the vast majority of these are likely due to noise in the sufficient statistics than a true reflection of the parameters $(\mu_i,\sigma_i^2)$.

\section{An Empirical Illustration}
\label{sec:appl}

% We conclude with an illustration of our methodology by estimating teacher quality using a matched student-teacher dataset. 

\subsection{Data}
The dataset is a panel of elementary school students in grades 3--5 in North Carolina. Each observation is a (student, year, grade, subject) cell with covariates such as ethnicity, gender, and economic status. 

\subsection{Estimation}
For simplicity we restrict attention to 2019 observations with class size between 16 and 19, to excludes abstract from special-education classes and class-size heterogeneity. This leaves 25,410 observations, each a unique student, across $n=1,458$ teachers.

We take a simple model of teacher value-added:
\begin{equation}
	y_{j}^* = \mu_{i(j)} + x_j^\prime \beta_0 + \sigma_{i(j)}\epsilon_{j},
	\quad 
	\epsilon_{j} \mid \mu_{i(j)},\sigma_{i(j)},x_j \sim\mathcal{N}[0,1]
\end{equation}
where $y_{j}^*$ is student $j$'s 2019 reading test score, and $i(\cdot)$ is the matching function that returns student $j$'s teacher identity in 2019. The microdata normality assumption underlying \eqref{eq:app.model} is examined in the Supplementary Appendix (Section~\ref{subsec:app.normality}) with residual-based specification tests and we find no evidence against normality.
The covariates $x_j$ include student $j$'s 2018 reading score, ethnicity, gender, family economic status, english-language-learner status, and age; test scores are standardized to mean 0 and standard deviation 1 within each (year,grade) cell. The identifying assumption is that, conditional on $x_j$, the student effect $\epsilon_j$ is independent of her matched teacher's profile $\left(\mu_{i(j)},\sigma_{i(j)}\right)$. Following the literature, we first purge the covariate effects by regression\footnote{We estimate $\beta_0$ using within-teacher variation and move $x_j^\prime \hat{\beta}$ to the left-hand side, defining $y_j:=y_j^*-x_j^\prime\hat{\beta}$; $\hat{\beta}$ is consistent for $\beta_0$ since the relevant sample size $n$ is large.}, leading to the model of Assumption~\ref{ass:disagg.model} in matching-function notation:
\begin{equation}
	y_{j} = \mu_{i(j)} + \sigma_{i(j)}\epsilon_j, 
	\quad 
	\epsilon_{j} \mid \mu_{i(j)},\sigma_{i(j)},x_j \sim \mathcal{N}[0,1]
	\label{eq:app.model}
\end{equation} 
for $j=1,\dots,J_i$ and $i=1,\dots,n$, where $J_i:=\lvert\left\{j:i(j)=i\right\}\rvert$.

The difference with the standard model in the literature is that we allow for teacher-level heteroskedasticity. 
When teachers differ in both $\mu_i$ and $\sigma_i$, teacher quality can be defined in several ways; for two teachers with the same $\mu_i$, whether a higher $\sigma_i$ is preferred depends on the policymaker's preferences. The most common measure is the mean value-added $\mu_i$, how the average student ($\epsilon_j=0$) would perform under each teacher. We also consider $q_{0.1,i}=\mu_i-1.28\cdot\sigma_i$, how the 10\textsuperscript{th}-percentile student would perform under teacher $i$, which penalizes teachers with high $\sigma_i$.

\subsection{Estimators}
\label{subsec:app.est}
With the model \eqref{eq:app.model} in hand, we construct the teacher-specific sufficient statistics 
\begin{equation}
	y_i = \tfrac{1}{J_i} \textstyle\sum_{j:i(j)=i}y_j,
	\quad
	s_i^2 = \tfrac{1}{J_i-1} \textstyle\sum_{j:i(j)=i}\left[y_j-y_i\right]^2
\end{equation}
and estimate $(\mu_i,q_{0.1,i})$ using firstly the proposed set of estimators (called \textsc{het}), secondly a modification of \textsc{het} that plugs in $s_i^2$ for the unknown $\sigma_i^2$ (called \textsc{het-s}), another that assumes homoskedasticity (called \textsc{hom}), and lastly \textsc{naive} that uses the teacher-specific sufficient statistics $(y_i,s_i^2)$ directly; see Section~\ref{subsec:simul.est} for more details. 
% We abstract from class size heterogeneity in the estimation which we justify by our sample restriction to a relatively narrow band of class sizes between 16 to 19. 

\subsection{Results}
\label{subsec:app.results}
\begin{figure}[t!]
	\caption{Estimates from \textsc{het} and Density Plots}
	\label{fig:app.dens1}	
	\begin{center}
		\begin{tabular}{cc}
			\includegraphics[width=.30\textwidth]{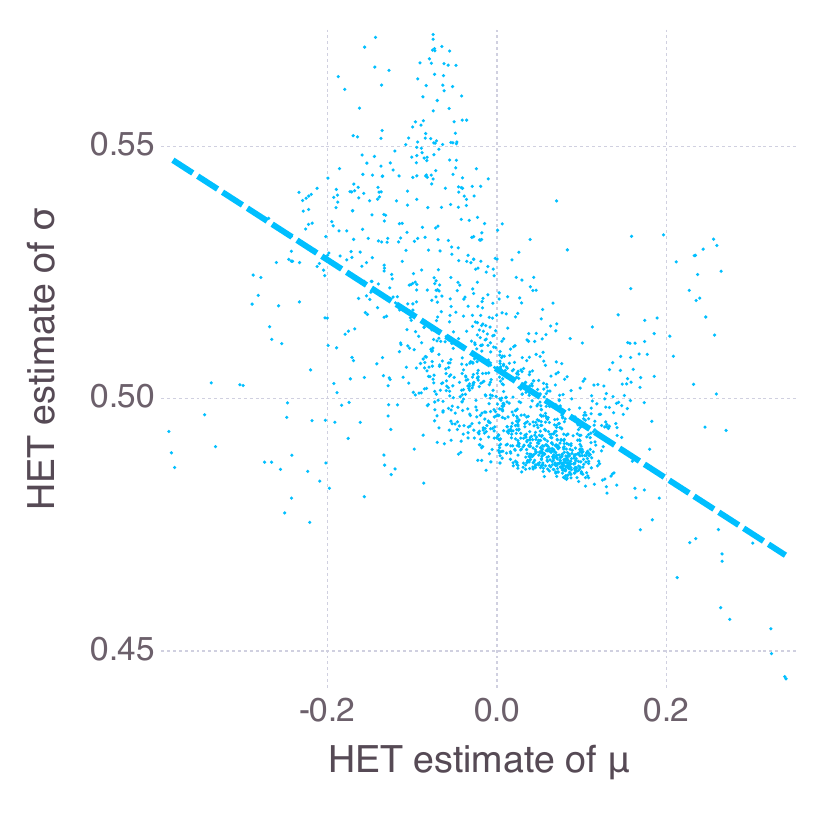} 
			&
			\includegraphics[width=.30\textwidth]{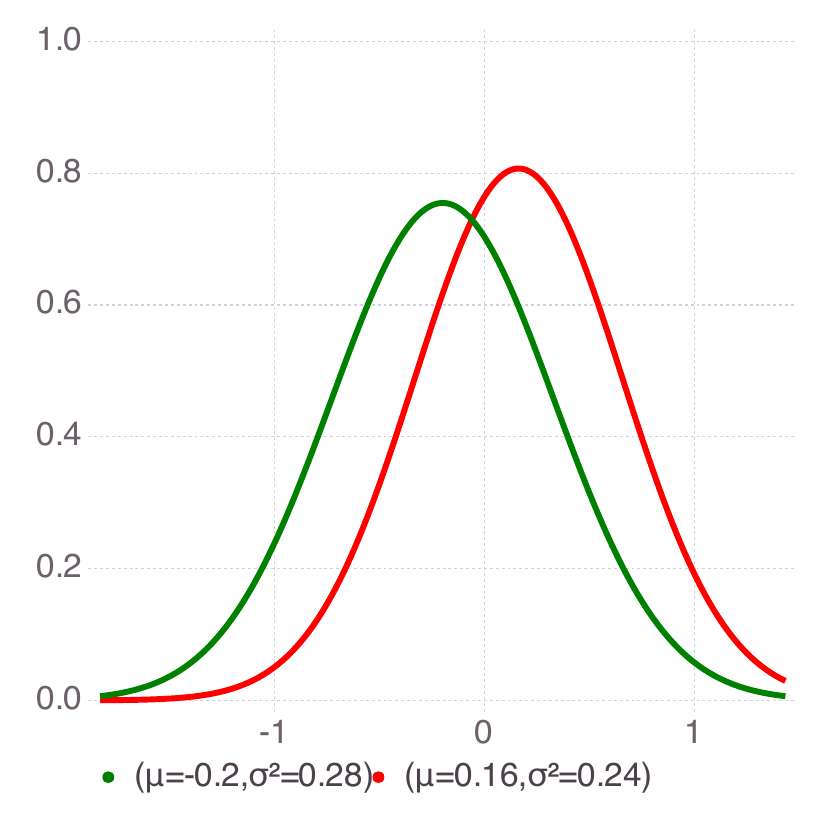}
		\end{tabular} 
	\end{center}
	{\footnotesize {\em Notes:} The blue line in the left panel is the line of best fit (slope $-0.14$). In the right panel, the blue (red) curve is the sampling density at the average \textsc{het} estimate of $(\mu_i,\sigma_i)$ among teachers in the bottom (top) 5\% by $q_{0.1,i}$.}\setlength{\baselineskip}{4mm}
\end{figure}

The main finding here is that teachers with small $\mu_i$ tend to have large $\sigma_i$, as is shown in the left panel of Figure~\ref{fig:app.dens1} that plots the \textsc{het} estimates of $\sigma_i$ against those of $\mu_i$. This implies differences in teacher quality is exacerbated, in that teachers who differ in how the average student performs under them (that is, $\mu_i$) tend to differ even more in how the lower-percentile students will perform under them (e.g., $q_{0.1,i}$). To see this, the right panel of Figure~\ref{fig:app.dens1} displays the test score distribution of \eqref{eq:app.model} for two representative teacher profiles $(\mu,\sigma)$ constructed from the \textsc{het} estimates; see the notes of Figure~\ref{fig:app.dens1} for construction details.  
Re-assigning a mean student from the lower to higher value-added teacher shifts the student's test score up by 0.36 (i.e. 0.36 standard deviations of raw test scores).  
The negative correlation between $\mu_i$ and $\sigma_i$ implies that re-assigning a $10^{\text{th}}$ percentile student from the lower to higher value-added teacher shifts the student's test score up by a larger 0.42.  
%The negative correlation between $\mu_i$ and $\sigma_i$ implies that while the two teachers' $\mu$ differ by $0.36$ (i.e., $0.36$ standard deviations of the raw test scores), their $q_{0.1}$ differ by $0.41$, a $12\%$ increase in difference. 

Next, recall by Lemma~\ref{lm:threshold.detect} that $\hat{\mu}^*$ (or $\hat{q}_{0.1}^*$) is also optimal for the problem of detecting teachers with $\mu_i$ (or $q_{0.1,i}$) below a fixed threshold. We thus explore the feasible estimators in this context, where teacher quality is measured by $\{\mu_i\}_{i\in[n]}$ and also by $\{q_{0.1,i}\}_{i\in[n]}$. 
The left (and right) panel of Figure~\ref{fig:app.detect} plots the number of teachers whose estimates of $\mu_i$ (and $q_{\alpha,i}$) are flagged as being less than the threshold value which is varied on the $x$-axis. Unsurprisingly, \textsc{naive} flags the greatest number for every threshold value. It is likely that most of these were flagged because of the large variability of the estimates $(y_i,s_i^2)$ rather than a true reflection of $\mu_i$.

Another feature from the left panel of Figure~\ref{fig:app.detect} is that \textsc{hom} flags more teachers than \textsc{het} for almost all threshold values. For instance, it flags out 18 more teachers than \textsc{het} for the threshold value of $-0.25$, or an increase of 80\% in relative terms. As we have seen from Figure~\ref{fig:app.dens1}, teachers with smaller $\mu_i$ tend to have larger $\sigma_i$. \textsc{hom}, which assumes $\sigma_i=\sigma$ for all $i$, applies a less-than-optimal Bayes correction for this set of teachers and the additional variability in the \textsc{hom} estimates results in more teachers being flagged. On the contrary, \textsc{het} allows for heterogeneous variances and dependence within $(\mu_i,\sigma_i)$ and does not suffer from this problem.

On the right panel where teacher quality is defined in terms of $\{q_{0.1,i}\}_{i\in[n]}$, \textsc{hom} now flags significantly fewer teachers than \textsc{het} for every threshold value. This is again due to the same combination of \textsc{hom} imposing a common $\sigma$ for all teachers and the negative correlation within $(\mu_i,\sigma_i)$, albeit through a different channel. 
Recall that $q_{0.1,i} = \mu_i - 1.28\sigma_i$. As a result, \textsc{hom} understates how poorly the 10\textsuperscript{th}-percentile students will do under teachers with smaller $\mu_i$, whereas \textsc{het}, which incorporates the negative correlation in $(\mu_i,\sigma_i)$, does not. The implications differ markedly in absolute terms: \textsc{hom} flags about 100 fewer teachers than \textsc{het}, a consequence of its a-priori homoskedasticity restriction -- one that Figure~\ref{fig:app.dens1} suggests is misguided.

\begin{figure}[t!]
	\caption{Detecting Teachers Below Fixed Threshold}
	\label{fig:app.detect}	
	\begin{center}
			\begin{tabular}{cc}
					Target: $\{\mu_i\}_{i=1}^n$
					&
					Target: $\{q_{0.1,i}\}_{i=1}^n$
					\\ [1.5ex]
					\includegraphics[width=.30\textwidth]{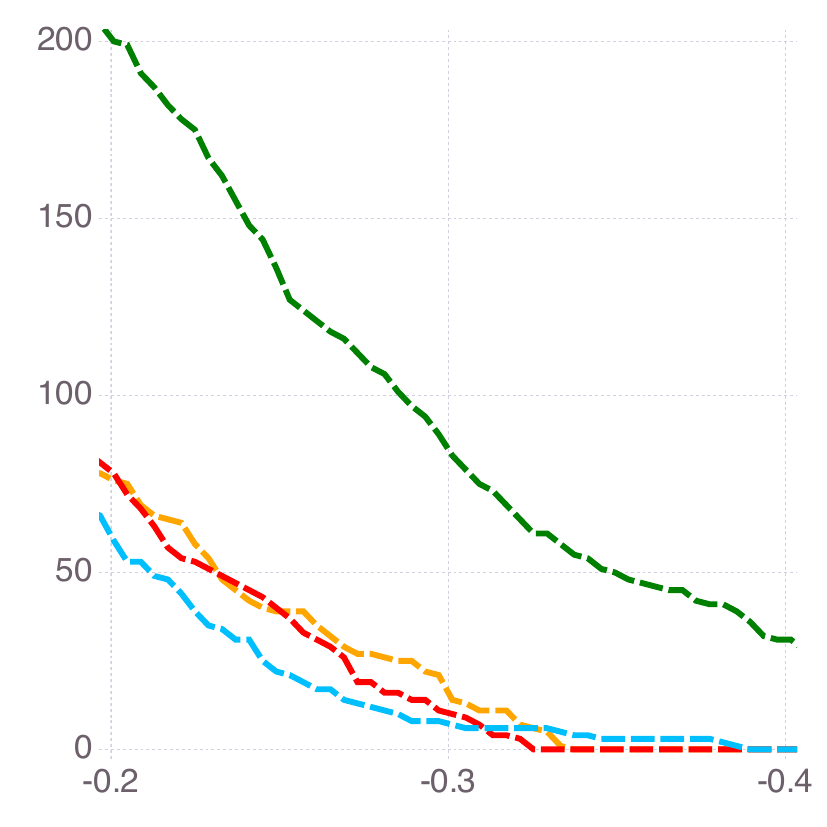} 
					&
					\includegraphics[width=.30\textwidth]{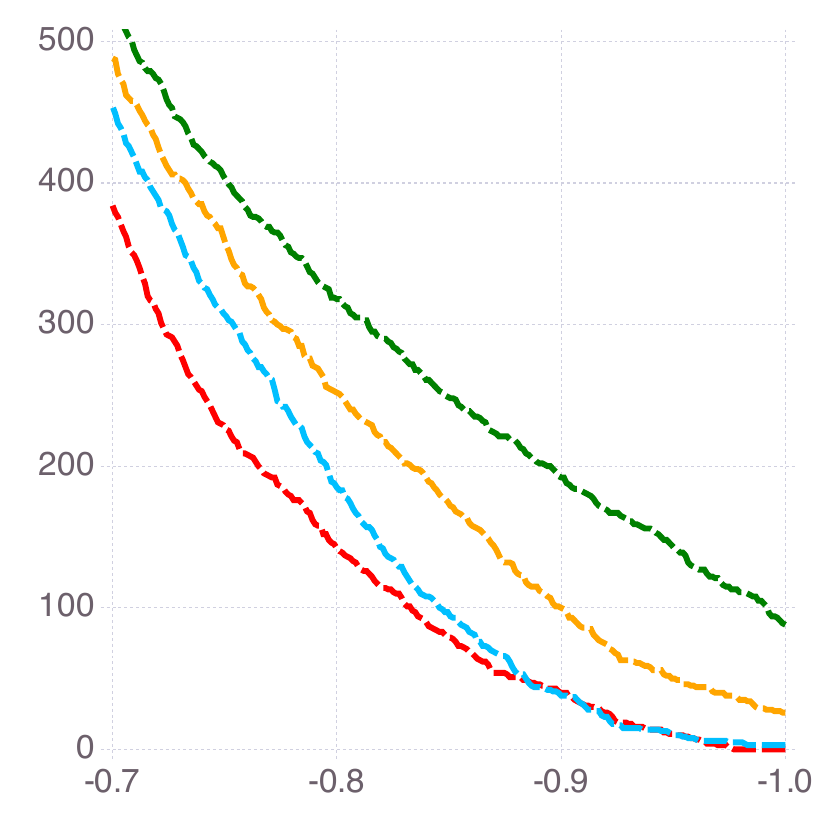}
				\end{tabular} 
		\end{center}
	{\footnotesize {\em Notes:} The color coding is {\color{jblue} \textsc{het}}, {\color{red} \textsc{hom}}, {\color{jorange} \textsc{het-s}}, and {\color{jgreen2} \textsc{naive}} respectively. }\setlength{\baselineskip}{4mm}
\end{figure}

\subsection{Validation Exercise}
\label{subsec:app.validation}
As a simple validation exercise, we split the dataset as $\mathcal{Y} = \mathcal{Y}^{\text{train}} \cup \mathcal{Y}^{\text{test}}$, with $\mathcal{Y}^{\text{train}}$ containing $8$ observations per teacher. The set $\mathcal{Y}^{\text{train}}$ is used to construct the estimators as before, whereas $\mathcal{Y}^{\text{test}}$ is used to unbiasedly estimate the risk of the estimators.
For estimators $\hat{\mu}_{\text{tr}}$ and $\hat{q}_{0.1,\text{tr}}$,
% and $\tilde{\mu}_{\text{tr}}$ 
(subscript ``tr'' indicating its construction based on $\mathcal{Y}^{\text{train}}$),
we provide 
% (see Supplementary Appendix~\ref{sec:app.supp}) 
risk estimators $\hat{R}(\hat{\mu}_{\text{tr}},\mu)$ and $\hat{R}(\hat{q}_{0.1,\text{tr}},q_{0,1})$ 
% and $\hat{R}_d(\hat{\mu}_{\text{tr}},\mu,c)$ 
satisfying $\mathbb{E}_{G_0}[\hat{R}(\hat{\mu}_{\text{tr}},\mu)] = R_{G_0}(\hat{\mu}_{\text{tr}},\mu)$ and $\mathbb{E}_{G_0}[\hat{R}(\hat{q}_{0.1,\text{tr}},q_{0.1})]=R(\hat{q}_{0.1,\text{tr}},q_{0.1})$. See Section~\ref{sec:app.supp} of the Supplementary Appendix for expressions of the risk estimators.

Thus as a truly out-of-sample validation of the estimators, albeit at the smaller sample sizes $J_{\text{tr},i}=8$ for all $i$, we compare the (estimated) risks of various estimators for differing targets.
Table~\ref{tab:app.1} plots the estimated risks for estimation of $\{ \mu_i \}_{i\in[n]}$ and $\{ q_{0.1,i} \}_{i\in[n]}$, where the \textsc{het} estimator outperforms all other estimators. 
In particular, the largest improvements are over \textsc{het-s} and \textsc{naive} that plug in $s_i$ for $\sigma_i$, confirming that treating the heteroskedasticity correctly is important, especially for the estimation of teacher quantiles. 
% The improvements of \textsc{het} over all other estimators are also largest for the estimation of teacher quantiles.   
\begin{table}[htbp]
	\centering
	\caption{Risk Estimates (MSE)}
	\small
	\begin{tabular}{ lcc }
		\toprule
		Estimator & $\{\mu_i\}_{i\in[n]}$ & $\{q_{0.1,i}\}_{i\in[n]}$ \\
		\midrule
		$\textsc{het}$    & 0.0120 & .0163 \\
		$\textsc{hom}$    & 0.0125 & .0173 \\
		$\textsc{het-s}$  & 0.0129 & .0422 \\
		$\textsc{naive}$  & 0.0317 & .0578 \\
		\bottomrule
	\end{tabular}
	\label{tab:app.mse}\\
	{\footnotesize {\em Notes: } Values represent $\hat{R}(\hat{\mu}_{\text{tr}},\mu)$ and $\hat{R}(\hat{q}_{0.1,\text{tr}},q_{0.1})$ with $\hat{\mu}_{\text{tr}}$ and $\hat{q}_{0.1,\text{tr}}$ being the estimators.
	}\setlength{\baselineskip}{4mm}
	\label{tab:app.1}
\end{table}

\section{Conclusion}
\label{sec:concl}

This paper contributes to the growing empirical Bayes (EB) literature by proposing nonparametric EB estimators with regret bounds of the order $\tfrac{1}{n}(\log n)^\kappa$. The upshot is that nonparametric EB methods come at little cost in typical sample sizes and in fact with much benefits in terms of risk reduction.
% The results depend on a normal sampling model and we show that, under certain conditions, the misspecification errors are of the order $O(1/J)$ and thus mild for moderate $J$. 
Importantly, the results are all developed within a framework that features unknown heteroskedasticity. 
Allowing for this feature is essential for two reasons. First, it delivers the optimal unit-specific Bayes correction when estimating the means. Second, it enables comparisons along richer dimensions of heterogeneity, such as the performance of lower-percentile students under each teacher.
The empirical illustration using a matched student-teacher dataset suggest that the detection of low-performing teachers -- a typical exercise with high-stakes consequences -- is sensitive to the estimator choice regardless of how teacher quality is defined. 
In particular, using nonparametric EB estimators that allow for dependence among the teacher-specific parameters is crucial for delivering optimal estimates of teacher quality and detecting low-performing teachers.
    \bibliography{Empirical_Bayes}

@misc{Banerjee2023,
  title = {Nonparametric {{Empirical Bayes Estimation}} on {{Heterogeneous Data}}},
  author = {Banerjee, Trambak and Fu, Luella J. and James, Gareth M. and Mukherjee, Gourab and Sun, Wenguang},
  year = 2023,
  month = aug,
  number = {arXiv:2002.12586},
  eprint = {2002.12586},
  primaryclass = {stat},
  publisher = {arXiv},
  doi = {10.48550/arXiv.2002.12586},
  urldate = {2024-12-30},
  archiveprefix = {arXiv}
}

@article{Brown2009,
  title = {Nonparametric {{Empirical Bayes}} and {{Compound Decision Approaches}} to {{Estimation}} of a {{High-Dimensional Vector}} of {{Normal Means}}},
  author = {Brown, Lawrence D. and Greenshtein, Eitan},
  year = 2009,
  month = aug,
  journal = {The Annals of Statistics},
  volume = {37},
  number = {4},
  issn = {0090-5364},
  doi = {10.1214/08-AOS630}
}

@article{Canale2017,
  title = {Posterior Asymptotics of Nonparametric Location-Scale Mixtures for Multivariate Density Estimation},
  author = {Canale, Antonio and De Blasi, Pierpaolo},
  year = 2017,
  month = feb,
  journal = {Bernoulli},
  volume = {23},
  number = {1},
  issn = {1350-7265},
  doi = {10.3150/15-BEJ746},
  urldate = {2025-11-12}
}

@article{Chen2026,
  title = {Empirical {{Bayes When Estimation Precision Predicts Parameters}}},
  author = {Chen, Jiafeng},
  year = 2026,
  journal = {Econometrica},
  volume = {94},
  number = {2},
  pages = {305--340},
  issn = {0012-9682},
  doi = {10.3982/ECTA22935},
  urldate = {2026-05-22},
  langid = {english}
}

@article{Chetty2014,
  title = {Measuring the {{Impacts}} of {{Teachers I}}: {{Evaluating Bias}} in {{Teacher Value-Added Estimates}}},
  author = {Chetty, Raj and Friedman, John N. and Rockoff, Jonah E.},
  year = 2014,
  month = sep,
  journal = {American Economic Review},
  volume = {104},
  number = {9},
  pages = {2593--2632},
  issn = {0002-8282},
  doi = {10.1257/aer.104.9.2593}
}

@article{Chetty2018,
  title = {The {{Impacts}} of {{Neighborhoods}} on {{Intergenerational Mobility II}}: {{County-Level Estimates}}*},
  author = {Chetty, Raj and Hendren, Nathaniel},
  year = 2018,
  month = aug,
  journal = {The Quarterly Journal of Economics},
  volume = {133},
  number = {3},
  pages = {1163--1228},
  issn = {0033-5533},
  doi = {10.1093/qje/qjy006}
}

@article{Cressie1986,
  title = {The Moment Generating Function Has Its Moments},
  author = {Cressie, Noel and Borkent, Marinus},
  year = 1986,
  month = jan,
  journal = {Journal of Statistical Planning and Inference},
  volume = {13},
  number = {2},
  pages = {337--344},
  issn = {03783758},
  doi = {10.1016/0378-3758(86)90143-6}
}

@article{Efron1973,
  title = {Stein's {{Estimation Rule}} and {{Its Competitors}}--{{An Empirical Bayes Approach}}},
  author = {Efron, Bradley and Morris, Carl},
  year = 1973,
  month = mar,
  journal = {Journal of the American Statistical Association},
  volume = {68},
  number = {341},
  eprint = {2284155},
  eprinttype = {jstor},
  pages = {117},
  issn = {01621459},
  doi = {10.2307/2284155}
}

@article{Efron2011,
  title = {Tweedie's {{Formula}} and {{Selection Bias}}},
  author = {Efron, Bradley},
  year = 2011,
  month = dec,
  journal = {Journal of the American Statistical Association},
  volume = {106},
  number = {496},
  pages = {1602--1614},
  issn = {0162-1459},
  doi = {10.1198/jasa.2011.tm11181}
}

@article{Fan1991,
  title = {On the {{Optimal Rates}} of {{Convergence}} for {{Nonparametric Deconvolution Problems}}},
  author = {Fan, Jianqing},
  year = 1991,
  month = sep,
  journal = {The Annals of Statistics},
  volume = {19},
  number = {3},
  issn = {0090-5364},
  doi = {10.1214/aos/1176348248}
}

@article{Finkelstein2016,
  title = {Sources of {{Geographic Variation}} in {{Health Care}}: {{Evidence From Patient Migration}}*},
  author = {Finkelstein, Amy and Gentzkow, Matthew and Williams, Heidi},
  year = 2016,
  month = nov,
  journal = {The Quarterly Journal of Economics},
  volume = {131},
  number = {4},
  pages = {1681--1726},
  issn = {0033-5533},
  doi = {10.1093/qje/qjw023}
}

@article{Ghosal2001,
  title = {Entropies and Rates of Convergence for Maximum Likelihood and {{Bayes}} Estimation for Mixtures of Normal Densities},
  author = {Ghosal, Subhashis and van der Vaart, Aad W.},
  year = 2001,
  month = oct,
  journal = {The Annals of Statistics},
  volume = {29},
  number = {5},
  pages = {1233--1263},
  publisher = {Institute of Mathematical Statistics},
  issn = {0090-5364, 2168-8966},
  doi = {10.1214/aos/1013203452},
  urldate = {2024-10-07}
}

@techreport{Gilraine2020,
  title = {A {{New Method}} for {{Estimating Teacher Value-Added}}},
  author = {Gilraine, Michael and Gu, Jiaying and McMillan, Robert},
  year = 2020,
  month = may,
  address = {Cambridge, MA},
  institution = {National Bureau of Economic Research},
  doi = {10.3386/w27094}
}

@article{Gu2017,
  title = {Empirical {{Bayesball Remixed}}: {{Empirical Bayes Methods}} for {{Longitudinal Data}}},
  author = {Gu, Jiaying and Koenker, Roger},
  year = 2017,
  month = apr,
  journal = {Journal of Applied Econometrics},
  volume = {32},
  number = {3},
  pages = {575--599},
  issn = {08837252},
  doi = {10.1002/jae.2530}
}

@article{Gu2017a,
  title = {Unobserved {{Heterogeneity}} in {{Income Dynamics}}: {{An Empirical Bayes Perspective}}},
  author = {Gu, Jiaying and Koenker, Roger},
  year = 2017,
  month = jan,
  journal = {Journal of Business and Economic Statistics},
  volume = {35},
  number = {1},
  pages = {1--16},
  issn = {15372707},
  doi = {10.1080/07350015.2015.1052457}
}

@article{Gu2023,
  title = {Invidious {{Comparisons}}: {{Ranking}} and {{Selection}} as {{Compound Decisions}}},
  author = {Gu, Jiaying and Koenker, Roger},
  year = 2023,
  journal = {Econometrica},
  volume = {91},
  number = {1},
  pages = {1--41},
  issn = {0012-9682},
  doi = {10.3982/ECTA19304}
}

@article{Ignatiadis2023,
  title = {Empirical {{Bayes Mean Estimation With Nonparametric Errors Via Order Statistic Regression}} on {{Replicated Data}}},
  author = {Ignatiadis, Nikolaos and Saha, Sujayam and Sun, Dennis L. and Muralidharan, Omkar},
  year = 2023,
  month = apr,
  journal = {Journal of the American Statistical Association},
  volume = {118},
  number = {542},
  pages = {987--999},
  issn = {0162-1459, 1537-274X},
  doi = {10.1080/01621459.2021.1967164},
  urldate = {2025-11-12},
  langid = {english}
}

@article{Ignatiadis2025,
  title = {Empirical Partially {{Bayes}} Multiple Testing and Compound {$\chi$}2 Decisions},
  author = {Ignatiadis, Nikolaos and Sen, Bodhisattva},
  year = 2025,
  month = feb,
  journal = {The Annals of Statistics},
  volume = {53},
  number = {1},
  issn = {0090-5364},
  doi = {10.1214/24-AOS2431},
  urldate = {2025-11-12}
}

@article{Jiang2009,
  title = {General {{Maximum Likelihood Empirical Bayes Estimation}} of {{Normal Means}}},
  author = {Jiang, Wenhua and Zhang, Cun-Hui},
  year = 2009,
  month = aug,
  journal = {The Annals of Statistics},
  volume = {37},
  number = {4},
  issn = {0090-5364},
  doi = {10.1214/08-AOS638}
}

@article{Jiang2020,
  title = {On {{General Maximum Likelihood Empirical Bayes Estimation}} of {{Heteroscedastic IID Normal Means}}},
  author = {Jiang, Wenhua},
  year = 2020,
  month = jan,
  journal = {Electronic Journal of Statistics},
  volume = {14},
  number = {1},
  issn = {1935-7524},
  doi = {10.1214/20-EJS1717}
}

@article{Kline2024,
  title = {A {{Discrimination Report Card}}},
  author = {Kline, Patrick and Rose, Evan K. and Walters, Christopher R.},
  year = 2024,
  month = aug,
  journal = {American Economic Review},
  volume = {114},
  number = {8},
  pages = {2472--2525},
  issn = {0002-8282},
  doi = {10.1257/aer.20230700},
  urldate = {2024-12-30},
  langid = {english}
}

@article{Koenker2014,
  title = {Convex {{Optimization}}, {{Shape Constraints}}, {{Compound Decisions}}, and {{Empirical Bayes Rules}}},
  author = {Koenker, Roger and Mizera, Ivan},
  year = 2014,
  month = apr,
  journal = {Journal of the American Statistical Association},
  volume = {109},
  number = {506},
  pages = {674--685},
  issn = {0162-1459},
  doi = {10.1080/01621459.2013.869224}
}

@article{Kwon2023,
  title = {On {{F}} -Modelling-Based Empirical {{Bayes}} Estimation of Variances},
  author = {Kwon, Yeil and Zhao, Zhigen},
  year = 2023,
  month = feb,
  journal = {Biometrika},
  volume = {110},
  number = {1},
  pages = {69--81},
  issn = {0006-3444},
  doi = {10.1093/biomet/asac019}
}

@article{Kwon2026,
  title = {Optimal {{Shrinkage Estimation}} of {{Fixed Effects}} in {{Linear Panel Data Models}}},
  author = {Kwon, Soonwoo},
  year = 2026,
  journal = {Econometrica},
  volume = {94},
  number = {2},
  pages = {663--677},
  issn = {0012-9682},
  doi = {10.3982/ECTA22386},
  urldate = {2026-05-22},
  langid = {english}
}

@article{Liu2020,
  title = {Forecasting {{With Dynamic Panel Data Models}}},
  author = {Liu, Laura and Moon, Hyungsik Roger and Schorfheide, Frank},
  year = 2020,
  journal = {Econometrica},
  volume = {88},
  number = {1},
  pages = {171--201},
  issn = {0012-9682},
  doi = {10.3982/ECTA14952}
}

@incollection{Robbins1956,
  title = {An {{Empirical Bayes Approach}} to {{Statistics}}},
  booktitle = {Proceedings of the {{Third Berkeley Symposium}} on {{Mathematical Statistics}} and {{Probability}}},
  author = {Robbins, Herbert},
  year = 1956,
  publisher = {{University of California Press, Berkeley and Los Angeles}}
}

@incollection{Robbins1982,
  title = {Estimating {{Many Variances}}},
  booktitle = {Statistical {{Decision Theory}} and {{Related Topics III}}},
  author = {Robbins, Herbert},
  year = 1982,
  pages = {251--261},
  publisher = {Elsevier},
  doi = {10.1016/B978-0-12-307502-4.50019-2}
}

@article{Wong1995,
  title = {Probability {{Inequalities}} for {{Likelihood Ratios}} and {{Convergence Rates}} of {{Sieve MLES}}},
  author = {Wong, Wing Hung and Shen, Xiaotong},
  year = 1995,
  month = apr,
  journal = {The Annals of Statistics},
  volume = {23},
  number = {2},
  pages = {339--362},
  publisher = {Institute of Mathematical Statistics},
  issn = {0090-5364, 2168-8966},
  doi = {10.1214/aos/1176324524},
  urldate = {2024-10-07}
}

@article{Xie2012,
  title = {{{SURE Estimates}} for a {{Heteroscedastic Hierarchical Model}}},
  author = {Xie, Xianchao and Kou, S. C. and Brown, Lawrence D.},
  year = 2012,
  journal = {Journal of the American Statistical Association},
  volume = {107},
  number = {500},
  pages = {1465--1479},
  issn = {01621459},
  doi = {10.1080/01621459.2012.728154}
}
    
    % \clearpage
    \renewcommand{\thepage}{A.\arabic{page}}
    \setcounter{page}{1}
    \begin{appendix}
        %\markright{Appendix -- This Version: \today }
        \renewcommand{\theequation}{A.\arabic{equation}}
        \setcounter{equation}{0}
        \renewcommand*\thetable{A-\arabic{table}}
        \setcounter{table}{0}
        \renewcommand*\thefigure{A-\arabic{figure}}
        \setcounter{figure}{0}

\section{Proofs of Selected Results from Sections~\ref{sec:model}--\ref{sec:missp}}
\label{sec:proofs}

\subsection{Proof of Theorem \ref{thm:main}}
\label{subsec:app.tweed}
\begin{proof}[Proof of Theorem \ref{thm:main}]
	The $i$ indices are omitted for compactness.
	We prove Theorem~\ref{thm:main}(i) in two parts: the first takes $\sigma$ as known (paralleling the homoskedastic Tweedie derivation of \cite{Liu2020}), the second integrates out $\sigma$ by the law of iterated expectations. Throughout we use the conditional independence $y \indep s \mid (\mu,\sigma)$ from Assumption~\ref{ass:disagg.model}.
	
	Let us first take $\sigma$ as known. Further let $p(\cdot)$ denote a probability density of a random variable, with the identity of the random variable inferred from the argument.
	Then since $\int p(\mu\mid y,s^2,\sigma^2)d\mu=1$, we have $\partial_y \int p(\mu\mid y,s^2,\sigma^2)d\mu = 0$. Expanding the integrand via Bayes' rule, applying $y \indep s^2 \mid \mu,\sigma$ to factor $p(y,s^2\mid\mu,\sigma^2)=p(y\mid\mu,\sigma^2)p(s^2\mid\mu,\sigma^2)$, differentiating, and combining terms yields
	\begin{equation}
		\int
		\left\{
		-\frac{y-\mu}{\sigma J^{-1}}
		p(\mu \mid y,s^2,\sigma^2)
		-
		p(\mu\mid y,s^2,\sigma^2)
		\frac{\partial_y p(y,s^2,\sigma^2)}{p(y,s^2,\sigma^2)}
		\right\}
		d\mu
		=
		0,
	\end{equation}
	which rearranges to
	\begin{equation}
		\mathbb{E}_{G_0}[\mu \mid y,s,\sigma]
		=
		y + \frac{\sigma^2}{J}  \cdot \frac{1}{p(y,s^2 \mid \sigma^2)} \cdot \partial_y p(y,s^2 \mid \sigma^2).
		\label{eq:thm.main.proof.1}
	\end{equation}
	
	We then apply the law of iterated expectations to integrate out $\sigma^2$:
	\begin{equation}
		\begin{aligned}[b]
			\mathbb{E}_{G_0} [\mu \mid y,s]
			=&
			\mathbb{E}_{G_0} \left\{\mathbb{E}_{G_0} [\mu \mid y,s,\sigma] \mid y,s\right\}
			\\
			=&_{(1)}
			y
			+
			\frac{1}{J}
			\cdot
			\frac{1}{p(y,s^2)}
			\int
			\sigma^2 \cdot \partial_y p(y,s^2 , \sigma^2) d\sigma^2 \\
			=&_{(2)}
			y
			+
			\frac{1}{J}
			\cdot
			\frac{1}{p(y,s^2)}
			\partial_y
			\left\{p(y,s^2) \mathbb{E}_{G_0} [\sigma^2 \mid y,s]\right\}
			\label{eq:alt.mu}
			\\
			=&_{(3)}
			y
			+
			\frac{\mathbb{E}_{G_0} [\sigma^2\mid y,s]}{J} \cdot \partial_y \log p(y \mid s)
			+
			\frac{1}{J}\partial_y \mathbb{E}_{G_0} [\sigma^2\mid y,s]
		\end{aligned}
	\end{equation}
	where $=_{(1)}$ substitutes \eqref{eq:thm.main.proof.1}, rewrites $p(\sigma^2 \mid y,s^2) = p(y,s^2,\sigma^2)/p(y,s^2)$ and cancels common factors; $=_{(2)}$ interchanges the order of integration and differentiation; and $=_{(3)}$ applies the product rule.
	Note that the second-to-last line of \eqref{eq:alt.mu}, together with the expression of $\mathbb{E}_{G_0} [\sigma^2 \mid y,s]$ from Theorem~\ref{thm:main} (proved below), implies an alternative representation of $\mathbb{E}_{G_0} [\mu \mid y,s]$:
	\begin{equation}
		\mathbb{E}_{G_0} [\mu \mid y,s]
		=
		y+ 
		\frac{k}{J} \cdot \frac{1}{p(y,s^2)}
		\int_{s^2}^\infty \left[\frac{s^2}{t}\right]^{k-1} \partial_y p(y,t) dt.
		\label{eq:tweed.alt}
	\end{equation}
	We will next prove Theorem~\ref{thm:main}(ii) and (iii), which again proceeds in two steps. The first step which we provide in Lemma~\ref{lm:mgftau}, involves deriving the posterior moment generating function (MGF) of the precision parameter $\tau:=\sigma^{-2}$ and the proof follows \cite{Banerjee2023}. The second involves an application of Proposition 6 of \cite{Cressie1986} to the posterior MGF of $\tau$, which allows us to recover the moments of $\sigma$ and $\sigma^2$ from its MGF: e.g. 
	$\mathbb{E}_{G_0}[\sigma^2 \mid y,s] = \int_0^\infty M_{\tau | y,s}(-t) dt$.
	The application requires us to verify the following integrability statements: 
	$ \int_0^\infty M_{\tau \mid y,s} (-t) dt < \infty $ and 
	$ \int_0^\infty \left[\frac{1}{t}\right]^{1/2} M_{\tau \mid y,s} (-t) dt \;< \infty$.
	The second condition is more tedious and so we verify it here:
	\begin{equation}
		\begin{aligned}[b]
			&\int_0^\infty \left[\tfrac{1}{t}\right]^{1/2} M_{\tau \mid y,s} (-t) dt
			\\
			\lesssim&_{(1)}
			\int_0^\infty
			\tfrac{1}{t^{1/2}}
			\left[\tfrac{1}{s^2+tk^{-1}}\right]^{k-1}
			\int \phi(y\mid \mu,\sigma) \gamma(s^2+tk^{-1}\mid \sigma)
			dG(\mu,\sigma)
			dt
			\\
			\lesssim&_{(2)}
			\int_0^\infty
			\tfrac{1}{t^{1/2}}
			\left[\tfrac{1}{s^2+tk^{-1}}\right]^{k-1}
			dt < \infty
		\end{aligned}
	\end{equation}
	where $\lesssim_{(1)}$ follows by Lemma~\ref{lm:mgftau}; and $\lesssim_{(2)}$ follows by substituting the sampling densities and noting that $z(\sigma) := \sigma^{-(2k+1)}\exp(-ks^2/\sigma^2)$ is continuous on $(0,\infty)$, vanishes as $\sigma\downarrow 0$ and $\sigma\uparrow\infty$, and is therefore uniformly bounded.
\end{proof}

\begin{lemma}[Posterior MGF of $\tau=\sigma^{-2}$]
	\label{lm:mgftau}
	Suppose Assumption~\ref{ass:disagg.model} and $J > 3$. Then 
	\begin{equation}
		M_{\tau \mid y,s}(t) 
		= 
		\left[\frac{s^2}{s^2 - tk^{-1}}\right]^{k-1}\frac{p(y,s^2 - tk^{-1})}{p(y,s^2)}
		.
	\end{equation}
\end{lemma}
\begin{proof}[Proof of Lemma \ref{lm:mgftau}]
	By Bayes' rule and re-arrangement,
	\begin{equation}
		p(\mu,\tau \mid y,s^2)
		\propto
		\text{exp}\left\{ \langle\bs{\eta} \; , \; [\tau\mu,\tau]^\prime \rangle - A(\bs{\eta})\right\}
	\end{equation}
	with natural parameters $\bs{\eta} := \left[Jy \ , \ -\tfrac{1}{2}(Jy^2 + 2ks^2)\right]^\prime$ and log-partition $A(\bs{\eta}) := -(k-1)\log s^2 + \log p(y,s^2)$. Hence $p(\mu,\tau \mid y,s^2)$ is a member of the exponential family with sufficient statistics $(\tau\mu,\tau)$, and the marginal posterior MGF of $\tau$ is
	\begin{equation}
		M_{\tau \mid y,s}(t) =
		\text{exp} \left\{ A(\bs{\eta} + [0,t]^\prime) - A(\bs{\eta})  \right\}
	\end{equation}
	for $t\in\mathbb{R}$. Substituting and simplifying yields the expression of this lemma.
\end{proof}

\subsection{Proofs of Lemmas~\ref{lem:missp.hom} and \ref{lem:missp.het}}
\label{subsec:app.missp}
	
\begin{proof}[Proof of Lemma~\ref{lem:missp.hom}]
	We drop the $i$ indices for compactness. With regularity conditions permitting the exchange of integration and differentiation, the chain rule and Bayes' rule give $\partial_y \log f_{p,G_0}(y) = \mathbb{E}_{p,G_0}[\partial_y \log f_{p}(y\mid\mu) \mid y]$. Adding and subtracting $\tfrac{1}{\sigma_0^2/J}(-y+\mu)$ inside the expectation and splitting,
	\begin{equation}
	\begin{aligned}[b]
		\partial_y \log f_{p,G_0}(y)
		=&
		\mathbb{E}_{p,G_0}[ \underbrace{\tfrac{\sqrt{J}}{\sigma_0}(\partial_z\log p_z(z) + z)}_{=\Delta_0}  \mid y ] +
		\mathbb{E}_{p,G_0}[ \tfrac{1}{\sigma_0^2/J} (-y+\mu)  \mid y ]
		\\
		=&
		\mathbb{E}_{p,G_0}[ \Delta_0 \mid y ] +\tfrac{1}{\sigma_0^2/J} (-y+\mathbb{E}_{p,G_0}[\mu \mid y]).
	\end{aligned}
	\end{equation}
	Rearranging, then adding and subtracting $\tfrac{\sigma_0^2}{J}\partial_y \log f_{\phi,B_0}(y)$ to recognize the Tweedie representation yields
	$
		\mathbb{E}_{p,G_0}[\mu \mid y]
		=
		\mathbb{E}_{\phi,B_0}[\mu \mid y]
		+
		\tfrac{\sigma_0^2}{J}\nabla_0
		-
		\tfrac{\sigma_0^2}{J}\mathbb{E}_{p,G_0}[ \Delta_0 \mid y ].
	$
\end{proof}

\begin{proof}[Proof of Lemma~\ref{lem:missp.het}]
		We establish three equalities in turn. 
		First, we can show as per the previous lemma that conditional on $\sigma$, 
		\begin{equation}
			\partial_y \log f_{p}(y,s^2 \mid \sigma ) 
			=
			\tfrac{1}{\sigma^2/J} (-y+\mathbb{E}_{p,G_0}[\mu\mid y,s,\sigma])  + \mathbb{E}_{p,G_0}[\Delta_1\mid y,s,\sigma]
		\end{equation}
		% where $\Delta_1 = \partial_y \log f_{p}(y \mid \mu,\sigma,s ) - \tfrac{1}{\sigma^2/J}[-y + \mu]$,
		and then by integrating out $\sigma$ to arrive at the first equality: 
		\begin{equation}
			\mathbb{E}_{p,G_0} [ y + \tfrac{\sigma^2}{J}\partial_y \log f_{p,G_0}(y,s^2 \mid \sigma )  \mid y,s ]
			=
			\mathbb{E}_{p,G_0} [ \mu \mid y,s ] + \mathbb{E}_{p,G_0}[ \tfrac{\sigma^2}{J} \Delta_1 \mid y,s ].
			\label{eq:missp.het.proof.1}
		\end{equation}
		The second: by the same manipulation as in the proof of Theorem~\ref{thm:main} (see the 2nd to 3rd line of \eqref{eq:alt.mu}, replacing ``$p$'' in that context with $f_{p,G_0}$ throughout),
		\begin{equation}
			\mathbb{E}_{p,G_0} [ y + \tfrac{\sigma^2}{J}\partial_y \log f_{p,G_0}(y,s^2 \mid \sigma )  \mid y,s ]
			=
			y + \tfrac{1}{J}\tfrac{1}{f_{p,G_0}(y,s^2)}\partial_y \{ f_{p,G_0}(y,s^2)\mathbb{E}_{p,G_0}[\sigma^2 \mid y,s] \},
			\label{eq:missp.het.proof.2}
		\end{equation}
		where $\mathbb{E}_{p,G_0}[\sigma^2 \mid y,s] = \int_0^\infty M_{\sigma^{-2} | y,s}(-t) dt$ recovers the posterior mean of $\sigma^2$ from the posterior MGF of the precision; again see the proof of Theorem~\ref{thm:main}. Next, by definition,
		\begin{equation}
			M_{\sigma^{-2} | y,s}(-t)
			=
			\frac{1}{f_{p,G_0}(y,s^2)}\int e^{-t/\sigma^2} f_{p}(y,s^2\mid \mu,\sigma) dG(\mu,\sigma).
		\end{equation}
		Adding and subtracting $\left[\tfrac{s^2}{s^2+tk^{-1}}\right]^{k-1} f_{p}(y,s^2+\tfrac{t}{k} \mid \mu,\sigma)$ inside the integral and applying Bayes' rule yields the third and final equality:
		\begin{equation}
			M_{\sigma^{-2} | y,s}(-t)
			=
			\left[\frac{s^2}{s^2+tk^{-1}}\right]^{k-1} \frac{f_{p,G_0}(y,s^2+\tfrac{t}{k})}{f_{p,G_0}(y,s^2)}
			+
			\mathbb{E}_{p,G_0}[ \tilde{\Delta}_2(t) \mid y,s],
			\label{eq:missp.het.proof.3}
		\end{equation}
		where
		\begin{equation}
			\tilde{\Delta}_2(t) :=
			e^{-t/\sigma^2} \left\{
				1-
				\frac{r_{z,w}(z,w+\delta_t)}{r_{z,w}(z,w)}
			\right\}
		\end{equation}
		by definition of $r_{z,w}(\cdot)$ in Lemma~\ref{lem:missp.het} and that under microdata normality, $(y,s^2)$ are independent with $s^2$ having a Gamma distribution with shape $k$ and scale $\sigma^2/k$.

		Now combine \eqref{eq:missp.het.proof.1} and \eqref{eq:missp.het.proof.2} to get the first equality of the display below, then substitute in \eqref{eq:missp.het.proof.3} to get the second: 
		\begin{align}
			&\mathbb{E}_{p,G_0} [ \mu \mid y,s ] + \mathbb{E}_{p,G_0} [ \tfrac{\sigma^2}{J} \Delta_1 \mid y,s ]
			\nonumber\\=&
			y+ \frac{1}{J}\frac{1}{f_{p,G_0}(y,s^2)}\partial_y \{ f_{p,G_0}(y,s^2)\int_0^\infty M_{\sigma^{-2} | y,s}(-t) dt \}
			\nonumber\\=&
			y+
			\frac{1}{J}\frac{1}{f_{p,G_0}(y,s^2)}\partial_y \left\{ f_{p,G_0}(y,s^2)\int_0^\infty \left[\tfrac{s^2}{s^2+tk^{-1}}\right]^{k-1} \tfrac{f_{p,G_0}(y,s^2+tk^{-1})}{f_{p,G_0}(y,s^2)}  dt \right\}  
			\nonumber\\+&
			\frac{1}{J}\frac{1}{f_{p,G_0}(y,s^2)}\partial_y \left\{ f_{p,G_0}(y,s^2)\int_0^\infty \mathbb{E}_{p,G_0}[ \tilde{\Delta}_2(t) \mid y,s] dt \right\}  
			\nonumber\\=_{(1)}&
			\mathbb{E}_{\phi,B_0}[\mu\mid y,s] + \tfrac{1}{J}\nabla_1
			\nonumber\\+& \frac{1}{J}\frac{1}{f_{p,G_0}(y,s^2)}\partial_y \left\{ f_{p,G_0}(y,s^2)\mathbb{E}_{p,G_0}\left[ \Delta_2 \mid y,s \right] \right\},
			\label{eq:missp.het.1}
		\end{align}
		where 
		$\Delta_2 = \int_0^\infty \tilde{\Delta}_2(t) dt$, and
		$=_{(1)}$ follows from the representation in \eqref{eq:tweed.alt}: 
		\begin{equation}
			\mathbb{E}_{\phi,B_0}[\mu\mid y,s] = 
			y+
			\tfrac{1}{J}\tfrac{1}{f_{\phi,B_0}(y,s^2)} \int_0^\infty \left[\tfrac{s^2}{s^2+tk^{-1}}\right]^{k-1} \partial_y  f_{\phi,B_0}(y,s^2+tk^{-1})  dt. 
		\end{equation}
		Finally, taking the derivative of the final term of \eqref{eq:missp.het.1} using the product rule yields
		\begin{equation}
		\begin{aligned}[b]
			&\frac{1}{J}\frac{1}{f_{p,G_0}(y,s^2)}
			\int \partial_y
			\left[
				\Delta_2 \cdot 
				f_{p,G_0}(y,s^2) \tfrac{f_p(y,s^2 | \mu,\sigma)}{f_{p,G_0}(y,s^2) }g(\mu,\sigma)
			\right]	
			d(\mu,\sigma) 
			\\
			=&
			\frac{1}{J} \int [\partial_y \Delta_2] dG(\mu,\sigma | y,s) 
			+ 
			\frac{1}{J}\int [\Delta_2 \cdot \partial_y \log f_p(y,s^2 | \mu,\sigma) ] dG(\mu,\sigma | y,s)
			\\
			=&
			\tfrac{1}{J} \mathbb{E}_{p,G_0}[ \tfrac{\sqrt{J}}{\sigma}  \partial_z \Delta_2 \mid y,s ]
			+ 
			\tfrac{1}{J} \mathbb{E}_{p,G_0}[\Delta_2 \cdot \tfrac{\sqrt{J}}{\sigma} \partial_z\log p_{z|w}(z|w) \mid y,s ].
		\end{aligned} 
		\end{equation}
\end{proof}

        \section{Proof of Theorem~\ref{thm:eb.opt}}
\label{sec:app.reg}
\subsection{Notations and Definitions}

\textit{Model}. Recall the premise of Theorem~\ref{thm:eb.opt} is Assumption~\ref{ass:disagg.model}, i.e. the model $y_{ij}\mid\mu_i,\sigma_i\sim_{\text{iid}}\mathcal{N}[\mu_i,\sigma_i^2]$ and $(\mu_i,\sigma_i)\sim_{\text{iid}}G_0$. We maintain this throughout Appendix~\ref{sec:app.reg}.

% Collect all of $y_{ij}$ for $j=1,\dots,J$ into $\mathcal{Y}_i$, and all of $\mathcal{Y}_i$ for $i=1,\dots,n$ into $\mathcal{Y}$. Further let $(y_i,s_i^2)$ denote the sample mean and variance from $\mathcal{Y}_i$. Let $\mathcal{Y}_{-i}$ denote $\mathcal{Y}$ excluding $\mathcal{Y}_i$.
% We use $\phi(\cdot)$ to denote the standard normal density, and $\gamma(\cdot|k,\theta)$ to denote the density of a gamma random variable with shape $k$ and scale $\theta$. 

% $\mathbb{E}_{G_0}^{\mathcal{Y}}[\cdot]$ indicates integrating out $\mathcal{Y}$ according to the model, in particular the subscript indicating that $(\mu_i,\sigma_i)\sim G_0$. 

\noindent\textit{Truncated Distribution Class}. For constants $(\epsilon,v_1,v_2,C)$ we define a general class of distributions 
$\mathcal{G}(\epsilon,v_1,v_2,C,\underline{\sigma})$ where
\begin{equation}
\begin{aligned}
	\mathcal{G}(\epsilon,v_1,v_2,L,\underline{\sigma}) :=& \{ G: G\big([-C(\log \tfrac{1}{\epsilon})^{v_1},L(\log \tfrac{1}{\epsilon})^{v_1}]\times [\underline{\sigma},L(\log \tfrac{1}{\epsilon})^{v_2}]\big) = 1 \},
	\label{eq:def.trunc}
\end{aligned}
\end{equation}
and in particular
\begin{equation}
	\mathcal{G}_n := \mathcal{G}(\tfrac{1}{n},\gamma_1,\gamma_2,L,\underline{\sigma})
	\label{eq:def.trunc.1}
\end{equation}
where $(\gamma_1,\gamma_2)$ are as defined in Assumption~\ref{ass:subexp}, and $L$ is a large enough constant satisfying both $-\lambda_1L^{1/\gamma_1}\leq-4$ and $-\lambda_2L^{1/\gamma_2}\leq-4$.

Let $\bar{\mu}_G$ and $\bar{\sigma}_G$ denote the maximum supports of $\mu$ and $\sigma$ for a distribution $G$. 
Let $\bar{\mu}_n$ and $\bar{\sigma}_n$ to denote the maximum support of all distributions in $\mathcal{G}_n$; thus,
\begin{equation}
	\bar{\mu}_n\asymp (\log n)^{\gamma_1} \text{ and } \bar{\sigma}_n\asymp (\log n)^{\gamma_2}.
\end{equation}
Let $G_{0,n}$ denote the data-generating distribution $G_0$ conditional on the event
$(\mu,\sigma)\in[-\bar{\mu}_n,\bar{\mu}_n]\times[\underline{\sigma},\bar{\sigma}_n]$. Thus $G_{0,n}$ is a member of $\mathcal{G}_n$.

\noindent\textit{Estimation}. Let $\hat{G}$ denote an estimator of $G_0$ satisfying Assumption~\ref{ass:npmle}.
Let $\hat{\mu}_{G_0}(y_i,s_i^2)$ be the Bayes estimator of $\mu_i$, and $\hat{\mu}_{\hat{G}}(y_i,s_i^2)$ the feasible version. Note that for the feasible version, $(y_i,s_i^2)$ appears twice -- in $\hat{G}$ and also as the argument of $\hat{\mu}_{\hat{G}}(y_i,s_i^2)$.

Throughout, we use a alternative representation of $\hat{\mu}_G$; see \eqref{eq:tweed.alt}: 
\begin{equation}
	\hat{\mu}_G(y,s^2)
	=
	y + \frac{k}{J}\frac{1}{f_G(y,s^2)} \int_{s^2}^\infty \left[\frac{s^2}{t}\right]^{k-1} \partial_y f_{G}(y,t) dt. 
\end{equation}

\noindent\textit{Additional Definitions \& Notations}. Let $M\geq 2$, $\rho_n := 1/n^3$ and 
\begin{align}
	\bar{y}_n :=& \bar{\mu}_n + \bar{\sigma}_n \cdot (\log n^{2M})^{1/2} &\asymp (\log n)^{\gamma_1\vee [\frac{1}{2} + \gamma_2]} &=: (\log n)^{\tilde{\gamma}_1} \nonumber\\
	\bar{s}_n :=& \bar{\sigma}_n \cdot (\log n^{M/k})^{1/2} &\asymp (\log n)^{\frac{1}{2}+\gamma_2} &=: (\log n)^{\tilde{\gamma}_2} \nonumber\\
	\varepsilon_n :=& n^{-1/2}(\log n)^{2[\tilde{\gamma}_1\vee\frac{1}{2} + \tilde{\gamma}_2] + \frac{1}{2} } &\asymp (\log n)^{ 2[\tilde{\gamma}_1 +\tilde{\gamma}_2] + \frac{1}{2}} \nonumber\\
	%\mathcal{L}:=& \{ \vee_{i\leq n} |\mu_i| \leq \bar{\mu}_n \text{ and } \vee_{i\leq n} \sigma_i \leq \bar{\sigma}_n \}\nonumber\\\
	\mathcal{D}:=& \{ \vee_{i\leq n} |y_i| \leq \bar{y}_n \text{ and } \vee_{i\leq n} s_i \leq \bar{s}_n \} \nonumber\\
	\mathcal{A} :=& \{ d( f_{\hat{G}} , f_{G_0} ) < \varepsilon_n \}
	\label{eq:def}
\end{align}
where $d(\cdot,\cdot)$ is the Hellinger distance.
Note that $\mathcal{D}$ (and $\mathcal{A}$) is also used as an indicator function for the event.
Let $C$ denote a generic constant, and $v_n$ denote a generic sub-polynomial sequence. That is, $v_n$ satisfies $v_n=o(n^{\epsilon})$ for every $\epsilon>0$.

\subsection{Proof of Theorem~\ref{thm:eb.opt}}

We give the proof for $\{\mu_i\}_{i\in[n]}$, then indicate the modifications for $\{\sigma_i\}_{i\in[n]}$ and $\{\sigma_i^2\}_{i\in[n]}$; several steps below are established formally in the Supplementary Appendix.

\begin{proof}[Regret bounds for $\{\mu_i\}_{i\in[n]}$]
	Since the data and latent variables are i.i.d. across $i$, the regret of $\{ \hat{\mu}_{\hat{G}}(y_i,s_i^2) \}_{i=1}^n$ relative to the oracles $\{ \hat{\mu}_{G_0}(y_i,s_i^2)\}_{i=1}^n$ reduces to 
	\begin{equation}
	\begin{aligned}[b]
		&\mathbb{E}_{G_0}^{(\mu_1,\mathcal{Y})}[\hat{\mu}_{\hat{G}}(y_1,s_1^2) - \mu_1 ]^2 - 
		\mathbb{E}_{G_0}^{(\mu_1,\mathcal{Y}_1)}[\hat{\mu}_{G_0}(y_1,s_1^2) - \mu_1 ]^2
		\\=&
		\mathbb{E}_{G_0}^{\mathcal{Y}}\Delta_{\mu,\hat{G},G_0}(y_1,s_1)
		\\\lesssim_{(1)}&
		\mathbb{E}_{G_{0,n}}^{\mathcal{Y}}\Delta_{\mu,\hat{G},G_{0,n}}(y_1,s_1) + o(1/n)
		\\=&
		\mathbb{E}_{G_{0,n}}^{\mathcal{Y}}\left[\Delta_{\mu,\hat{G},G_{0,n}}(y_1,s_1)[\mathcal{D}+\mathcal{D}^c][\mathcal{A}+\mathcal{A}^c]\right] + o(1/n)
	\end{aligned}
	\end{equation}
	where 
	$
	\Delta_{\mu,H,G}(y,s) := [\hat{\mu}_{H}(y,s^2) - \hat{\mu}_{G}(y,s^2)]^2
	$ and 
	$\lesssim_{(1)}$ follows by Lemma~\ref{lem:truncation} using sub-exponentiality of $G_0$ in Assumption~\ref{ass:subexp}.
	Next, as a result of Assumption~\ref{ass:subexp} and the super-smoothness of the mixture density induced by Assumption~\ref{ass:disagg.model}, the regret in regions $\mathcal{D}^c$ and $\mathcal{A}^c$ is shown by Lemma~\ref{lem:complement} to be $o(1/n)$. Thus the regret will be determined by $\Delta_{\mu,\hat{G},G}(y_1,s_1)$ in region $\mathcal{D}\cap\mathcal{A}$. 
	From now on, we will work in the region $\mathcal{D}\cap\mathcal{A}$ and suppress this notation.
	
	Note that for any $H\in\mathcal{G}_n$, 
	\begin{equation}
		\mathbb{E}_{G_{0,n}}^{\mathcal{Y}}\Delta_{\mu,\hat{G},G_{0,n}}(y_1,s_1) 
		\leq 
		2\left\{ \mathbb{E}_{G_{0,n}}^{\mathcal{Y}}\Delta_{\mu,\hat{G},H}(y_1,s_1)  +
		\mathbb{E}_{G_{0,n}}^{\mathcal{Y}_1}\Delta_{\mu,H,G_{0,n}}(y_1,s_1) \right\}.
		\label{eq:reg.mu.cp}
	\end{equation}
	Now take $\{\hat{\mu}_{H_j}\}_{j\leq N}$ to be an $\delta_n$-net of $\mathcal{T}_\mu:=\{\hat{\mu}_H:d(f_H,f_{G_{0,n}})<\varepsilon_n \text{ and } H\in\mathcal{G}_n\}$ in the sup-norm\footnote{I.e., for every $\hat{\mu}_H\in\mathcal{T}_\mu$, there is at least one $j\leq N$ such that $\sup_{y,s^2}|\hat{\mu}_H(y,s^2) - \hat{\mu}_{H_j}(y,s^2)| < \delta_n$.}, where $N$ denotes the cardinality of the net which depends implicitly on $\delta_n$. We assume WLOG\footnote{E.g., first find a $\delta_n$-net for $\mathcal{T}_\mu$, then for each member of the net, take an element within $\mathcal{T}_\mu$ that is less than $\delta_n$ from it; these elements then form a $2\delta_n$-net of $\mathcal{T}_\mu$ and also by definition each will be less than $\varepsilon_n$ away from $f_{G_{0,n}}$ in Hellinger distance.}  that every $H_j$ satisfies $d(f_{H_j},f_{G_{0,n}}) < \varepsilon_n$. Further note that this net is non-random. 
	Then
	\begin{align}
		\mathbb{E}_{G_{0,n}}^{\mathcal{Y}}\Delta_{\mu,\hat{G},G_{0,n}}(y_1,s_1)
		\leq_{(1)}& 
		2\left\{
		\min_{j\leq N} \mathbb{E}_{G_{0,n}}^{\mathcal{Y}}\Delta_{\mu,\hat{G},H_j}(y_1,s_1)
		+
		\max_{j\leq N} \mathbb{E}_{G_{0,n}}^{\mathcal{Y}_1}\Delta_{\mu,H_j,G_{0,n}}(y_1,s_1)
		\right\}
		\nonumber\\\leq_{(2)}& 
		2\left\{ \delta_n^2 + \max_{j\leq N} \mathbb{E}_{G_{0,n}}^{\mathcal{Y}_1}\Delta_{\mu,H_j,G_{0,n}}(y_1,s_1) \right\}
		\label{eq:reg.mu.net}
	\end{align}
	where $\leq_{(1)}$ holds because \eqref{eq:reg.mu.cp} is true for every $H_j$. Furthermore $\leq_{(2)}$ holds because $\hat{\mu}_{\hat{G}}\in\mathcal{T}_\mu$, since we are in the region $\mathcal{A}$, and thus $\hat{\mu}_{\hat{G}}$ has to be no more than $\delta_n$ (in sup-norm) away from a member of the net.  
	By taking $\delta_n\downarrow0$ fast enough
	% \footnote{As a side note, there is typically a balance needed between the rate of $\delta_n\downarrow0$ and the entropy of the $\delta_n$-net. This is not needed in our proof because by defining the regret in terms of the integrated risk as opposed to the compound frequentist risk, we do not need empirical process control over the estimator class.}
	, the expectation of $\Delta_{\mu,\hat{G},G_{0,n}}(y_1,s_1)$ will be determined by the second term in \eqref{eq:reg.mu.net}.
	
	For each $j$, we bound $\Delta_{\mu,H_j,G_{0,n}}(y_1,s_1)$ by its regularized counterpart:
	\begin{align}
		&\mathbb{E}_{G_{0,n}}^{\mathcal{Y}_1}\Delta_{\mu,H_j,G_{0,n}}(y_1,s_1)
		\nonumber\\\leq_{(1)}& 
		\mathbb{E}_{G_{0,n}}^{\mathcal{Y}_1}\Delta_{\mu,H_j,G_{0,n}}(y_1,s_1,\rho_n)
		+ 
		\mathbb{E}_{G_{0,n}}^{\mathcal{Y}_1} \left\{
		\left[\frac{f_{G_{0,n}}(y_1,s_1^2)}{f_{G_{0,n}}(y_1,s_1^2)\vee\rho_n}-1 \right]
		[\hat{\mu}_{G_{0,n}}(y_1,s_1^2)-y_1]
		\right\}^2
		\nonumber\\\leq_{(2)}&
		\mathbb{E}_{G_{0,n}}^{\mathcal{Y}_1}\Delta_{\mu,H_j,G_{0,n}}(y_1,s_1,\rho_n) 
		+ \bar{y}_n^2 \cdot \mathbb{E}_{G_{0,n}}^{\mathcal{Y}_1}
		\left[\frac{f_{G_{0,n}}(y_1,s_1^2)}{f_{G_{0,n}}(y_1,s_1^2)\vee\rho_n}-1 \right]^2
		\label{eq:reg.mu.trunc}
	\end{align}
	where 
	$\Delta_{\mu,H,G_{0,n}}(y,s,\rho)  
	:= \hat{\mu}_{H}(y,s;\rho) - \hat{\mu}_{G_{0,n}}(y,s;\rho)$ and
	\begin{equation}
		\hat{\mu}_{G}(y,s^2;\rho) 
		:=
		y + \frac{k}{J}\frac{1}{f_G(y,s^2)\vee\rho} \partial_y\left\{ \int_{s^2}^\infty \left[\frac{s^2}{t}\right]^{k-1} f_{G}(y,t) dt\right\},
	\end{equation}
	and $\leq_{(1)}$ of \eqref{eq:reg.mu.trunc} follows by Lemma~\ref{lem:regu.regret} and $\leq_{(2)}$ follows because we are in region $\mathcal{D}$. 
	Furthermore where $\rho_n < f_{G_{0,n}}$, $\frac{f_{G_{0,n}}(y_1,s_1^2)}{f_{G_{0,n}}(y_1,s_1^2)\vee\rho_n}-1=0$, and where $f_{G_{0,n}} < \rho_n$, 
	\begin{equation}
	\begin{aligned}[b]
		\mathbb{E}_{G_{0,n}}^{\mathcal{Y}_1}
		\left[\frac{f_{G_{0,n}}(y_1,s_1^2)}{f_{G_{0,n}}(y_1,s_1^2)\vee\rho_n}-1 \right]^2
		\leq_{(1)} &
		\int_{\mathcal{D}} 
		f_{G_{0,n}}(y_1,s_1^2) d(y_1,s_1^2 )
		\\\leq_{(2)} &
		\rho_n \cdot 
		\int_{\mathcal{D}} d(y_1,s_1^2 )
		\\= &
		\rho_n \cdot 2 \bar{y}_n\bar{s}_n^2
		\; =_{(3)}
		o(1/n^2)
	\end{aligned}
	\end{equation}
	where $\leq_{(1)}$ follows as $\frac{f_{G_{0,n}}(y_1,s_1^2)}{f_{G_{0,n}}(y_1,s_1^2)\vee\rho_n}-1$ is bounded in magnitude by 1 (also recalling that we are in $\mathcal{D}$), $\leq_{(2)}$ follows by $f_{G_{0,n}} < \rho_n$ and $=_{(3)}$ by definition of $\rho_n$. The terms in $=_{(3)}$ do not depend on $j$ and so by returning to \eqref{eq:reg.mu.trunc}, we have for every $j$
	\begin{equation}
		\mathbb{E}_{G_{0,n}}^{\mathcal{Y}_1}\Delta_{\mu,H_j,G_{0,n}}(y_1,s_1)
		\leq
		\mathbb{E}_{G_{0,n}}^{\mathcal{Y}_1}\Delta_{\mu,H_j,G_{0,n}}(y_1,s_1,\rho_n)
		+ 
		\underbrace{\bar y_n o(1/n^2)}_{=o(1/n)}.
		\label{eq:penulti}
	\end{equation}
	By\footnote{This entails comparing the rates of various objects within Theorem~\ref{thm:jz.p3.mu}; in particular note that $\tilde{L}(\rho_n)\asymp \log n$ is of a larger order than $|\log \varepsilon_n^2|$.} Theorem~\ref{thm:jz.p3.mu}, an extension of results in \cite{Jiang2009} to bound $\Delta_{\mu,H_j,G_{0,n}}(y_1,s_1,\rho_n)$ by the squared Hellinger distance of the mixture densities (which is less than $\varepsilon_n^2$ since we are in $\mathcal{A}$) up to log factors,
	\begin{equation}
	\begin{aligned}[b]
		\mathbb{E}_{G_{0,n}}^{\mathcal{Y}_1}\Delta_{\mu,H_j,G_{0,n}}(y_1,s_1,\rho_n)
		\leq&
		C\cdot
		\bar{s}_n^4  [ \bar{y}_n^2\log^{1/2}(n) + \log(n)] \cdot \varepsilon_n^2
		\\\leq&
		C\cdot \bar{s}_n^4\bar{y}_n^2\log^{1/2}(n) \cdot \varepsilon_n^2
	\end{aligned}
	\end{equation}
	and so returning to \eqref{eq:penulti}, we have for every $j$
	\begin{equation}
		\mathbb{E}_{G_{0,n}}^{\mathcal{Y}_1}\Delta_{\mu,H_j,G_{0,n}}(y_1,s_1)
		\leq 
		C \cdot \bar{s}_n^4\bar{y}_n^2\log^{1/2}(n) \cdot \varepsilon_n^2.
	\end{equation}
	
	Plugging the above into \eqref{eq:reg.mu.net}, we get that the regret is bounded by 
	\begin{equation}
		C \cdot \bar{s}_n^4\bar{y}_n^2\log^{1/2}(n) \cdot \varepsilon_n^2
		\asymp
		(\log n)^{ 4\tilde{\gamma}_2 + 2\tilde{\gamma}_1 + 1}\cdot \tfrac{1}{n}(\log n)^{4[\tilde{\gamma}_1+\tilde{\gamma}_2] + 1}
		\asymp
		\tfrac{1}{n}(\log n)^{ 6\tilde{\gamma}_1 + 8\tilde{\gamma}_2 + 2 }
	\end{equation}
	uniformly over $G_0\in \mathcal{G}$ because the constants above depend only on $\mathcal{G}$.
\end{proof}
\begin{proof}[Regret bounds for $\{\sigma_i^2\}_{i\in[n]}$ or $\{\sigma_i^2\}_{i\in[n]}$]
	The proof is similar to that for $\{\mu_i\}_{i\in[n]}$, except that we replace Theorem~\ref{thm:jz.p3.mu} with Theorem~\ref{thm:jz.p3.sigma} or \ref{thm:jz.p3.sigma2}. The regret bound will be  
	\begin{equation*}
		C \cdot  \bar{s}_n^4 \varepsilon_n^2
		\asymp 
		\tfrac{1}{n}(\log n)^{ 4\tilde{\gamma_2} + 4[\tilde{\gamma}_1+\tilde{\gamma}_2] + 1 }
		\asymp 
		\tfrac{1}{n}(\log n)^{ 4\tilde{\gamma_1} + 8\tilde{\gamma}_2 + 1 }. \qedhere
	\end{equation*}
\end{proof}
% \begin{proof}[Regret bounds for $\{\sigma_i\}_{i\in[n]}$]
% 	The proof is similar to that for $\{\mu_i\}_{i\in[n]}$, except that we replace Theorem~\ref{thm:jz.p3.mu} with Theorem~\ref{thm:jz.p3.sigma}. As a result, the regret bound will be  
% 	\begin{equation*}
% 		C \cdot  \bar{s}_n^4 \varepsilon_n^2
% 		\asymp 
% 		\tfrac{1}{n}(\log n)^{ 4\tilde{\gamma_2} + 4[\tilde{\gamma}_1+\tilde{\gamma}_2] + 1 }
% 		\asymp 
% 		\tfrac{1}{n}(\log n)^{ 4\tilde{\gamma_1} + 8\tilde{\gamma}_2 + 1 }. \qedhere
% 	\end{equation*}
% \end{proof}

        %%% ---------- Supplementary Appendix ----------
        \clearpage
        \renewcommand{\thepage}{S.\arabic{page}}
        \setcounter{page}{1}

        \thispagestyle{empty}

        \begin{center}
            \vspace*{2cm}
            {\large \bf Supplementary Appendix:\\[6pt]Large-Scale Estimation under Unknown Heteroskedasticity}\\[12pt]
            {Sheng Chao Ho}
        \end{center}

        \vspace{1cm}

        This Supplementary Appendix starts with section C (Sections A and B are in the paper's Appendix). It consists of the following sections:

        \begin{enumerate}[label=\Alph*., start=3, leftmargin=*, itemsep=4pt]
            \item Proofs of additional results in the paper
            \item Supporting results for Theorem~\ref{thm:eb.opt}
            \item Supplementary results for empirical illustration
        \end{enumerate}

        \newpage

        \setcounter{section}{2}
        \numberwithin{equation}{section}
        \renewcommand*\thetable{\thesection-\arabic{table}}
        \renewcommand*\thefigure{\thesection-\arabic{figure}}

\section{Proofs of additional results in the paper}
\label{sec:supp.add.proofs}

\begin{proof}[Proof of Lemma~\ref{lm:threshold.detect}]
	Define $d_i:= \{\hat{\mu}_i \leq c\}$ for a generic estimator $\hat{\mu}:\mathcal{Y}\mapsto\mathbb{R}^n$. Using the identity $\lvert \mu_i-c\rvert (\{\mu_i>c\} - \{\mu_i \leq c\}) \cdot d_i = (\mu_i-c) \cdot d_i$,
	\begin{equation}
		\begin{aligned}[b]
			&\argmin_{\hat{\mu}}
			\mathbb{E}_{G_0}^{(\mu,\sigma)}
			\mathbb{E}_{(\mu,\sigma)}^{\mathcal{Y}}
			\frac{1}{n}
			\sum_{i=1}^n
			\big[
			\lvert \mu_i-c \rvert
			\cdot
			\big(
			d_i
			\left\{\mu_i>c\right\}
			+
			(1-d_i)
			\left\{\mu_i<c\right\}
			\big)
			\big]
			\\
			=&
			\argmin_{\hat{\mu}}
			\mathbb{E}_{G_0}^{\mathcal{Y}}
			\frac{1}{n}
			\sum_{i=1}^n
			\left(\mathbb{E}_{G_0}[\mu_i\mid \mathcal{Y}_i]-c\right) \cdot d_i,
		\end{aligned}
	\end{equation}
	which is minimized by $d_i^* = \mathbf{1}\{\mathbb{E}_{G_0}[\mu_i\mid y_i,s_i]\leq c\}$.
\end{proof}

\begin{proof}[Proof of Lemma~\ref{lm:main.quant}]
	Follows by Theorem~\ref{thm:main} and the linearity of expectations.
\end{proof}

\begin{proof}[Proof of Lemma~\ref{lem:main.hetj}]
	Denote the conditional distribution of $(\mu_i,\sigma_i)$ given $J_i=J$ as $H_{(\mu_i,\sigma_i)|J}$. Clearly $\check{\mu}^*=\mathbb{E}_{H_0}[\mu_i\mid y_i,s_i,J_i]$. Applying the proof of Theorem~\ref{thm:main} pointwise in $J$ — valid since $\mathbb{E}_{H_0}\{J_i > 3\} =1$ by assumption — yields the result.
\end{proof}

\section{Supporting results for Theorem~\ref{thm:eb.opt}}
\label{sec:supp.thm.eb}

\subsection{Auxiliary Lemmas}
The following is used to reduces the relative regret of feasible versus oracle estimator relative to their regularized counterparts.
\begin{lemma}[Regularized Regret]
	\label{lem:regu.regret}
	For two distributions $G$ and $H$,
	\begin{align}
		\hat{\mu}_{H}(y,s^2 ) - \hat{\mu}_{G}(y,s^2)
		=&
		\hat{\mu}_{H}(y,s ; \rho_n)
		-
		\hat{\mu}_{G}(y,s ; \rho_n)
		+
		\left\{\tfrac{f_{G}(y,s^2)}{f_{G}(y,s^2)\vee\rho_n} - 1\right\}\left\{\hat{\mu}_{G}(y,s)-y\right\},
		\nonumber\\
		\hat{\sigma}_{H}(y,s^2 ) - \hat{\sigma}_{G_0}(y,s^2)		
		=&
		\hat{\sigma}_{H}(y,s ; \rho_n)
		-
		\hat{\sigma}_{G}(y,s ; \rho_n)
		+
		\left\{\tfrac{f_{G}(y,s^2)}{f_{G}(y,s^2)\vee\rho_n} - 1\right\}\hat{\sigma}_{G}(y,s),
		\nonumber\\
		\hat{\sigma}_{H}^2(y,s^2 ) - \hat{\sigma}_{G}^2(y,s^2)
		=&
		\hat{\sigma}_{H}^2(y,s ; \rho_n)
		-
		\hat{\sigma}_{G}^2(y,s ; \rho_n)
		+
		\left\{\tfrac{f_{G}(y,s^2)}{f_{G}(y,s^2)\vee\rho_n} - 1\right\}\hat{\sigma}_{G}^2(y,s)
		\nonumber.
	\end{align}
\end{lemma}
\begin{proof}
	Note that 
	\begin{align}
		\hat{\sigma}_{H}^2(y,s | \rho_n) - \hat{\sigma}_{G}^2(y,s)
		\nonumber
		=&
		\hat{\sigma}_{H}^2(y,s \mid\rho_n)
		- 
		\tfrac{k}{f_{G}(y,s^2) } \cdot \int_{s^2}^\infty \left[\tfrac{s^2}{t}\right]^{k-1} f_{G}(y,t) dt
		\nonumber\\
		=&
		\hat{\sigma}_{H}^2(y,s \mid\rho_n)
		- 
		\underbrace{\tfrac{k}{f_{G}(y,s^2) \vee \rho_n} \cdot \int_{s^2}^\infty \left[\tfrac{s^2}{t}\right]^{k-1} f_{G}(y,t) dt}_{=\hat{\sigma}_{G}^2(y,s \mid\rho_n)}
		\nonumber\\
		+&
		\left\{\tfrac{f_{G}(y,s^2)}{f_{G_0}(y,s^2)\vee\rho_n} - 1\right\} \underbrace{\tfrac{k}{f_{G}(y,s^2)} \int_{s^2}^\infty \left[\tfrac{s^2}{t}\right]^{k-1}f_{G}(y,t)dt}_{=\hat{\sigma}_{G}^2(y,s)}
		\label{eq:regret.1}
	\end{align}
	The statements for $\{\mu_i\}_{i\in[n]}$ and $\{\sigma_i\}_{i\in[n]}$ follow identically and are thus omitted.
\end{proof}

The following proves that the difference between the untruncated $G_0$ and the truncated $G_{0,n}$ is of order $o(1/n)$.
\begin{lemma}[Truncation]
	\label{lem:truncation}
	Suppose that Assumptions~\ref{ass:subexp}~and~\ref{ass:npmle} hold. Then
	\begin{align*}
		\mathbb{E}_{G_0}^{\mathcal{Y}}\Delta_{\mu,\hat{G},G_0}(y_1,s_1)
		\lesssim&
		\mathbb{E}_{G_{0,n}}^{\mathcal{Y}}\Delta_{\mu,\hat{G},G_{0,n}}(y_1,s_1) + o(1/n) 
		\\
		\mathbb{E}_{G_0}^{\mathcal{Y}}\Delta_{\sigma,\hat{G},G_0}(y_1,s_1)
		\lesssim&
		\mathbb{E}_{G_{0,n}}^{\mathcal{Y}}\Delta_{\sigma,\hat{G},G_{0,n}}(y_1,s_1) + o(1/n) 
		\\
		\mathbb{E}_{G_0}^{\mathcal{Y}}\Delta_{\sigma^2,\hat{G},G_0}(y_1,s_1)
		\lesssim&
		\mathbb{E}_{G_{0,n}}^{\mathcal{Y}}\Delta_{\sigma^2,\hat{G},G_{0,n}}(y_1,s_1) + o(1/n) 
	\end{align*}
\end{lemma}
\begin{proof}
	We prove only the first statement, with the latter two following along identical lines. 
	First by the $C_p$-inequality, 
	\begin{align}
		\mathbb{E}_{G_0}^{\mathcal{Y}}\Delta_{\mu,\hat{G},G_0}(y_1,s_1)
		\leq
		2\left\{
			\mathbb{E}_{G_0}^{\mathcal{Y}}\Delta_{\mu,\hat{G},G_{0,n}}(y_1,s_1)
			+
			\mathbb{E}_{G_0}^{\mathcal{Y}_1}\Delta_{\mu,G_{0,n},G_{0}}(y_1,s_1)
		\right\}
		\label{eq:trunc.1}
	\end{align}
	We take the two terms on the RHS in turn. 
	
	Define $\mathcal{L} := \{ \vee_{i\leq n} |\mu_i| \leq \bar{\mu}_n \cap \vee_{i\leq n} \sigma_i \leq \bar{\sigma}_n \}$ and
	we split the first term of \eqref{eq:trunc.1} into two regions, $\mathcal{L}$ and $\mathcal{L}^c$. Then 
	\begin{equation}
	\begin{aligned}[b]
		\mathbb{E}_{G_0}^{\mathcal{Y}}\Delta_{\mu,\hat{G},G_{0,n}}(y_1,s_1)\mathcal{L}
		=&
		\mathbb{E}_{G_0}^{\mathcal{Y}}[\Delta_{\mu,\hat{G},G_{0,n}}(y_1,s_1) \mid \mathcal{L}]
		\cdot \mathbb{E}_{G_0}\mathcal{L}\\
		=_{(1)}&
		\mathbb{E}_{G_{0,n}}^{\mathcal{Y}}\Delta_{\mu,\hat{G},G_{0,n}}(y_1,s_1)
		\cdot \mathbb{E}_{G_0}\mathcal{L}\\
		\leq&
		\mathbb{E}_{G_{0,n}}^{\mathcal{Y}}\Delta_{\mu,\hat{G},G_{0,n}}(y_1,s_1)
	\end{aligned}
	\end{equation}
	where $=_{(1)}$ follows by definitions of $G_{0,n}$ and $\mathcal{L}$. 
	Next
	\begin{equation}
	\begin{aligned}[b]
		\mathbb{E}_{G_0}^{\mathcal{Y}}\Delta_{\mu,\hat{G},G_{0,n}}(y_1,s_1)\mathcal{L}^c
		\leq&
		\left\{ 
			\mathbb{E}_{G_0}^{\mathcal{Y}}\Delta_{\mu,\hat{G},G_{0,n}}^2(y_1,s_1)
		\right\}^{1/2}
		\left\{ 
			\mathbb{E}_{G_0}^{\mathcal{Y}}\mathcal{L}^c
		\right\}^{1/2}\\
		\leq&
		\left\{ 
			2\mathbb{E}_{G_0}^{\mathcal{Y}}[\hat{\mu}_{\hat{G}}^2(y_1,s_1)+\hat{\mu}_{G_{0,n}}^2(y_1,s_1)]
		\right\}^{1/2} \left\{ 
			\mathbb{E}_{G_0}^{\mathcal{Y}}\mathcal{L}^c
		\right\}^{1/2}\\
		=_{(1)}& v_n \cdot o(1/n^{3/2})
		\;= o(1/n)
	\end{aligned}
	\end{equation}
	where $=_{(1)}$ follows because the supports of the Bayes' estimators are bound by the supports of $\hat{G}$ and $G_{0,n}$; this implies $\mathbb{E}_{G_0}^{\mathcal{Y}}\hat{\mu}_{\hat{G}}^2(y_1,s_1)$ is sub-polynomial by Assumptions~\ref{ass:subexp} and \ref{ass:npmle}, and $\mathbb{E}_{G_0}^{\mathcal{Y}_1}\hat{\mu}_{G_{0,n}}^2(y_1,s_1)$ is also sub-polynomial by definition of $G_{0,n}$ and Assumption~\ref{ass:subexp}. 
	Furthermore, $\mathbb{E}_{G_0}^{\mathcal{Y}}\mathcal{L}^c = o(1/n^3)$ by Assumption~\ref{ass:subexp} and definitions of $\bar{\mu}_n$ and $\bar{\sigma}_n$. 
	Combining the previous two displays, we have that the first term of \eqref{eq:trunc.1}
	\begin{equation}
		\mathbb{E}_{G_0}^{\mathcal{Y}}\Delta_{\mu,\hat{G},G_{0,n}}(y_1,s_1)
		\leq
		\mathbb{E}_{G_{0,n}}^{\mathcal{Y}}\Delta_{\mu,\hat{G},G_{0,n}}(y_1,s_1) + o(1/n).
		\label{eq:trunc.1.1}
	\end{equation}

	We now show that the second term of \eqref{eq:trunc.1} is $o(1/n)$. Define 
	\begin{equation}
	\begin{aligned}[b]
		E :=& \{|\mu| \leq \bar{\mu}_n \cap \sigma \leq \bar{\sigma}_n\} \\
		\tau :=& \mathbb{E}_{G_0} E^c  \\
		G_{0,n}^c :=& G_0( \cdot \mid E^c) \\ 
		\omega(y,s) :=& \frac{\tau f_{G_{0,n}^c}(y,s^2)}{f_{G_0}(y,s^2)}
	\end{aligned}
	\end{equation}
	As a result we have 
	\begin{equation}
	\begin{aligned}[b]
		G_{0,n} =& G_0( \cdot \mid E) \\ 
		G_0 =& (1-\tau) G_{0,n} + \tau G_{0,n}^c \\
		f_{G_0} =& (1-\tau) f_{G_{0,n}} + \tau f_{G_{0,n}^c} \\
		\mathbb{E}_{G_0}[\mu\mid y,s] =& (1-\omega)  \mathbb{E}_{G_{0,n}}[\mu\mid y,s] + \omega \mathbb{E}_{G_{0,n}^c}[\mu\mid y,s]
	\end{aligned}
	\end{equation}
	From the final equality above, we have 
	\begin{align}
		\mathbb{E}_{G_0}[\mu\mid y,s] - \mathbb{E}_{G_{0,n}}[\mu\mid y,s] 
		=&
		\omega ( \mathbb{E}_{G_{0,n}^c}[\mu\mid y,s] - \mathbb{E}_{G_{0,n}}[\mu\mid y,s] ) \nonumber\\
		\implies 
		\left\{ \mathbb{E}_{G_0}[\mu\mid y,s] - \mathbb{E}_{G_{0,n}}[\mu\mid y,s]  \right\}^2
		\leq& 2 \omega^2 \left\{  (\mathbb{E}_{G_{0,n}^c}[\mu\mid y,s])^2 + (\mathbb{E}_{G_{0,n}}[\mu\mid y,s])^2 \right\} \nonumber\\
		\leq& 
		2 \omega^2  \left\{ \mathbb{E}_{G_{0,n}^c}[\mu^2\mid y,s] + \mathbb{E}_{G_{0,n}}[\mu^2\mid y,s] \right\} \nonumber\\
		=& 
		2 \omega \cdot \omega  \left\{ \mathbb{E}_{G_{0,n}^c}[\mu^2\mid y,s] - \mathbb{E}_{G_{0,n}}[\mu^2\mid y,s] + 2 \mathbb{E}_{G_{0,n}}[\mu^2\mid y,s]  \right\} \nonumber\\
		=& 
		2 \omega  \left\{ \mathbb{E}_{G_0}[\mu^2\mid y,s] - \mathbb{E}_{G_{0,n}}[\mu^2\mid y,s]  + \omega \cdot 2 \mathbb{E}_{G_{0,n}}[\mu^2\mid y,s]  \right\} \label{eq:trunc.2} 
	\end{align}
	Since $\mathbb{E}_{G_{0,n}}[\mu^2\mid y,s]\asymp v_n$ for a sub-polynomial $v_n$ (by definition of $G_{0,n}$), then applying Cauchy-Schwarz to \eqref{eq:trunc.2} yields
	\begin{equation}
	\begin{aligned}[b]
		\underbrace{\mathbb{E}_{G_0}\left\{ \mathbb{E}_{G_0}[\mu\mid y,s] - \mathbb{E}_{G_{0,n}}[\mu\mid y,s]  \right\}^2}_{=\mathbb{E}_{G_0}\Delta_{\mu,G_{0,n},G_{0}}(y,s)}
		\lesssim&
		[\mathbb{E}_{G_0}\omega^2]^{1/2} \cdot 
		\left[\mathbb{E}_{G_0} \left\{ (\mathbb{E}_{G_0}[\mu^2\mid y,s])  + v_n\right\}^2\right]^{1/2} \\
		\lesssim&
		[\mathbb{E}_{G_0}\omega]^{1/2} \cdot 
		\left[\mathbb{E}_{G_0} (\mathbb{E}_{G_0}[\mu^4\mid y,s])  + v_n \right]^{1/2} \\
		=& [\mathbb{E}_{G_0}\omega]^{1/2} \cdot 
		[\mathbb{E}_{G_0}[\mu^4]  + v_n ]^{1/2} \\
		\lesssim& \tau^{1/2}  \cdot 
		[\mathbb{E}_{G_0}[\mu^4]  + v_n ]^{1/2}
		\;= o(1/n)
	\end{aligned}
	\end{equation}
	where the final equality follows since $\mathbb{E}_{G_0}[\mu^4]$ is finite by Assumption~\ref{ass:subexp}, and $\tau = o(1/n^3)$ by Assumption~\ref{ass:subexp} and definitions of $\bar{\mu}_n$ and $\bar{\sigma}_n$.
\end{proof}

The following proves that the MSE outside of the region $\mathcal{D}\cap\mathcal{A}$ is $o(1/n)$.
\begin{lemma}[MSE outside of $\mathcal{D}\cap\mathcal{A}$]
	\label{lem:complement}
	$\mathbb{E}_{G_{0,n}} \left[\Delta_{\mu,\hat{G},G_{0,n}}(y_1,s_1)
	[\mathcal{D}\cap\mathcal{A}]^c\right]=o(1/n)$.
\end{lemma}
\begin{proof}
	We use $\Delta$ as shorthand for $\Delta_{\mu,\hat{G},G_{0,n}}(y_1,s_1)$. 
	% We first show 
	% \begin{equation}
	% 	\mathbb{E}_{G_0}[ \Delta \mathcal{L}^c ] \leq 
	% 	[\mathbb{E}_{G_0}\Delta^2]^{1/2}[\mathbb{E}_{G_0}\mathcal{L}^c]^{1/2}
	% 	\lesssim v_n^{1/2} \tfrac{1}{n^{3/2}} 
	% 	= o(\tfrac{1}{n}).
	% 	\label{eq:complement.1}
	% \end{equation}
	% First note $\mathcal{L}$ is composed of the intersection of regions $\{ \vee_{i\leq n}|\mu_i| \leq \bar{\mu}_n\}$ and $\{ \vee_{i\leq n}|\sigma_i| \leq \bar{\sigma}_n\}$. Now
	% \begin{align*}
	% 	\mathbb{E}_{G_0}\{ \vee_{i\leq n}|\mu_i| > \bar{\mu}_n\}
	% 	\leq&
	% 	\sum_i\mathbb{E}_{G_0}\{ |\mu_i| > \bar{\mu}_n\}
	% 	\\ = &
	% 	C_1 n \exp\{ - \lambda_1L^{\alpha_1} (\log n)^{\gamma_1\alpha_1} \}
	% 	\\=&
	% 	C_1 \exp \{ \log n^{-\lambda_1L^{\alpha_1} + 1} \} 
	% 	\\=&
	% 	C_1 \exp \{ \log n^{-3} \} 
	% 	= O( \tfrac{1}{n^3} )
	% \end{align*}
	% where $\gamma_1\alpha_1=1$ and $-\lambda_1L^{\alpha_1}\leq-4$ follow by definitions of $(\gamma_1,L)$. 
	% Similarly, we get 
	% $\mathbb{E}_{G_0}\{ \vee_{i\leq n}|\sigma_i| > \bar{\sigma}_n\}= O( \tfrac{1}{n^3} )$ so that 
	% $\mathbb{E}_{G_0}\mathcal{L}^c = O(\tfrac{1}{n^3} )$. 

	% Further note that $\mathbb{E}_{G_0}\Delta^2$ is sub-polynomial since $G_0$ is sub-exponential and the sampling distribution of the observations is Gaussian. 
	% As a result, \eqref{eq:complement.1} obtains. 
	% By the above, we restrict outselves to the region $\mathcal{L}$. This implies that $\vee_{i\leq n} |\mu_i| \leq \bar{\mu}_n$ and $\vee_{i\leq n}\sigma_i \leq \bar{\sigma}_n$.  
	By the sub-Gaussianity of $(y_i,s_i^2)$ conditional on $(\mu_i,\sigma_i)$, we have 
	\begin{equation}
	\begin{aligned}[b]
		\mathbb{E}_{G_{0,n}}\left\{ \max_{i\leq n}y_i > \bar{y}_n \right\} \leq& n \exp\left( -\frac{1}{2}\left[\frac{\bar{y}_n-\bar{\mu}_n}{\bar{\sigma}_n}\right]^2\right)
		\leq n\exp\left( -\frac{1}{2} \log n^{2M} \right) = \frac{1}{n^{M-1}} \\
		\mathbb{E}_{G_{0,n}}\left\{ \max_{i\leq n}s_i > \bar{s}_n \right\} \leq & n\exp\left( -k\left[\frac{\bar{s}_n}{\bar{\sigma}_n} - 1\right]^2 \right) \leq n\exp\left(-k\log n^{M/k} \right) = \frac{1}{n^{M-1}}, 
	\end{aligned}
	\end{equation}
	so that $\mathbb{E}_{G_{0,n}}[\mathcal{D}^c] = O(n^{1-M})$. As a result 
	\begin{equation}
	\begin{aligned}[b]
		\mathbb{E}_{G_{0,n}}[\Delta \mathcal{D}^c] 
		\leq& 
		[\mathbb{E}_{G_{0,n}}\Delta^2]^{1/2} \cdot [\mathbb{E}_{G_{0,n}}\mathcal{D}^c]^{1/2}
		\\\leq&
		C
		\left[\mathbb{E}_{G_{0,n}}\max_{i\leq n} y_i^4\right]^{1/2} \cdot O(n^{1-M})
		\;=_{(1)} o(n^{-1})
	\end{aligned}
	\end{equation}
	where $=_{(1)}$ follows because $\mathbb{E}_{G_{0,n}}[\max_{i\leq n} y_i^4]$ is at most sub-polynomial, and $M\geq 2$.
	
	By the above, we further restrict ourselves to $\mathcal{D}$. 
	We next show
	$\mathbb{E}_{G_{0,n}}[\Delta\mathcal{D}\cdot\mathcal{A}^c] = o(n^{-1})$:
	since we are in $\mathcal{D}$, then $\hat{G}\in\mathcal{G}(\tfrac{1}{n},\tilde{\gamma}_1,\tilde{\gamma}_2,L,\underline{\sigma})$. Certainly $G_{0,n}$ also lies within the same class, since $\tilde{\gamma}_1\geq\gamma_1$ and $\tilde{\gamma}_2\geq\gamma_2$. Thus we have by Lemma~\ref{lem:vdv.thm.4.1}
	\begin{equation}
		\mathbb{E}_{G_{0,n}}\{ d(f_{\hat{G}},f_{G_{0,n}}) > \varepsilon_n \} \lesssim \exp(-C [\log n]^3 ) 
		= o(n^{-2}) 
	\end{equation}
	and as a result
	$
		\mathbb{E}_{G_{0,n}}[\Delta \mathcal{D} \cdot \mathcal{A}^c] 
		= v_n \cdot o(n^{-2})
		= o(n^{-1})
	$,
	because $\Delta $ is at most sub-polynomial in region $\mathcal{D}$. 
\end{proof}

\subsection{MSEs within region $\mathcal{D}$}

This section provides adaptations of results from \cite{Jiang2009} that reduce the MSE from using a density estimate in place of the true density $f_{G_{0,n}}$ to the Hellinger distance between the density estimate and the true density. 
In this regard, Theorems~\ref{thm:jz.p3.sigma2}, \ref{thm:jz.p3.sigma} and \ref{thm:jz.p3.mu} are adaptations of Proposition 3 of \cite{Jiang2020} to the setting of unknown heteroskedasticity for the estimation objectives of $\{\sigma_i^2\}_{i\in[n]}$, $\{\sigma_i\}_{i\in[n]}$ and $\{\mu_i\}_{i\in[n]}$ respectively.
Lemmas~\ref{lem:jz.lem.a1.sub} is used in the proofs of all three theorems, whereas Lemmas~\ref{lem:jz.lem.a1} and \ref{lem:jz.lem.a1.aux} are specifically for handling the additional differentiation term required in the proof of Theorem~\ref{thm:jz.p3.mu}.
\begin{theorem}%[JZ Proposition 3 For $\sigma^2$]
	\label{thm:jz.p3.sigma2}
	For two distributions $G$ and $G_1$,
	we have
	\begin{align*}
		\mathbb{E}_{G}^{(y,s)}
		\left\{\big[
		\hat{\sigma}_{G_1}^2(y,s;\rho_n)
		- 
		\hat{\sigma}_{G}^2(y,s;\rho_n)
		\big]^2\mathcal{D}\right\}
		\leq& 
		\big(\bar{\sigma}_{G_1}^4 + \bar{\sigma}_{G}^4\big)
		2d^2(f_{G_1},f_{G})
		+
		4k^2
		\tfrac{\bar{s}_n^4}{k-2}
		d^2\big(f_{G_1},f_{G}\big).
	\end{align*}
\end{theorem}
\begin{proof}
	Recall that by definition of $\hat{\sigma}_{G}^2(y,s;\rho_n)$,
	\begin{align}
		&\mathbb{E}_{G}^{(y,s)}
		\left\{\big[
		\hat{\sigma}_{G_1}^2(y,s;\rho_n)
		- 
		\hat{\sigma}_{G}^2(y,s;\rho_n)
		\big]^2\mathcal{D}\right\}
		\nonumber\\	
		=&
		\mathbb{E}_{G}^{(y,s)}
		\bigg[
		\frac{k}{f_{G_1}(y,s^2)\vee \rho_n} \cdot \int_{s^2}^\infty \left[\frac{s^2}{t}\right]^{k-1} f_{G_1}(y,t) dt
		- 
		\frac{k}{f_{G}(y,s^2)\vee \rho_n} \cdot \int_{s^2}^\infty \left[\frac{s^2}{t}\right]^{k-1} f_{G}(y,t) dt
		\bigg]^2
		\mathcal{D}.
		\label{eq:jz.p3.sigma2.1}
	\end{align}
	Now define 
	\begin{equation}
		w^* := [f_{G_1}(y,s^2)\vee\rho_n + f_{G}(y,s^2)\vee\rho_n]^{-1}
		\label{eq:defn.wstar}
	\end{equation}
	and note that for $\tilde{G}=G_1$ or $G$, algebraic manipulation gives $\tfrac{k}{f_{\tilde{G}}\vee\rho_n} - 2kw^* = (f_{G}\vee\rho_n - f_{G_1}\vee\rho_n) w^* / (f_{\tilde{G}}\vee\rho_n)$ (up to sign), and so
	\begin{align}
		&\mathbb{E}_{G}^{(y,s)} \bigg[
		\left(\tfrac{k}{f_{\tilde{G}}(y,s^2)\vee\rho_n} - 2kw^*\right)
		\int_{s^2}^\infty \left[\tfrac{s^2}{t}\right]^{k-1} f_{\tilde{G}}(y,t)dt
		\bigg]^2 \mathcal{D}
		\nonumber\\
		\leq_{(1)}&
		\bar{\sigma}_{\tilde{G}}^4
		\cdot
		\mathbb{E}_{G}^{(y,s)} \big[
		(f_{G}(y,s^2)  - f_{G_1}(y,s^2) )w^*
		\big]^2 \mathcal{D}
		\label{eq:jz.p3.sigma2.triangle.1}
	\end{align}
	where $\leq_{(1)}$ bounds $\hat{\sigma}^2_{\tilde{G}}(y,s\mid\rho_n) \leq \bar{\sigma}_{\tilde{G}}^2$ and uses $|a\vee\rho_n - b\vee\rho_n| \leq |a-b|$ for any $a,b,\rho_n>0$.
	
	Inserting \eqref{eq:jz.p3.sigma2.triangle.1} into \eqref{eq:jz.p3.sigma2.1} and bounding using the triangle inequality gives 
	\begin{align}
		&\mathbb{E}_{G}^{(y,s)}
		[\hat{\sigma}_{G_1}^2(y,s;\rho_n) - \hat{\sigma}_{G}^2(y,s;\rho_n)]^2 \mathcal{D}
		\nonumber\\
		\leq&
		\big(\bar{\sigma}_{G_1}^4 + \bar{\sigma}_{G}^4\big)
		\mathbb{E}_{G}^{(y,s)} \big[
		(f_{G}  - f_{G_1} )w^* 
		\big]^2
		\nonumber\\
		+&
		4k^2
		\mathbb{E}_{G}^{(y,s)}
		\left(
		w^*
		\int_{s^2}^\infty 
		\left[\tfrac{s^2}{t}\right]^{k-1}
		\left[f_{G_1}(y,t) - f_{G}(y,t)\right] dt
		\right)^2 \mathcal{D}. 
		\label{eq:jz.p3.sigma2.triangle.2}
	\end{align}
	We bound the two terms on the RHS of \eqref{eq:jz.p3.sigma2.triangle.2} in turn. For the first term, we have
	\begin{align}
		&\mathbb{E}_{G}^{(y,s)} \big[
		(f_{G}(y,s^2)  - f_{G_1}(y,s^2) )w^* 
		\big]^2 
		\nonumber\\
		=&
		\mathbb{E}_{G}^{(y,s)} \bigg[
		(f_{G}^{1/2}(y,s^2)  - f_{G_1}^{1/2}(y,s^2) ) (f_{G}^{1/2}(y,s^2)  + f_{G_1}^{1/2}(y,s^2) ) w^* 
		\bigg]^2
		\nonumber\\
		=&
		\int 
		\left[f_{G}^{1/2}(y,s^2)  - f_{G_1}^{1/2}(y,s^2) \right]^2
		\cdot 
		\left[ f_{G}^{1/2}(y,s^2)  + f_{G_1}^{1/2}(y,s^2) \right]^2 w^* 
		\cdot 
		w^* f_{G}(y,s^2)d(y,s^2)
		\nonumber\\
		\leq&
		d^2(f_{G_1},f_{G}) \cdot 2
	\end{align}
	because $w^*f_{G}(y,s^2) \leq 1$ and $\big[f_{G}(y,s^2)^{1/2} + f_{G_1}^{1/2}(y,s^2)\big]^2w^* \leq 2$ for all $(y,s^2)$.
	
	For the second term on the RHS of the inequality \eqref{eq:jz.p3.sigma2.triangle.2}:
	\begin{align}
		&\mathbb{E}_{G}^{(y,s)}
		\left(
		w^*
		\int_{s^2}^\infty 
		\left[\tfrac{s^2}{t}\right]^{k-1}
		\left[f_{G_1}(y,t) - f_{G}(y,t)\right]dt
		\right)^2 \mathcal{D}
		\nonumber\\
		=&
		\int_{\mathcal{D}} 
		\left(
		w^*
		\int_{s^2}^\infty 
		\left[\tfrac{s^2}{t}\right]^{k-1}
		\left[f_{G_1}(y,t) - f_{G}(y,t)\right]dt
		\right)^2 f_{G}(y,s^2)d(y,s^2) 
		\nonumber\\
		\leq_{(1)}&
		\int_{\mathcal{D}}  
		\frac{
			\big(\int_{s^2}^\infty 
			\big[\tfrac{s^2}{t}\big]^{k-1}
			\big[f_{G_1}(y,t) - f_{G}(y,t)\big]dt\big)^2
		}{
			f_{G_1}(y,s^2)\vee\rho_n + f_{G}(y,s^2)\vee\rho_n
		}
		d(y,s^2)
		\nonumber\\
		=&
		\int_{\mathcal{D}}  
		\left\{
		\int_{s^2}^\infty 
		\big[\tfrac{s^2}{t}\big]^{k-1}
		\left[ 
		\tfrac{
			f_{G_1}(y,t)\vee\rho_n + f_{G}(y,t)\vee\rho_n
		}{
			f_{G_1}(y,s^2)\vee\rho_n + f_{G}(y,s^2)\vee\rho_n
		}
		\right]^{1/2}
		\left[
		\tfrac{
			f_{G_1}(y,t) - f_{G}(y,t)
		}{(f_{G_1}(y,t)\vee\rho_n + f_{G}(y,t)\vee\rho_n)^{1/2}}
		\right]
		dt
		\right\}^2
		d(y,s^2)
		\nonumber\\
		\leq_{(2)}&
		\int_{\mathcal{D}}  
		\left\{
		\int_{s^2}^\infty 
		\big[\tfrac{s^2}{t}\big]^{k-1}\cdot \big[\tfrac{s^2}{t}\big]^{k-1}
		\left[ 
		\tfrac{
			f_{G_1}(y,t)\vee\rho_n + f_{G}(y,t)\vee\rho_n
		}{
			f_{G_1}(y,s^2)\vee\rho_n + f_{G}(y,s^2)\vee\rho_n
		}
		\right]
		dt
		\cdot 
		\int_{s^2}^\infty 
		\left[
		\tfrac{
			f_{G_1}(y,t) - f_{G}(y,t)
		}{(f_{G_1}(y,t)\vee\rho_n + f_{G}(y,t)\vee\rho_n)^{1/2}}
		\right]^2
		dt
		\right\}
		d(y,s^2)
		\label{eq:jz.p3.sigma2.triangle.3}
	\end{align}
	where $\leq_{(1)}$ follows by $\tfrac{f_{G}(y,s^2)}{f_{G_1}(y,s^2)\vee\rho_n+f_{G}(y,s^2)\vee\rho_n}\leq 1$ $\forall (y,s^2)$ and $\leq_{(2)}$ by Cauchy-Schwarz.
	
	We now focus on the first term within the first inner integral of \eqref{eq:jz.p3.sigma2.triangle.3}. Substituting the Gamma density, the factor $\gamma(t\mid k,\theta)/\gamma(s^2\mid k,\theta) = [t/s^2]^{k-1}\exp(-(t-s^2)/\theta)$ cancels the $[s^2/t]^{k-1}$ in front, leaving an integrand bounded by $\exp(-(t-s^2)/\theta) \leq 1$ for $t>s^2$. Hence
	\begin{equation}
		\{t>s^2\}\big[\tfrac{s^2}{t}\big]^{k-1}
		\tfrac{f_{\tilde{G}}(y,t)}{f_{\tilde{G}}(y,s^2)}
		\leq 1
		\quad\text{for } \tilde{G}\in\{G_1,G\}.
		\label{eq:jz.p3.sigma2.triangle.4}
	\end{equation}
	Now plugging \eqref{eq:jz.p3.sigma2.triangle.4} into \eqref{eq:jz.p3.sigma2.triangle.3}
	gives an upper bound on \eqref{eq:jz.p3.sigma2.triangle.3} as 
	\begin{align}
		&\int_{\mathcal{D}}
		\left\{
		\int_{s^2}^\infty 
		\big[ \tfrac{s^2}{t} \big]^{k-1} 2 
		dt
		\cdot 
		\int_{s^2}^\infty 
		\left[
		\tfrac{
			f_{G_1}(y,t) - f_{G}(y,t)
		}{(f_{G_1}(y,t)\vee\rho_n + f_{G}(y,t)\vee\rho_n)^{1/2}}
		\right]^2
		dt
		\right\}
		d(y,s^2)
		\nonumber\\
		\leq&
		\int_{s\leq\bar{s}_n} \int_{|y|\leq\bar{y}_n}
		\left\{
		\tfrac{2}{k-2}s^2
		\cdot 
		\int_{\mathbb{R}_+} 
		\left[
		\tfrac{
			f_{G_1}(y,t) - f_{G}(y,t)
		}{(f_{G_1}(y,t)\vee\rho_n + f_{G}(y,t)\vee\rho_n)^{1/2}}
		\right]^2
		dt
		\right\}
		dyds^2
		\nonumber\\
		=&
		\tfrac{2}{k-2}
		\int_{s\leq\bar{s}_n} 
		s^2
		ds^2
		\cdot 
		\int_{|y|\leq\bar{y}_n} 
		\int_{\mathbb{R}_+} 
		\left[
		\tfrac{
			f_{G_1}(y,t) - f_{G}(y,t)
		}{(f_{G_1}(y,t)\vee\rho_n + f_{G}(y,t)\vee\rho_n)^{1/2}}
		\right]^2
		dtdy
		\nonumber\\
		=&
		\tfrac{1}{k-2}
		\bar{s}_n^4
		\cdot 
		\int_{|y|\leq\bar{y}_n} 
		\int_{\mathbb{R}_+} 
		\left[
		\tfrac{
			f_{G_1}(y,t) - f_{G}(y,t)
		}{(f_{G_1}(y,t)\vee\rho_n + f_{G}(y,t)\vee\rho_n)^{1/2}}
		\right]^2
		dtdy 
		\nonumber\\
		\leq &
		\tfrac{1}{k-2}
		\bar{s}_n^4
		\cdot d^2\big(f_{G_1},f_{G}\big)
		\label{eq:jz.p3.sigma2.triangle.6}
	\end{align}
	and the last inequality follows by Lemma~\ref{lem:jz.lem.a1.sub}.
\end{proof}

\begin{theorem}%[JZ Proposition 3 For $\sigma$]
	\label{thm:jz.p3.sigma}
	%	{\color{red} requires conditions for Lemma~\ref{lem:unif.bd.dens} to be satisfied AND requires the gaussian mixture structure of $f_G$, and requires the integration to be wrt $G$.}
	For two distributions $G$ and $G_1$, we have
	\begin{align*}
		&\mathbb{E}_{G}^{(y,s)}
		[\hat{\sigma}_{G_1}(y,s;\rho_n) - \hat{\sigma}_{G}(y,s;\rho_n)]^2 \mathcal{D}
		\nonumber\\
		\leq&
		\left\{
		\bar{\sigma}_{G_1}^2  +
		\bar{\sigma}_{G}^2
		\right\} 
		\cdot 2d^2\big(f_{G_1},f_{G}\big)
		+
		C \left[\bar{s}_n^2 \log n + \bar{s}_n^4 \right] \left[\tfrac{1}{n} + d^2\big(f_{G_1},f_{G}\big)\right]
	\end{align*}
\end{theorem}
\begin{proof}
	Following through the arguments as in \eqref{eq:jz.p3.sigma2.1} to \eqref{eq:jz.p3.sigma2.triangle.2} gives us 
	\begin{align}
		&\mathbb{E}_{G}^{(y,s)}
		[\hat{\sigma}_{G_1}(y,s;\rho_n) - \hat{\sigma}_{G}(y,s;\rho_n)]^2 \mathcal{D}
		\nonumber\\
		\leq&
		\left\{
		\bar{\sigma}_{G_1}^2  +
		\bar{\sigma}_{G}^2
		\right\}
		\cdot 2d^2\big(f_{G_1},f_{G}\big)
		\nonumber\\
		+&
		4k^2
		\mathbb{E}_{G}^{(y,s)}
		\left(
		w^*
		\int_{s^2}^\infty 
		\tfrac{1}{\sqrt{k(t-s^2)}}
		\left[\tfrac{s^2}{t}\right]^{k-1}
		\left[f_{G_1}(y,t) - f_{G}(y,t)\right] dt
		\right)^2 \mathcal{D}
		\label{eq:jz.p3.sigma.triangle.0}
	\end{align}
	And further following the arguments in \eqref{eq:jz.p3.sigma2.triangle.3} to \eqref{eq:jz.p3.sigma2.triangle.4} (splitting $(t-s^2)^{-1} = (t-s^2)^{-2\gamma_n} \cdot (t-s^2)^{-2(1/2-\gamma_n)}$ for $\gamma_n:=\frac{1}{2}-\frac{1}{\log n}$, applying Cauchy-Schwarz, bounding the ratio of densities by 1, and interchanging the order of integration) gives us
	\begin{align}
		&
		\mathbb{E}_{G}^{(y,s)}
		\left(
		w^*
		\int_{s^2}^\infty
		\left[\tfrac{1}{(t-s^2)^{1/2}}\right]
		\left[\tfrac{s^2}{t}\right]^{k-1}
		\left[f_{G_1}(y,t) - f_{G}(y,t)\right] dt
		\right)^2 \mathcal{D}
		\nonumber\\
		\leq&
		\int_{ s\leq\bar{s}_n}
		\bigg\{
		\int_{s^2}^\infty
		\tfrac{1}{(t-s^2)^{2\gamma_n}} \cdot
		\big[\tfrac{s^2}{t}\big]^{k-1}
		dt
		\cdot
		\int_{|y|\leq\bar{y}_n}
		\int_{s^2}^\infty
		\tfrac{1}{(t-s^2)^{2\left(\frac{1}{2}-\gamma_n\right)}} \cdot
		\left[
		\tfrac{
			f_{G_1}(y,t) - f_{G}(y,t)
		}{(f_{G_1}(y,t)\vee\rho_n + f_{G}(y,t)\vee\rho_n)^{1/2}}
		\right]^2
		dt dy
		\bigg\}
		ds^2.
		\label{eq:jz.p3.sigma.triangle.1}
	\end{align}
	
	We bound the first inner term on the RHS of the final equality \eqref{eq:jz.p3.sigma.triangle.1} as follows:
	\begin{align}
		\int_{s^2}^\infty
		\tfrac{1}{(t-s^2)^{2\gamma_n}} \cdot
		\big[\tfrac{s^2}{t}\big]^{k-1}
		dt
		\leq&
		\int_{s^2}^{s^2+v_n}
		\tfrac{1}{(t-s^2)^{2\gamma_n}}
		dt +
		\tfrac{1}{v_n^{2\gamma_n}}
		\int_{s^2+v_n}^\infty
		\big[\tfrac{s^2}{t}\big]^{k-1}
		dt
		\nonumber\\
		=&
		\tfrac{ v_n^{1-2\gamma_n} }{ 1-2\gamma_n }
		+
		\tfrac{1}{v_n^{2\gamma_n}}
		\tfrac{1}{k-2}  (s^2+v_n) \left[\tfrac{s^2}{s^2+v_n}\right]^{k-1}
		\nonumber\\
		\leq_{(1)}&
		v_n^{\tfrac{1}{\log n}} \left[ \tfrac{1}{2}\log n + \tfrac{1}{k-2}(s^2+1) \right]
		\nonumber\\
		\leq_{(2)}&
		C\left[\log n + s^2 \right]
		\label{eq:jz.p3.sigma.triangle.2}
	\end{align}
	where $v_n \uparrow\infty$ is an arbitrary sub-polynomial sequence, $\leq_{(1)}$ follows from substituting in the definition of $\gamma_n$, and $\leq_{(2)}$ follows from $v_n^{\frac{2}{\log n}} \downarrow 0$.
	
	We bound the second (inner) term on the RHS of \eqref{eq:jz.p3.sigma.triangle.1} as follows:
	\begin{align}
		&\int_{y\leq\bar{y}_n} \int_{s^2}^\infty 
		\tfrac{1}{(t-s^2)^{2\left(\frac{1}{2}-\gamma_n\right)}} \cdot
		\left[
		\tfrac{
			f_{G_1}(y,t) - f_{G}(y,t)
		}{(f_{G_1}(y,t)\vee\rho_n + f_{G}(y,t)\vee\rho_n)^{1/2}}
		\right]^2
		dt dy
		\nonumber\\
		\leq_{(1)}&
		2
		\int_{\mathbb{R}} \int_{ [s^2,s^2+\delta_n] \cup [s^2+\delta_n,\infty] } 
		\tfrac{1}{(t-s^2)^{2\left(\frac{1}{2}-\gamma_n\right)}} \cdot
		\left[
		f_{G_1}^{1/2}(y,t) - f_{G}^{1/2}(y,t)
		\right]^2 
		dt dy
		\label{eq:jz.p3.sigma.triangle.3}
	\end{align}
	where $\delta_n := \tfrac{1}{n}\downarrow0$, and $\leq_{(1)}$ follows like in the proof of Lemma~\ref{lem:jz.lem.a1.sub}.
	Now for the integration over $[s^2,s^2+\delta_n]$, we have 
	\begin{align}
		&\int_{ [s^2,s^2+\delta_n] } 
		\tfrac{1}{(t-s^2)^{2\left(\frac{1}{2}-\gamma_n\right)}} \cdot
		\int_{\mathbb{R}} \left[
		f_{G_1}^{1/2}(y,t) - f_{G}^{1/2}(y,t)
		\right]^2 
		dy dt
		\nonumber\\
		\leq&
		2 \int_{ [s^2,s^2+\delta_n] } 
		\tfrac{1}{(t-s^2)^{2\left(\frac{1}{2}-\gamma_n\right)}} \cdot
		\int_{\mathbb{R}} \left[
		f_{G_1}(y,t) + f_{G}(y,t)
		\right]
		dy dt
		\nonumber\\
		\leq&	
		2 \sup_{t>0}\bigg| \int_{\mathbb{R}} f_{G_1}(y,t) + f_{G}(y,t) dy \bigg| \cdot
		\int_{ [s^2,s^2+\delta_n] } \tfrac{1}{(t-s^2)^{2\left(\frac{1}{2}-\gamma_n\right)}} dt
		\nonumber\\
		\leq_{(1)}&	
		C \cdot \tfrac{\delta_n^{2\gamma_n}}{2\gamma_n}
		\label{eq:jz.p3.sigma.triangle.4}
	\end{align}
	and $\leq_{(1)}$ follows by Lemma~\ref{lem:unif.bd.dens}, 
	whereas for the integration over $[s^2+\delta_n,\infty]$ we have
	\begin{align}
		&\int_{ [s^2+\delta_n,\infty] } 
		\tfrac{1}{(t-s^2)^{2\left(\frac{1}{2}-\gamma_n\right)}} \cdot
		\int_{\mathbb{R}} \left[
		f_{G_1}^{1/2}(y,t) - f_{G}^{1/2}(y,t)
		\right]^2 
		dy dt
		\nonumber\\
		\leq&
		\tfrac{1}{\delta_n^{2\left(\frac{1}{2}-\gamma_n\right)}} \cdot
		\int_{ [s^2+\delta_n,\infty] } 
		\int_{\mathbb{R}}
		\left[
		f_{G_1}^{1/2}(y,t) - f_{G}^{1/2}(y,t)
		\right]^2 
		dy dt
		\nonumber\\
		\leq&
		\tfrac{1}{\delta_n^{2\left(\frac{1}{2}-\gamma_n\right)}} \cdot
		d^2\big(f_{G_1},f_{G}\big).
		\label{eq:jz.p3.sigma.triangle.5}
	\end{align}
	By substituting \eqref{eq:jz.p3.sigma.triangle.4} and \eqref{eq:jz.p3.sigma.triangle.5} into
	\eqref{eq:jz.p3.sigma.triangle.3}, together with substituting in the definitions of $\gamma_n$ and $\delta_n$ we get 
	\begin{align}
		&\int_{y\leq\bar{y}_n} \int_{s^2}^\infty 
		\tfrac{1}{(t-s^2)^{2\left(\frac{1}{2}-\gamma_n\right)}} \cdot
		\left[
		\tfrac{
			f_{G_1}(y,t) - f_{G}(y,t)
		}{(f_{G_1}(y,t)\vee\rho_n + f_{G}(y,t)\vee\rho_n)^{1/2}}
		\right]^2
		dt dy
		\nonumber\\
		\leq&
		C \left[
		n^{\frac{2}{\log n}}  \cdot \tfrac{1}{n}
		+ n^{-\frac{2}{\log n}}\cdot d^2\big(f_{G_1},f_{G}\big)
		\right].
		\label{eq:jz.p3.sigma.triangle.6}
	\end{align}
	Note that $n^{\frac{2}{\log n}}$ is a constant. 
	
	Therefore, combining \eqref{eq:jz.p3.sigma.triangle.0}, \eqref{eq:jz.p3.sigma.triangle.1}, \eqref{eq:jz.p3.sigma.triangle.2}, and \eqref{eq:jz.p3.sigma.triangle.6} we get
	\begin{align}
		\mathbb{E}^{(y,s)}
		[\hat{\sigma}_{G_1}(y,s;\rho_n) - \hat{\sigma}_{G_1}(y,s;\rho_n)]^2 \mathcal{D}
		\leq&
		\left\{
		\bar{\sigma}_{G_1}^2  +
		\bar{\sigma}_{G}^2
		\right\}
		\cdot 2d^2\big(f_{G_1},f_{G}\big)
		\nonumber\\
		+&
		C \cdot 
		\left[\bar{s}_n^2 \log n + \bar{s}_n^4 \right]
		\cdot 
		\left[\tfrac{1}{n}
		+ d^2\big(f_{G_1},f_{G}\big)
		\right].
	\end{align}
\end{proof}

\begin{theorem}%[JZ Proposition 3 For $\mu$]
	\label{thm:jz.p3.mu}
	For $G$ and $G_1$ with $\sigma$-supports lower-bounded by $\underline{\sigma}$, we have
	\begin{align}
		\mathbb{E}^{(y,s)}_{G}
		[\hat{\mu}_{G_1}(y,s;\rho_n) - \hat{\mu}_{G}(y,s;\rho_n)]^2 \mathcal{D}
		\leq
		C\big[\bar{\mu}_{G_1}^2+\bar{\mu}_{G}^2 + \bar{s}_n^4\big(\tilde{L}(\rho_n)x_n^2+a_n^2\big)\big] d^2\big(f_{G_1},f_{G}\big)
		\nonumber
	\end{align}
	where $C>0$ is a constant depending on $\underline{\sigma}$, and where
	\begin{align}
		x_n 
		:=& \bar{\mu}_{G_1} + \bar{\mu}_{G} + 2\bar{y}_n,
		\nonumber\\
		a_n
		:=& \max\{\tilde{L}(\rho_n)+1, |\log d^2(f_{G_1},f_{G})|\}
		\nonumber\\ 
		\tilde{L}(\rho) 
		:=& \sqrt{-\log(2\pi \rho^2)} \text{ for } 0<\rho\leq\tfrac{1}{\sqrt{2\pi}}.
	\end{align}
\end{theorem}
\begin{proof}
	Following through the proof of Theorem~\ref{thm:jz.p3.sigma2} (i.e., from \eqref{eq:jz.p3.sigma2.1} up to before the final inequality of \eqref{eq:jz.p3.sigma2.triangle.6}), we will have 
	\begin{align}
		&\mathbb{E}^{(y,s)}
		[\hat{\mu}_{G_1}(y,s;\rho_n) - \hat{\mu}_{G}(y,s;\rho_n)]^2 \mathcal{D}
		\nonumber\\
		\leq&
		\big(\bar{\mu}_{G_1}^2 + \bar{\mu}_{G}^2 + \bar{y}_n^2\big)
		2d^2\big(f_{G_1},f_{G}\big)
		\nonumber\\
		+&
		4\left[\tfrac{k}{J}\right]^2
		\tfrac{1}{k-2}
		\bar{s}_n^4
		\cdot 
		\int_{|y|\leq\bar{y}_n} 
		\int_{\mathbb{R}_+} 
		\left[
		\tfrac{
			\partial_yf_{G_1}(y,t) - \partial_yf_{G}(y,t)
		}{(f_{G_1}(y,t)\vee\rho_n + f_{G}(y,t)\vee\rho_n)^{1/2}}
		\right]^2
		dtdy
	\end{align}
	where by Lemma~\ref{lem:jz.lem.a1} we have 
	\begin{align}
		\int_{|y|\leq\bar{y}_n} 
		\int_{\mathbb{R}_+} 
		\left[
		\tfrac{
			\partial_yf_{G_1}(y,t) - \partial_yf_{G}(y,t)
		}{(f_{G_1}(y,t)\vee\rho_n + f_{G}(y,t)\vee\rho_n)^{1/2}}
		\right]^2
		dtdy
		\leq 
		C \left[\tilde{L}(\rho_n)x_n^2 + a_n^2 \right] d^2(f_{G_1},f_{G}).
	\end{align}
	and therefore 
	\begin{align}
		&\mathbb{E}^{(y,s)}
		[\hat{\mu}_{G_1}(y,s;\rho_n) - \hat{\mu}_{G}(y,s;\rho_n)]^2 \mathcal{D}
		\nonumber\\
		\leq&
		C\big[\bar{\mu}_{G_1}^2+\bar{\mu}_{G}^2 + \bar{y}_n + \bar{s}_n^4\big(\tilde{L}(\rho_n)x_n^2+a_n^2\big)\big] d^2\big(f_{G_1},f_{G}\big).
	\end{align}
\end{proof}

\begin{lemma}%[JZ Lemma A.1. Sub]
	\label{lem:jz.lem.a1.sub}
	$
		\int 
		\left[
		f_{G_1}(y,s^2) - f_{G}(y,s^2)
		\right]^2 w^*
		d(y,s^2)
		\leq 2d^2(f_{G_1},f_{G})
	$
	where $w^*$ is from \eqref{eq:defn.wstar}.
\end{lemma}
\begin{proof}
	\begin{align}
		&\int
		\left[
		f_{G_1}(y,s^2) - f_{G}(y,s^2)
		\right]^2 w^*
		d(y,s^2)
		\nonumber\\
		=_{(1)}&
		\int
		\left[
		f_{G_1}^{1/2}(y,s^2) - f_{G}^{1/2}(y,s^2)
		\right]^2 
		\cdot 
		\left[
		f_{G_1}^{1/2}(y,s^2) + f_{G}^{1/2}(y,s^2)
		\right]^2
		w^*
		d(y,s^2)
		\nonumber\\
		\leq_{(1)}&
		2\int
		\left[
		f_{G_1}^{1/2}(y,s^2) - f_{G}^{1/2}(y,s^2)
		\right]^2 
		d(y,s^2)
		\nonumber\\
		=&
		2d^2(f_{G_1},f_{G})
	\end{align}
	where $=_{(1)}$ follows by the identity $x^2-y^2 = (x-y)(x+y)$, and $\leq_{(2)}$ follows because 
	$[f_{G_1}^{1/2}(y,s^2) + f_{G}^{1/2}(y,s^2)]^2 w^* \leq 2$.
\end{proof}

The following lemma (and Lemma~\ref{lem:jz.lem.a1.aux}) is an extension of Lemma~1 of \cite{Jiang2009} from a univariate to a bivariate density of sufficient statistics, which is relevant for our framework of unknown heteroskedasticity.
\begin{lemma}%[JZ Lemma A.1. Main]
	\label{lem:jz.lem.a1}
	%{\color{red} requires the Gaussian mixture property of $f_G$ and conditions for Lemma~\ref{lem:jz.lem.a1.aux}, and also $\rho_n$ to be small enough}
	For $G$ and $G_1$ with $\sigma$-supports lower-bounded by $\underline{\sigma}$, we have
	\begin{align}
		&\int_{\{|y|\leq\bar{y}_n\}\times\mathbb{R}_+} 
		\left[
		\partial_yf_{G_1}(y,s^2) - \partial_yf_{G}(y,s^2)
		\right]^2 w^*
		d(y,s^2)
		\leq
		C\cdot d^2(f_{G_1},f_{G})
		\left[ \tilde{L}^2(\rho_n)x_n^2 + a^2 \right]
		\nonumber
	\end{align}
	for $a^2:=\max\{\tilde{L}^2(\rho_n)+1,|\log d^2(f_{G_1},f_{G})| \}$ and $x_n:=\left[ \bar{\mu}_{G_1} + \bar{\mu}_{G} + 2\bar{y}_n \right] $, and where $C>0$ is a constant depending on $\underline{\sigma}$.
\end{lemma}
\begin{proof}
	We abbreviate the densities $f_{G}(y,s^2)$ and $f_{G_1}(y,s^2)$ as $f_{G}$ and $f_{G_1}$ respectively for notational compactness. 
	Let $D^l(\cdot)$ denote the derivative with respect to $y$; e.g., $D^lf_{G} := \tfrac{\partial^l}{\partial y^l} f_{G}(y,s^2)$.	
	
	First define
	$
		\triangle_l := 
		\left(\int_{\{|y|\leq\bar{y}_n\}\times\mathbb{R}_+} \left[D^l (f_{G_1}-f_{G})\right]^2 w^* d(y,s^2)\right)^{1/2} 
	$
	and note that integration by parts (with respect to the argument $y$) gives
	\begin{align}
		\triangle_l^2 
		=&
		\int_{\{|y|\leq\bar{y}_n\}\times\mathbb{R}_+}  D^l (f_{G_1}-f_{G})w^* \cdot D^l (f_{G_1}-f_{G}) dy ds^2
		\nonumber\\
		=&
		0 - 
		\int_{\{|y|\leq\bar{y}_n\}\times\mathbb{R}_+}  D^{l-1} (f_{G_1}-f_{G})
		\left\{ 
		D^{l+1}(f_{G_1}-f_{G})w^* + [D^l(f_{G_1}-f_{G})][Dw^*]
		\right\} dy ds^2
		\nonumber\\
		\leq& 
		\bigg| 
		\int_{\{|y|\leq\bar{y}_n\}\times\mathbb{R}_+}  D^{l-1} (f_{G_1}-f_{G}) w^{*1/2} \cdot 
		D^{l+1}(f_{G_1}-f_{G})w^{*1/2} dyds^2
		\bigg|
		\nonumber\\
		+&
		\bigg|
		\int_{\{|y|\leq\bar{y}_n\}\times\mathbb{R}_+}  D^{l-1} (f_{G_1}-f_{G}) \cdot 
		D^{l}(f_{G_1}-f_{G})Dw^* dyds^2
		\bigg|
		\nonumber\\
		\leq_{(1)}&
		\triangle_{l-1}\triangle_{l+1} 
		\nonumber\\
		+& 
		\left(
		\int_{\{|y|\leq\bar{y}_n\}\times\mathbb{R}_+}  \left\{ 
		D^{l-1} (f_{G_1}-f_{G})
		\right\}^2 |Dw^*| dyds^2
		\right)^{1/2}
		\nonumber\\
		&\cdot
		\left(
		\int_{\{|y|\leq\bar{y}_n\}\times\mathbb{R}_+}  \left\{ 
		D^{l} (f_{G_1}-f_{G})
		\right\}^2 |Dw^*| dyds^2
		\right)^{1/2}
		\label{eq:jz.lem.a1.1}
	\end{align}
	where $\leq_{(1)}$ is Cauchy-Schwarz. Now (let $f_{G}'$ denote the derivative of $f_{G}$ w.r.t $y$)
	\begin{align}
		| Dw^* |
		=
		\big| \partial_y\tfrac{1}{f_{G_1}\vee\rho_n + f_{G}\vee\rho_n} \big|
		\leq_{(0)}&
		w^*\left[\tfrac{|f'_{G_1}|}{f_{G_1}\vee\rho_n} + \tfrac{|f'_{G}|}{f_{G}\vee\rho_n}\right]
    	\nonumber\\
		\leq&
		w^*\left[\tfrac{|f'_{G_1}|}{f_{G_1}} + \tfrac{|f'_{G}|}{f_{G}}\right] 
		\nonumber\\
		=_{(1)}&
		w^*\left[
		\tfrac{J\left|\mathbb{E}_{G_1}[\mu | y,s,\sigma] - y\right| }{\sigma^2}
		+ 
		\tfrac{J\left|\mathbb{E}_{G}[\mu | y,s,\sigma] - y\right|}{\sigma^2}
		\right]
		\nonumber\\
		\leq& 
		w^*
		\tfrac{J}{\underline{\sigma}^2}\left[ \bar{\mu}_{G_1} + \bar{\mu}_{G} + 2|y| \right]
		\label{eq:jz.lem.a1.2}
	\end{align}
	where $\leq_{(0)}$ follows by case analysis on whether $f_{G_1},f_{G}$ exceed $\rho_n$ when differentiating $w^*$, and $=_{(1)}$ follows from \eqref{eq:thm.main.proof.1}.
	Inserting \eqref{eq:jz.lem.a1.2} into \eqref{eq:jz.lem.a1.1} gives 
	\begin{equation}
		\triangle_l^2 
		\leq 
		\triangle_{l-1}\triangle_{l+1}
		+
		\tfrac{J}{\underline{\sigma}^2}
		\underbrace{\left[ \bar{\mu}_{G_1} + \bar{\mu}_{G} + 2\bar{y}_n \right]}_{=: x_n}
		\triangle_{l-1}\triangle_{l}.
		\label{eq:jz.lem.a1.3}
	\end{equation}
	
	We now define $l_0\in\mathbb{N}$ to be the integer satisfying $l_0 < \tfrac{\tilde{L}(\rho_n)}{2} < l_0+1$, and
	\begin{equation}
		l^*:= \min\{l : \triangle_{l+1} \leq \tfrac{J}{\underline{\sigma}^2} x_nl_0\triangle_l \}.
	\end{equation}
	We split the analysis of $\triangle_1$ according to whether $l^* \leq l_0$.
	When $l^*\leq l_0$, we have for every $l<l^*$ that $\triangle_{l+1}>\tfrac{J}{\underline{\sigma}^2}x_nl_0\triangle_l$, or equivalently, that $\triangle_l < \tfrac{1}{J\underline{\sigma}^{-2}x_nl_0} \triangle_{l+1}$. Inserting this final inequality into \eqref{eq:jz.lem.a1.3} gives
	\begin{equation}
		\triangle_l^2 \leq \triangle_{l-1}\triangle_{l+1}[1+l_0^{-1}]
		\label{eq:jz.lem.a1.3.0}
	\end{equation}
	and dividing throughout by $\triangle_{l}\triangle_{l-1}$ gives 
	\begin{equation}
		\tfrac{\triangle_{l}}{\triangle_{l-1}} \leq [1+l_0^{-1}] \tfrac{\triangle_{l+1}}{\triangle_{l}}
	\end{equation}
	for every $l<l^*$. And thus iterating this through every $l<l^*$ gives
	\begin{equation}
		\tfrac{\triangle_1}{\triangle_0} 
		\leq 
		[1+l_0^{-1}] \tfrac{\triangle_2}{\triangle_1} 
		\leq
		[1+l_0^{-1}]^{l^*-1} \tfrac{\triangle_{l^*}}{\triangle_{l^*-1}}.
		\label{eq:jz.lem.a1.4}
	\end{equation}
	We now divide \eqref{eq:jz.lem.a1.3} for $l=l^*$ by $\triangle_{l^*}\triangle_{l^*-1}$ to get 
	\begin{equation}
		\tfrac{\triangle_{l^*}}{\triangle_{l^*-1}} 
		\leq 
		\tfrac{\triangle_{l^*+1}}{\triangle_{l^*}} + \tfrac{J}{\underline{\sigma}^2} x_n
	\end{equation}
	which we insert into \eqref{eq:jz.lem.a1.4} to get 
	\begin{align}
		\tfrac{\triangle_1}{\triangle_0} 
		\leq
		[1+l_0^{-1}]^{l^*-1} \left[
		\tfrac{\triangle_{l^*+1}}{\triangle_{l^*}}+\tfrac{J}{\underline{\sigma}^2}x_n
		\right]
		\leq_{(1)}&
		[1+l_0^{-1}]^{l^*-1}[l_0+1] \tfrac{J}{\underline{\sigma}^2}x_n
		\nonumber\\
		\leq_{(2)}&
		\tfrac{1}{1+l_0^{-1}} e [1+l_0] \tfrac{J}{\underline{\sigma}^2}x_n
		\nonumber\\
		=&
		e\cdot l_0 \cdot \tfrac{J}{\underline{\sigma}^2} x_n 
		\nonumber\\
		\leq_{(3)}&
		\tfrac{eJ}{\underline{\sigma}^2} \tfrac{\tilde{L}(\rho_n)}{2} x_n
	\end{align}
	where $\leq_{(1)}$ follows from definition of $l^*$ and $\leq_{(2)}$ follows\footnote{By assumption of this part of analysis -- see line before \eqref{eq:jz.lem.a1.3.0}.} from $l^* \leq l_0$ and $\leq_{(3)}$ follows from the definition of $l_0$. We thus have $\triangle_1^2 \leq \tfrac{(eJ)^2}{4\underline{\sigma}^4}\cdot d^2(f_{G_1},f_{G})\tilde{L}^2(\rho_n)x_n^2$ by Lemma~\ref{lem:jz.lem.a1.sub}. This concludes the analysis under $l^*\leq l_0$.
	
	When $l_0 < l^*$, we use arguments identical as above to show that for every $l\leq l_0$, 
	\begin{equation}
		\begin{aligned}[b]
		\tfrac{\triangle_1}{\triangle_0}
		\leq& 
		[1+l_0^{-1}]^l \tfrac{\triangle_{l+1}}{\triangle_{l}}
		\\
		\text{so that }
		\left(\tfrac{\triangle_1}{\triangle_0}\right)^{l_0+1}
		\leq& 
		\left[
		\prod_{l=0}^{l_0} (1+l_0)^l \tfrac{\triangle_{l+1}}{\triangle_{l}}
		\right]
		\\
		\text{ or }
		\tfrac{\triangle_1}{\triangle_0}
		\leq&
		\left[
		\prod_{l=0}^{l_0} (1+l_0)^l
		\right]^{\tfrac{1}{l_0+1}} 
		\left[\tfrac{\triangle_{l_0+1}}{\triangle_0}\right]^{\tfrac{1}{l_0+1}}
		\\
		=& 
		(1+l_0^{-1})^{l_0/2}
		\cdot 
		\left[\tfrac{\triangle_{l_0+1}}{\triangle_0}\right]^{\tfrac{1}{l_0+1}}
		\end{aligned}
		\label{eq:jz.lem.a1.5}
	\end{equation}
	where the final equality of \eqref{eq:jz.lem.a1.5} may be directly verified. 
	
	Now since $w^*\leq (2\rho_n)^{-1}$, coupled with the fact that $a^2\geq 2(l_0+1)-1$ by definitions of $(a,l_0)$ for small enough $\rho_n$, applying Lemma~\ref{lem:jz.lem.a1.aux} and bounding $\exp(-a^2)\leq d^2(f_{G_1},f_{G})$ (by definition of $a$) and $a\geq 1$ gives
	\begin{equation}
		\triangle_{l_0+1}^2
		\leq
		C
		\tfrac{1}{\rho_n\sqrt{2\pi}}d^2(f_{G_1},f_{G})a^{2(l_0+1)}.
		\label{eq:jz.lem.a1.6}
	\end{equation}
	We then insert \eqref{eq:jz.lem.a1.6} into \eqref{eq:jz.lem.a1.5} to get 
	\begin{align}
		\triangle_1
		\leq&
		(1+l_0^{-1})^{l_0/2} \cdot \triangle_0^{\tfrac{l_0}{l_0+1}}
		\cdot 
		\left[
		\tfrac{1}{\rho_n\sqrt{2\pi}}d^2(f_{G_1},f_{G})a^{2(l_0+1)} C
		\right]^{\tfrac{1}{2(l_0+1)}}
		\nonumber\\
		\leq_{(1)}&
		\sqrt{e} d(f_{G_1},f_{G}) a \cdot (2\pi\rho_n^2)^{-\tfrac{1}{4l_0+4}} \cdot C^{\frac{1}{2(l_0+1)}}
		\nonumber\\
		\leq_{(2)}&
		\sqrt{e} d(f_{G_1},f_{G}) a \cdot \exp(\tfrac{1}{2}) \cdot C^{\frac{1}{2(l_0+1)}}
	\end{align}
	where $\leq_{(1)}$ follows by using Lemma~\ref{lem:jz.lem.a1.sub}, and $\leq_{(2)}$ follows because 
	\begin{align}
		-\tfrac{1}{4l_0+4}\log(2\pi\rho_n^2) 
		=&
		\tfrac{1}{4l_0+4} \tilde{L}(\rho_n)
		\leq_{(3)}
		\tfrac{1}{4l_0+4} (2l_0+2)
		\;=\tfrac{1}{2}
	\end{align}
	where $\leq_{(3)}$ follows by definition of $l_0$.
	Note that $C^{\frac{1}{2(l_0+1)}}\rightarrow 1$, because $l_0\uparrow\infty$ by $\rho_n\downarrow0$.
\end{proof}

\begin{lemma}%[Aux Lemma for JZ L.A.1]
	\label{lem:jz.lem.a1.aux}
	%{\color{red} requires the Gaussian mixture property of $f_G$, and conditions for Lemma~\ref{lem:unif.bd.dens} which is the lower bound on $\underline{\sigma}$.}
	Let $G$ and $G_1$ be distributions with $\sigma$-supports lower-bounded by $\underline{\sigma}$. For all $l\in\mathbb{N}$ and $a\geq \sqrt{2l-1}$,
	\begin{align}
		&\int_{|y|\leq\bar{y}_n} 
		\int_{\mathbb{R}_+} 
		\{
		D^l (f_{G_1} - f_{G})
		\}^2 d(y,s^2)
		\leq
		C \cdot
		\left[ a^{2l} d^2(f_{G_1},f_{G}) 
		+ a^{2l-1}\exp(-a^2)\right]
	\end{align}
	where $D^l$ represents the $l^{th}$ derivative wrt $y$, and $C>0$ depends only on $\underline{\sigma}$.
\end{lemma}
\begin{proof}
	First define $h^*(u;s^2) := \int e^{iux} h(y,s^2) dy$ where $h(y,s^2)$ is an arbitrary integrable function. By a change of variables $\tilde{y} = (y-\mu)/(\sigma/\sqrt{J})$ and the Fourier transform of the Gaussian density, the inner integral evaluates as $(\sigma/\sqrt{J})^{-1}\int e^{iuy}\phi((y-\mu)/(\sigma/\sqrt{J})) dy = \exp(iu\mu)\exp(-u^2/2)$, so
	\begin{equation}
		| f_{G_1}^*(u;s^2) |
		=
		\bigg|
		\exp(iu\mu) \exp(-\tfrac{u^2}{2}) \int \gamma(s^2 | k,\theta) dG(\mu,\sigma)
		\bigg|
		\leq
		\exp(-\tfrac{u^2}{2}) f_{G_1}(s^2).
		\label{eq:jz.lem.a1.a.3}
	\end{equation}
	Therefore it follows from the Plancherel identity, splitting the integration domain at $\{|u|\leq a\}\cup\{|u|>a\}$, and using \eqref{eq:jz.lem.a1.a.3} on the tail that
	\begin{align}
		&\int \{
		D^l (f_{G_1} - f_{G})
		\}^2 dy
		\nonumber\\
		=&
		\tfrac{1}{2\pi} \int u^{2l} | f_{G_1}^*(u;s^2) - f_{G}^*(u;s^2) |^2 du
		\nonumber\\
		\leq_{(1)}&
		\tfrac{1}{2\pi}a^{2l} \int
		| f_{G_1}^*(u;s^2) - f_{G}^*(u;s^2) |^2 du
		+
		\tfrac{1}{2\pi} \int_{ |u| > a } u^{2l} \exp(-u^2) [f_{G_1}(s^2) + f_{G}(s^2)]^2 du
		\nonumber\\
		\leq_{(2)}&
		a^{2l} \int | f_{G_1}(y,s^2) - f_{G}(y,s^2) |^2 dy
		+ \tfrac{1}{2\pi} [f_{G_1}(s^2) + f_{G}(s^2)]^2 \cdot a^{2l-1} \exp(-a^2).
	\end{align}
	Integrating over $s^2$, then applying the identity $x^2-y^2=(x+y)(x-y)$, Cauchy-Schwarz, the $C_p$-inequality, and Lemma~\ref{lem:unif.bd.dens},
	\begin{equation}
		\int \{
		D^l (f_{G_1} - f_{G})
		\}^2 d(y,s^2)
		\leq
		C \cdot \left[ a^{2l} d^2(f_{G_1},f_{G}) + a^{2l-1}\exp(-a^2) \right].
	\end{equation}
	Here $\leq_{(1)}$ enlarges the $\{|u|\leq a\}$ domain to $\mathbb{R}$ and uses \eqref{eq:jz.lem.a1.a.3}, and $\leq_{(2)}$ applies Plancherel in reverse to the first term and the calculus inequality $\int_{\{|u| > a\}} u^{2l} \exp(-u^2)du \leq a^{2l-1}\exp(-a^2)$ for $a\geq\sqrt{2l-1}$ to the second.
\end{proof}

\begin{lemma}[Bounded densities]
	\label{lem:unif.bd.dens}
	%{\color{red} requires the lower bound on $\sigma$: $\sigma\geq\underline{\sigma}$.}
	Let $\theta:=\sigma^2/k$, and suppose that $\sigma\geq\underline{\sigma}$. Then
	\begin{align}
		\tfrac{1}{\sigma}\phi(\tfrac{y-\mu}{\sigma}) \leq& \tfrac{1}{\sqrt{2\pi} \underline{\sigma}^2 } 
		\qquad\text{ and }\qquad
		\gamma(s^2\mid k,\theta) \leq \tfrac{k(k-1)^{k-1}}{\Gamma(k)} \tfrac{1}{\underline{\sigma}^2}.
		\nonumber
	\end{align}
\end{lemma}
\begin{proof}
	For the first statement, $\tfrac{1}{\sigma}\phi(\tfrac{y-\mu}{\sigma}) \leq \tfrac{1}{\sqrt{2\pi} \underline{\sigma}^2 }$ is immediate. 
	For the second statement, note that the gamma density $\gamma(s^2 \mid k,\theta)$ attains a maximum (over $s^2$) of  
	\begin{equation}
		\begin{aligned}[b]
			\tfrac{1}{\Gamma(k)\theta^k} [(k-1)\theta]^{k-1} \exp\left(-(k-1)\right)
			\leq&
			\tfrac{k(k-1)^{k-1}}{\Gamma(k)} \tfrac{1}{\underline{\sigma}^2}.
		\end{aligned}
	\end{equation}
\end{proof}

\subsection{NPMLE density estimation}

This section provides adaptations of results in \cite{Ghosal2001} on maximum likelihood estimation of normal mixtures to accommodate the setting\footnote{That is, to accommodate bivariate distributions on $(\mu,\sigma)$ with possibly growing supports for both $\mu$ and $\sigma$ and also multiple observations per unit.} of Theorem~\ref{thm:eb.opt}. 
Lemmas~ \ref{lem:vdv.lem.3.4}, \ref{lem:vdv.thm.3.3.1}+\ref{lem:vdv.thm.3.3.2}  and \ref{lem:vdv.thm.4.1} are extensions of Lemma~3.4, Theorem~3.3 and Theorem~4.1 of \cite{Ghosal2001} respectively.

\begin{remark}[Relation to literature on location-scale Gaussian mixtures]
\label{rem:relation.cdb}
Viewing our model as $n$ iid draws of $J$-dimensional vectors where each is from a multivariate Gaussian with common mean and spherical covariance matrix, we are able to exploit this special structure to achieve better (i.e. poly-logarithimic) entropy bounds relative to the literature, such as the linear entropy bounds in \cite{Canale2017} that allow for more general mean and covariance matrix structures; see Remark 2(ii) of \cite{Canale2017}. 
% Our results match those of \cite{Ghosal2001} when $\sigma$ is a known constant (homoskedasticity), so the logarithmic orders are sharp; the excess over the fixed-support case is the cost of the growing supports and the Gamma factor.
Indeed abstracting from the fact that we have multiple observations per unit, and assuming a fixed support for $\sigma$, our entropy bound in Lemma~\ref{lem:vdv.thm.3.3.2} matches that of Theorem 3.3 of \cite{Ghosal2001}.
\end{remark}

Let 
$
\mathcal{F}(\epsilon,v_1,v_2,C,\underline{\sigma}) := \{ f_G: G \in \mathcal{G}(\epsilon,v_1,v_2,C,\underline{\sigma}) \}
$,
where 
$\mathcal{G}(\cdot)$ is as defined in \eqref{eq:def.trunc}. These two classes are sometimes shorthand as $\mathcal{F}(\epsilon)$ and $\mathcal{G}(\epsilon)$ respectively.
\begin{lemma} %[SC Lemma B.11 / VG Theorem 4.1]
	\label{lem:vdv.thm.4.1}
	Suppose
	$
	f_{0,n}\in\mathcal{F}(\frac{1}{n}).
	$
	Define
	\begin{equation*}
		\varepsilon_n := n^{-1/2} (\log n)^{2\left[v_1\vee\frac{1}{2} + v_2\right] + \frac{1}{2}}.
	\end{equation*}
	Let $\hat{f}_n \in\mathcal{F}(n^{-1})$ satisfy
	\begin{equation*}
		n^{-1}\sum_{i=1}^n \hat{f}_n(y_i,s_i^2) \geq 
		\sup_{f\in\mathcal{F}(n^{-1})}
		n^{-1}\sum_{i=1}^n f(y_i,s_i^2)
		- 
		\tfrac{1}{24} \varepsilon_n^2.
	\end{equation*}
	Then there exist positive constants\footnote{See e.g., final remark of Section 3 of \cite{Wong1995} for choice of constants.} $(c_1,c_2)$ such that
	\begin{equation*}
		\mathbb{P}\left\{
		d(\hat{f}_n,f_{0,n}) > c_1 \varepsilon_n
		\right\}
		\lesssim 
		\exp\big(-c_2 (\log n)^3\big).
	\end{equation*}
\end{lemma}
\begin{proof}
	By Lemma~\ref{lem:vdv.thm.3.3.2}, we have that 
	\begin{equation}
		\log N_{[\,]}\big(\varepsilon_n,\mathcal{F}(n^{-1}),d\big) 
		\lesssim 
		\big(\log\tfrac{1}{\varepsilon_n}\big)^{4\big[v_1\vee\frac{1}{2}+v_2\big]+1}.
	\end{equation}
	Then, using $\varepsilon_n = n^{-1/2}(\log n)^{2[v_1\vee\frac{1}{2}+v_2]+\frac{1}{2}}$ and that the ratio of $(\tfrac{1}{2}\log n - C\log\log n)^p$ to $(\log n)^p$ is bounded,
	\begin{equation}
		\int_0^{\varepsilon_n} \sqrt{\log N_{[\,]}\big(u,\mathcal{F}(n^{-1}),d\big) } du
		\lesssim
		\big( \log \sqrt{n} - \big[2\big(v_1\vee\tfrac{1}{2}+v_2\big)+\tfrac{1}{2}\big]\log\log n \big)^{2\big[v_1\vee\tfrac{1}{2}+v_2\big]+\tfrac{1}{2}}
		\cdot \varepsilon_n
		\lesssim
		\sqrt{n} \varepsilon_n^2.
	\end{equation}
	Thus by Theorem 2 of \cite{Wong1995}, there exists positive constants $(c_1,c_2)$ such that, using $n\varepsilon_n^2 = (\log n)^{4[v_1\vee\frac{1}{2}+v_2]+1}$,
	\begin{equation}
		\mathbb{P}\left\{
		d(\hat{f}_n,f_{0,n}) > c_1 \varepsilon_n
		\right\}
		\leq
		4\exp\big(-c_2 n\varepsilon_n^2\big)
		\lesssim
		\exp\big(-c_2 (\log n)^3\big).
	\end{equation}
\end{proof}

\begin{lemma} %[SC Lemma B.12 / VG Theorem 3.3 II]
	\label{lem:vdv.thm.3.3.2}
	Let $v_1 \geq 1/2$, $v_2 \geq 0 $. We have
	$
		\log N_{[\,]}(\epsilon , \mathcal{F}(\epsilon),d)
		\lesssim 
		\big(\log\tfrac{1}{\epsilon}\big)^{4[v_1+v_2]+1}.
	$
\end{lemma}
\begin{proof}
	Let $\eta$ satisfying $\eta \leq \epsilon$ be a positive constant to be chosen later, and $f_i$ for $i=1,\dots,N$ be an $\eta$-net for $\mathcal{F}(\epsilon)$ in the sup-norm, with $N$ to be chosen later as well. 
	
	Define 
	\begin{equation}
	\begin{aligned}[b]
		E_1(y) :=& \begin{cases}
			\frac{\sqrt{J}}{\underline{\sigma}} \phi( \sqrt{J}\tfrac{y}{2\bar{\sigma}} ) & \text{ for } |y|> 2\bar{\mu} \\
			\frac{1}{\sqrt{2\pi}\underline{\sigma}} & \text{ otherwise.} 
		\end{cases}
		\\
		E_2(s^2) :=& \begin{cases}
			\frac{k^k}{\Gamma(k) \underline{\sigma}^2} [\tfrac{s^2}{\bar{\sigma}^2}]^{k-1} \exp(-k\tfrac{s^2}{\bar{\sigma}^2}) & \text{ for } s^2 > (k-1)\bar{\sigma}^2 \\
			\frac{k(k-1)^{k-1}}{\Gamma(k)} \frac{1}{\underline{\sigma}^2} & \text{ otherwise.}
		\end{cases}
	\end{aligned}
	\end{equation}
	Note that  $\frac{1}{\sigma} \phi( \frac{y-\mu}{\sigma/J^{1/2}} ) \leq E_1(y)$ and 
	$\gamma(s^2 \mid k,\theta) \leq H_2(s^2)$. Define $E(y,s^2)= E_1(y)E_2(s^2)$ and thus we have for any $f_G\in\mathcal{F}(\epsilon)$,
	\begin{equation}
		0 \leq f_G(y,s^2) \leq E(y,s^2).
	\end{equation}
	Define $l_i:= \max(f_i-\eta,0)$ and $u_i = \min( f_i+\eta , E )$ and note that $\mathcal{F}(\epsilon)\subseteq \cup_{i=1}^N [l_i,u_i]$, because if $|| p-f_i ||_\infty < \eta$, then $l_i(y,s) \leq p(y,s) \leq u_i(y,s)$. Furthermore, $u_i-l_i \leq \min(2\eta,E)$. Let $B:=\max(2C,\sqrt{8}\bar{\sigma})(\log 1/\eta)^{v_1} $ and $ D := C(\log 1/\eta)^{\frac{1}{2}+v_2}$. We now show that $||u_i-l_i||_1 \leq v_n \cdot \eta$. Splitting the integral over the bounded region $\{(|y|<B)\cap(s^2<D^2)\}$ (where $\int\min(2\eta,E) \leq 4BD^2\eta$) and its complement,
	\begin{equation}
	\begin{aligned}[b]
		\int_{ |y| > B } E(y,s^2) dyds^2
		=&
		C \cdot \int_{|y| > B} \exp(-\tfrac{1}{8\bar{\sigma}^2} y^2)dy \cdot\int E_2(s^2)ds^2
		\\\lesssim_{(1)}&
		\eta \cdot \bar{\sigma}^2
	\end{aligned}
	\end{equation}
	by a Mill's ratio bound, and similarly $\int_{ s^2 > D^2 } E(y,s^2) dyds^2 \lesssim \bar{\mu} \cdot \log(1/\eta)^{k-1+2v_2} \cdot \eta$ by integration by parts together with $D > \bar{\sigma}$.
	Thus $\int u_i - l_i \leq v_n \cdot \eta$, giving $N_{[\,]}(v_n\eta,\mathcal{F}(\epsilon),||\cdot||_1) \leq N$. By Lemma~\ref{lem:vdv.thm.3.3.1}, $N\lesssim (\log \frac{1}{\epsilon})^{4[v_1+v_2]+1}$; choosing $\eta v_n = \epsilon$ (so $\log \frac{1}{\eta} \sim \log \frac{1}{\epsilon}$),
	\begin{equation}
		N(\epsilon,\mathcal{F}(\epsilon),||\cdot||_1) \lesssim (\log \tfrac{1}{\epsilon})^{4[v_1+v_2]+1},
	\end{equation}
	and the lemma follows since $d^2(f,g) \leq ||f-g||_1$.
\end{proof}

\begin{lemma} %[SC Lemma B.13 / VG Theorem 3.3 I]
	\label{lem:vdv.thm.3.3.1}
	Let $v_1 \geq 1/2$, $v_2 \geq 0 $. We have
	\begin{equation}
		\log N(\epsilon , \mathcal{F}(\epsilon) ,||\cdot||_\infty)
		\lesssim 
		\big(\log\tfrac{1}{\epsilon}\big)^{4[v_1+v_2]+1}
	\end{equation}
\end{lemma}
\begin{proof}
	Let $\mathcal{F}_{\text{d}}(\epsilon):= \{ f_G : G\in\mathcal{G}_{\text{d}}(\epsilon) \}$ where 
	\begin{equation}
		\mathcal{G}_{\text{d}}(\epsilon) 
		:=
		\big\{
		G\in\mathcal{G}(\epsilon)
		\text{ is discrete with at most 
			$N\leq C\big(\log\tfrac{1}{\epsilon}\big)^{4[v_1+v_2]}$} 
		\text{ support points}
		\big\}.
	\end{equation}
	
	By Lemma~\ref{lem:vdv.lem.3.4}, $C$ is a constant that can be chosen such that $\mathcal{F}_{\text{d}}(\epsilon)$ is an $\epsilon$-net\footnote{In the $|| \cdot ||_\infty$ norm.} over $\mathcal{F}(\epsilon)$. Thus an $\epsilon$-net over $\mathcal{F}_{\text{d}}(\epsilon)$ will be an $2\epsilon$-net\footnote{Also in the $|| \cdot ||_\infty$ norm.} over $\mathcal{F}(\epsilon)$.
	
	Now we choose an $\epsilon$-net $\mathcal{S}$ over the $N$-dimensional simplex for the $L_1$ norm. By Lemma A.4 of \cite{Ghosal2001}, this can be chosen such that $|\mathcal{S}|\leq \big(\tfrac{5}{\epsilon}\big)^N$. 
	
	We further define $\mathcal{F}_{\text{d}}'(\epsilon) := \{ f_{G} : G\in\mathcal{G}_{\text{d}}'(\epsilon)\}$ where
	\begin{align}
		\mathcal{G}_{\text{d}}'(\epsilon)
		:=
		\big\{
		&G\in\mathcal{G}_{\text{d}}(\epsilon)
		\text{ with $N$ support points of the form } 
		(\pm\eta_1\epsilon,\underline{\sigma}+\eta_2\epsilon)
		\nonumber\\
		& \text{where $\eta_1,\eta_2=0,1,\dots,$ with its weights coming from $\mathcal{S}$}
		\big\}.
	\end{align}
	Clearly we have $\mathcal{F}_{\text{d}}'(\epsilon)\subset \mathcal{F}_{\text{d}}(\epsilon)$.
	
	Now for each
	\begin{equation}
		f_{G}(y,s^2) := \textstyle\sum_{j=1}^N w_j \phi(y |\mu_j,\frac{\sigma_j}{\sqrt{J}})\gamma(s^2|k,\theta_j) \in \mathcal{F}_{\text{d}}(\epsilon),
	\end{equation}
	we choose the 
	\begin{equation}
		f_{G'}(y,s^2) := \textstyle\sum_{j=1}^N w_j' \phi(y |\mu_j',\frac{\sigma_j'}{\sqrt{J}})\gamma(s^2|k,\theta_j') \in \mathcal{F}_{\text{d}}'(\epsilon),
	\end{equation}
	that is closest in the sense that 
	$|\mu_j-\mu_j^\prime| < \epsilon$, and $|\sigma_j - \sigma_j'| < \epsilon$, and 
	$\textstyle\sum_{j=1}^N |w_j - w_j'| < \epsilon$ for all $j=1,\dots,N$.
	Then
	\begin{align}
		&\sup_x
		\big|
		f_G(y,s^2) - f_{G'}(y,s^2)
		\big|
		\nonumber\\
		\leq&
		\sup_x 
		\bigg|
		\sum_j^N w_j 
		\bigg\{ \phi(y |\mu_j,\frac{\sigma_j}{\sqrt{J}})\gamma(s^2|k,\theta_j) - 
		\phi(y |\mu_j',\frac{\sigma_j}{\sqrt{J}})\gamma(s^2|k,\theta_j)
		\bigg\}
		\bigg|
		\nonumber\\
		+&
		\sup_x 
		\bigg|
		\sum_j^N w_j 
		\bigg\{ \phi(y |\mu_j',\frac{\sigma_j}{\sqrt{J}})\gamma(s^2|k,\theta_j) - 
		\phi(y |\mu_j',\frac{\sigma_j'}{\sqrt{J}})\gamma(s^2|k,\theta_j')
		\bigg\}
		\bigg|
		\nonumber\\
		+&
		\sup_x 
		\bigg|
		\sum_j^N [w_j - w_j'] \phi(y |\mu_j',\frac{\sigma_j'}{\sqrt{J}})\gamma(s^2|k,\theta_j')
		\bigg|.
		\label{eq:vdv.lem.3.4.1}
	\end{align}
	Using the fact that the derivatives of $\phi(y |\mu,\frac{\sigma}{\sqrt{J}})\gamma(s^2|k,\theta)$ with respect to $\sigma$ and $\mu$ are uniformly bounded for $\sigma \geq \underline{\sigma}$ and $k\geq1$, the first and second terms on the RHS of \eqref{eq:vdv.lem.3.4.1} are bounded by constant multiples of $\sup_{j\leq N}|\mu_j-\mu_j'| < \epsilon$ and $\sup_{j\leq N}|\sigma_j-\sigma_j'| < \epsilon$ respectively. And again using a similar argument we can show that the final term on the RHS of \eqref{eq:vdv.lem.3.4.1} is bounded by a constant multiple of $\textstyle\sum_{j=1}^N |w_j-w_j'|<\epsilon$.
	
	Thus we have $\sup_x
	\big|
	f_G(y,s^2) - f_{G'}(y,s^2)
	\big| \lesssim \epsilon$, which implies that $\mathcal{F}_{\text{d}}'(\epsilon)$ is a $\epsilon$-net over $\mathcal{F}_{\text{d}}(\epsilon)$, and consequently a $2\epsilon$-net over $\mathcal{F}(\epsilon)$.
	Furthermore, $|\mathcal{F}_{\text{d}}'(\epsilon)| \lesssim (\tfrac{2\bar{\mu}}{\epsilon})^N \cdot (\tfrac{\bar{\sigma}-\underline{\sigma}}{\epsilon})^N \cdot (\tfrac{5}{\epsilon})^N = (10\bar{\mu}[\bar{\sigma}-\underline{\sigma}])^N \epsilon^{-3N}$. Taking logs and using $N \lesssim (\log 1/\epsilon)^{4[v_1+v_2]}$, $\bar{\mu},\bar{\sigma} \lesssim (\log 1/\epsilon)^{v_1\vee v_2}$,
	\begin{equation}
		\log N(C_1\epsilon , \mathcal{F}(\epsilon),||\cdot||_\infty)
		\lesssim
		\big(\log\tfrac{1}{\epsilon}\big)^{4[v_1+v_2] + 1}.
	\end{equation}
\end{proof}

\begin{lemma} %[SC Lemma B.14 / VG Lemma 3.4]
	\label{lem:vdv.lem.3.4}
	Suppose $v_1\geq \frac{1}{2}$, $v_2\geq 0$. For any $G\in \mathcal{G}(\epsilon)$, there exists 
	a discrete $G'\in\mathcal{G}(\epsilon)$ 
	with at most $N\lesssim (\log\tfrac{1}{\epsilon})^{4[v_1+v_2]}$ support points such that 
	$
		|| f_G - f_{G'} ||_\infty < \epsilon. 
	$
\end{lemma}
\begin{proof}
	Define
	\begin{align}
		\mathcal{D}(\epsilon)
		:=& \bigg\{ 
		(y,s^2) : 
		|y| \leq \bar{\mu}  + \bar{\sigma}(\log 1/\epsilon^2)^{1/2} \text{ and } 
		s \leq \tfrac{1}{k^{1/2}}\bar{\sigma}(\log [1/\epsilon^2])^{1/2}
		\bigg\}
	\end{align}
	where $(\bar{\mu},\bar{\sigma})$ denote $(C(\log 1/\epsilon)^{v_1},C(\log 1/\epsilon)^{v_2})$. 
	
	Note that we can rewrite 
	\begin{equation}
		\phi(y | \mu,\tfrac{\sigma}{\sqrt{J}}) \cdot\gamma(s^2 | k,\theta)
		=
		M_J [s^2]^{\frac{J-3}{2}} \sigma^{-J} \exp\left(-\tfrac{1}{2\sigma^2} A(y,s^2,\mu)\right)
	\end{equation}
	where $M_J:= \sqrt{\frac{J}{2\pi}} \left(\frac{J-1}{2}\right)^{\frac{J-1}{2}} \div \Gamma(\frac{J-1}{2})$ and $A(y,s^2,\mu) :=J(y-\mu)^2 + (J-1)s^2$.
	
	\noindent\textsc{Taylor}: note the taylor approximation of the exponential function:
	\begin{equation}
		\left| \exp(-t) - \sum_{j=0}^{l-1} \frac{(-t)^j}{j!} \right| \leq \frac{(et)^l}{l^1}
	\end{equation}
	As a result, substituting in $t=\frac{1}{2\sigma^2}A$, 
	\begin{align}
		\label{eq:taylor}
		&\sup_{(y,s^2)\in\mathcal{D}(\epsilon)}| f_{G}(y,s^2) - f_{G'}(y,s^2) |
		\\\leq&  
		\sup_{(y,s^2)\in\mathcal{D}}
		\bigg| 
		M_J [s^2]^{\frac{J-3}{2}} \underbrace{\sum_{j=1}^{l-1} \frac{1}{j!} \int \frac{1}{\sigma^J} \left[-\frac{A}{2\sigma}\right]^j d(G-G')}_{\text{Moment}} 
		\bigg| 
		+
		2\sup_{(y,s^2)\in\mathcal{D}}
		\bigg| 
		M_J 
		\underbrace{[s^2]^{\frac{J-3}{2}} \sigma^{-J} \left[\frac{eA}{2\sigma^2l}\right]^l}_{\text{Remainder}}
		\bigg| 
	\end{align}
	
	\noindent\textsc{Extremities} ($\mathcal{F}$): 
	Where $(y,s^2)\notin\mathcal{D}(\epsilon)$, we have by Lemma B.9 that
	\begin{equation}
		\phi(y | \mu,\tfrac{\sigma}{\sqrt{J}}) \cdot\gamma(s^2 | k,\theta) = o(\epsilon) 
	\end{equation}
	and thus $f_G(y,s^2) = o(\epsilon)$ uniformly over $\mathcal{G}(\epsilon)$.
	
	\noindent\textsc{Taylor Remainder} ($\mathcal{F}$): 
	Note that in $\mathcal{D}(\epsilon)$, 
	\begin{equation}
	\begin{aligned}[b]
		A(y,s^2,\mu) \leq& \bar{\mu}^2 + \bar{\mu}\bar{\sigma}(\log 1/\epsilon^2)^{1/2} + \bar{\sigma}^2(\log 1/\epsilon^2)
		\leq
		[\log^{2v_1}(1/\epsilon) + \log^{2v_2+1} (1/\epsilon) ] 
		\\
		\text{ and }s^2 \leq& \log^{2v_2+1}(1/\epsilon)
	\end{aligned}
	\end{equation}
	and the remainder term of \eqref{eq:taylor} is bounded by
	\begin{equation}
		C
		\exp\left\{
		-l \bigg[\log l -\log ( \log^{2v_1}[1/\epsilon] + \log^{2v_2+1}[1/\epsilon] ) - \tfrac{J-3}{2l}\log (\log^{2v_2+1}[1/\epsilon]) + C\bigg]
		\right\},
	\end{equation}
	which is of the order $\exp\{ -l^*\}$ if $l^* \asymp \log^{2[v_1+v_2]}[1/\epsilon]$. Thus with $l=l^*$, the remainder is $\exp( -\log^{2(v_1+v_2)} 1/\epsilon ) \lesssim \epsilon$ since $v_1\geq1/2$.
	
	\noindent\textsc{Taylor Moment} ($\mathcal{F}$): Note that
	\begin{equation}
		\frac{1}{\sigma^J} \left[\frac{A}{\sigma}^2\right]^j =
		\frac{1}{\sigma^{J+2j}}(J[y-\mu]^2+[J-1]s^2) 
		=
		\frac{1}{\sigma^{J+2j}} \sum_{k=0}^{2j} a_k\mu^k 
	\end{equation}
	where $a_k$ for $k=0,\dots,2j$ are functions of $(y,s^2)$ and not of $(\mu,\sigma)$. As a result, in the moment term of \eqref{eq:taylor}, for each $j$ there are at most $(2j+1)$ powers of $(\mu,\sigma)$, and thus the entire moment term has at most $(l-1)^2$ powers of $(\mu,\sigma)$. With $l^*\asymp \log^{2(v_1+v_2)}[1/\epsilon]$, by Lemma A.1 of Ghosal and van der Vaart (2001), we can find a discrete distribution with $(l^*-1)^2$ number of support points such that the moment term of \eqref{eq:taylor} is zero. 
	
\end{proof}

\section{Results for Empirical Illustration}
\label{sec:app.supp}

\subsection{Microdata Normality}
\label{subsec:app.normality}

\begin{figure}[t!]
	\caption{Density and QQ Plots}
	\label{fig:app.qq}
	\begin{center}
		\begin{tabular}{cc}
			\includegraphics[width=.33\textwidth]{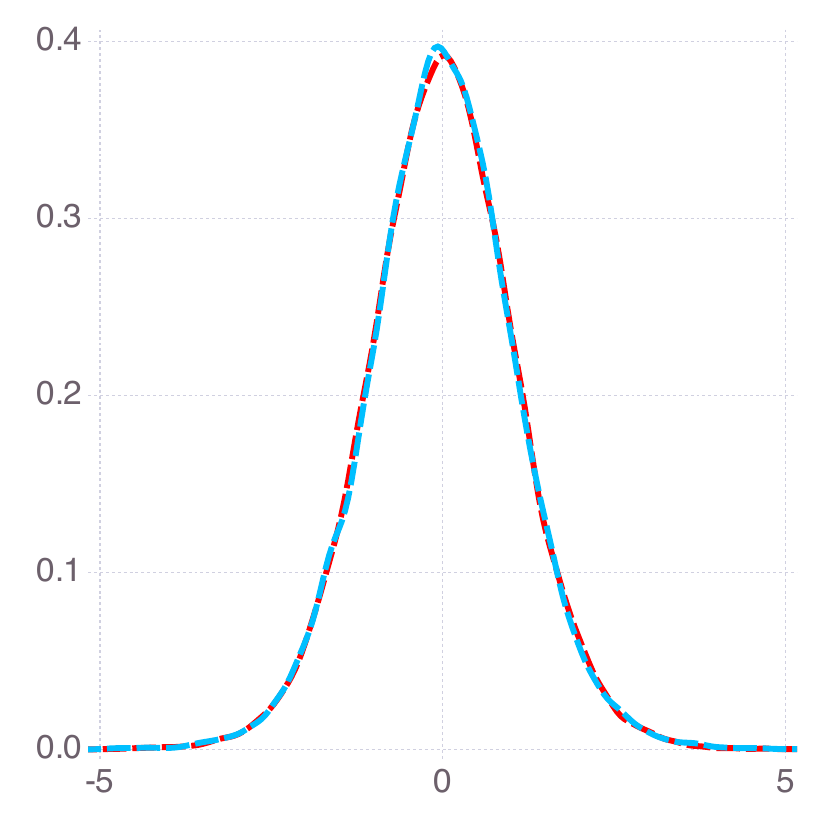}
			&
			\includegraphics[width=.33\textwidth]{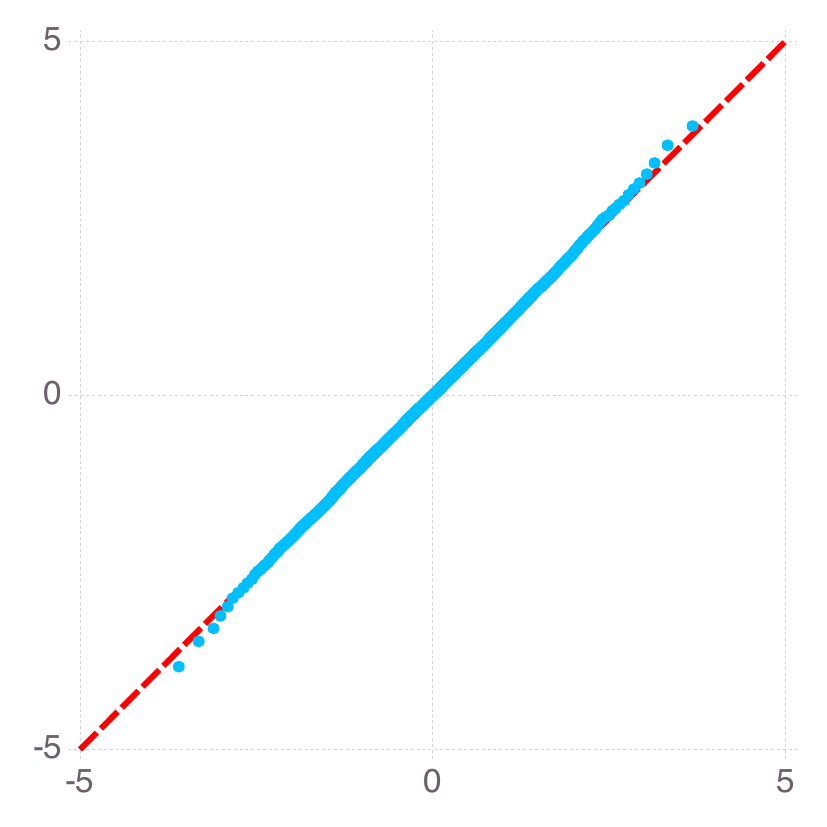}
		\end{tabular}
	\end{center}
	{\footnotesize {\em Notes:} Left panel plots the empirical density of the residuals \eqref{eq:ext.resid} (in blue) and the density of the $t$-mixture \eqref{eq:t.mixture} (in red). Right panel plots the residuals' empirical quantiles against the quantiles of the $t$-mixture, where the minimum and maximum quantiles are that of 0.001 and 0.999. }\setlength{\baselineskip}{4mm}
\end{figure}
We examine the microdata normality assumption: under \eqref{eq:app.model} of the paper, each studentized residual\footnote{See e.g. (4.13) of Montgomery, Peck \& Vining (2012).}
\begin{equation}
	r_j := \frac{y_j-y_{i(j)}}{\sqrt{\tilde{s}_{i(j)}^2 (1-1/J_{i(j)})  }} \sim t_{J_{i(j)}-2}
	\label{eq:ext.resid}
\end{equation}
where $i(j)$ is student $j$'s teacher identity and $\tilde{s}_{i(j)}^2$ is the leave-$j$-out sample variance for teacher $i(j)$. Thus under the null of \eqref{eq:app.model}, $r_j$ (marginalizing over $j$) follows a mixture of $t$-distributions:
\begin{equation}
	r_j \sim \sum_{i=1}^{n} \frac{J_i}{\sum_{i=1}^{n}J_i} \cdot t_{J_i-2}.
	\label{eq:t.mixture}
\end{equation}
A specification test based on the parametric bootstrap and either the Kolmogorov-Smirnov or Cramer-von Mises test statistics yield $p$-values of $0.34$ and $0.22$ respectively. Figure~\ref{fig:app.qq} displays the empirical density of the residuals against the density of \eqref{eq:t.mixture}, and also a QQ-plot. The results suggest that the microdata normality assumption is reasonable for the application, with only minimal differences in the tails.

\subsection{Unbiased Risk Estimators for the Validation Exercise}
\label{subsec:app.risk.est}
\begin{align}
	\hat{R}(\hat{\mu}_{\text{tr}},\mu)
	:=&
	\frac{1}{n}\sum_{i=1}^n
	[ \hat{\mu}_{\text{tr},i} - y_{\text{te},i} ]^2 - \tfrac{s_i^2}{J_{\text{te},i}}
	\nonumber\\
	\hat{R}(\hat{q}_{0.1,\text{tr}},q_{0.1})
	:=& 
	\frac{1}{n}\sum_{i=1}^n
	[ \hat{q}_{0.1,\text{tr},i} - \hat{q}_{0.1,\text{te},i}^{\text{ub}} ]^2 - \tfrac{s_i^2}{J_{\text{te},i}} - 1.28^2s_i^2 \left[
		\tfrac{J_{\text{te},i}-1}{2}
		\left(
		\tfrac{\Gamma\!\left(\frac{J_{\text{te},i}-1}{2}\right)}
		{\Gamma\!\left(\frac{J_{\text{te},i}}{2}\right)}
		\right)^{2}
		- 1
	\right]
\end{align}
where ``te'' subscript indicates construction based on $\mathcal{Y}^{\text{te}}$, and
$\hat{q}_{0.1,\text{te},i}^{\text{ub}} := y_{\text{te},i} - 1.28\hat{\sigma}_{\text{te},i}^{\text{ub}}$ where
\begin{equation}
		\hat{\sigma}_{\text{te},i}^{\text{ub}}
		:=
		s_{\text{te},i}
		\left(
		\tfrac{J_{\text{te},i}-1}{2}
		\tfrac{\Gamma\left(\frac{J_{\text{te},i}-1}{2}\right)}{\Gamma\left(\frac{J_{\text{te},i}}{2}\right)}
		\right).
	\end{equation}
\begin{lemma}
	\label{lm:unbiased.risk.1}
	Suppose that Assumption~\ref{ass:disagg.model} holds. Then,
	$\mathbb{E}_{G_0}[ \hat{R}(\hat{\mu}_{\text{tr}},\mu) ] = R_{G_0}(\hat{\mu}_{\text{tr}},\mu)$ and $\mathbb{E}_{G_0}[ \hat{R}(\hat{q}_{0.1,\text{tr}},q_{0.1}) ] = R_{G_0}(\hat{q}_{0.1,\text{tr}},q_{0.1})$.
\end{lemma}
\begin{proof}
	Conditional on $(\mu_i,\sigma_i)$, 
	\begin{align}
		&\mathbb{E}_{\mu,\sigma}[ \hat{\mu}_{\text{tr},i} - y_{\text{te},i} ]^2\nonumber\\
		=&
		\mathbb{E}_{\mu,\sigma}[\hat{\mu}_{\text{tr},i} - \mu_i]^2 
		+ 
		\mathbb{E}_{\mu,\sigma}[\mu_i - y_{\text{te},i} ]^2
		+
		\mathbb{E}_{\mu,\sigma}[\hat{\mu}_{\text{tr},i} - \mu_i][\mu_i - y_{\text{te},i} ] 
		\nonumber\\
		=&
		\mathbb{E}_{\mu,\sigma}[\hat{\mu}_{\text{tr},i} - \mu_i]^2 
		+ 
		\mathbb{E}_{\mu,\sigma}[\mu_i - y_{\text{te},i} ]^2
		+
		\mathbb{E}_{\mu,\sigma}[\hat{\mu}_{\text{tr},i} - \mu_i]
		\mathbb{E}_{\mu,\sigma}[\mu_i - y_{\text{te},i} ] 
		\nonumber\\
		=&
		\mathbb{E}_{\mu,\sigma}[\hat{\mu}_{\text{tr},i} - \mu_i]^2 
		+ 
		\tfrac{\sigma_i^2}{J_{\text{te},i}}
		+
		0
	\end{align}
	and an unbiased estimator of $\frac{\sigma_i^2}{J_{\text{te},i}}$ is $\frac{s_i^2}{J_{\text{te},i}}$. Integrating out $(\mu_i,\sigma_i)$ yields
	\begin{equation}
		\mathbb{E}_{G_0}\left([ \hat{\mu}_{\text{tr},i} - y_{\text{te},i} ]^2 - \tfrac{s_i^2}{J_{\text{te},i}}\right)
		=
		\mathbb{E}_{G_0}[ \hat{\mu}_{\text{tr},i} - \mu_i]^2.
	\end{equation}
	Average over $i$ yields the first part of the lemma. 

	Next, note that 
	$\hat{\sigma}_{\text{te},i}^{\text{ub}}$
	unbiasedly estimates $\sigma_i$, so 
	$\mathbb{E}_{\mu,\sigma}[\hat{q}_{0.1,\text{te},i}^{\text{ub}}] = q_{0.1,i}$.
	Following as above, we have
	\begin{align}
		&\mathbb{E}_{\mu,\sigma}[ \hat{q}_{0.1,\text{tr},i} - \hat{q}_{0.1,\text{te},i}^{\text{ub}} ]^2\nonumber\\
		=&
		\mathbb{E}_{\mu,\sigma}[\hat{q}_{0.1,\text{tr},i} - q_{0.1,i}]^2 
		+ 
		\mathbb{E}_{\mu,\sigma}[q_{0.1,i} - \hat{q}_{0.1,\text{te},i}^{\text{ub}}]^2
		\nonumber\\
		=&
		\mathbb{E}_{\mu,\sigma}[\hat{q}_{0.1,\text{tr},i} - q_{0.1,i}]^2 
		+ 
		\mathbb{E}_{\mu,\sigma}[\mu_i - y_{\text{te},i} ]^2
		+
		1.28^2\mathbb{E}_{\mu,\sigma}[\sigma_i-\hat{\sigma}_{\text{te},i}^{\text{ub}}]^2 
		\nonumber\\
		=&
		\mathbb{E}_{\mu,\sigma}[\hat{q}_{0.1,\text{tr},i} - q_{0.1,i}]^2  
		+ 
		\tfrac{\sigma_i^2}{J_{\text{te},i}}
		+
		1.28^2 \mathbb{V}_{\mu,\sigma}[\hat{\sigma}_{\text{te},i}^{\text{ub}}].
	\end{align}
	An unbiased estimator of $\frac{\sigma_i^2}{J_{\text{te},i}}$ is given above, and that of  $\mathbb{V}_{\mu,\sigma}[\hat{\sigma}_{\text{te},i}^{\text{ub}}]$ 
	is 
	\begin{equation}
		s_i^2 \left[
		\frac{J_{\text{te},i}-1}{2}
		\left(
		\frac{\Gamma\!\left(\frac{J_{\text{te},i}-1}{2}\right)}
		{\Gamma\!\left(\frac{J_{\text{te},i}}{2}\right)}
		\right)^{2}
		- 1
		\right].
	\end{equation}
	As a result, an unbiased estimator of 
	$\mathbb{E}_{\mu,\sigma}[\hat{q}_{0.1,\text{tr},i} - q_{0.1,i}]^2 $
	is
	\begin{equation}
		[ \hat{q}_{0.1,\text{tr},i} - \hat{q}_{0.1,\text{te},i}^{\text{ub}} ]^2 - \frac{s_i^2}{J_{\text{te},i}} - 1.28^2s_i^2 \left[
		\frac{J_{\text{te},i}-1}{2}
		\left(
		\frac{\Gamma\!\left(\frac{J_{\text{te},i}-1}{2}\right)}
		{\Gamma\!\left(\frac{J_{\text{te},i}}{2}\right)}
		\right)^{2}
		- 1
		\right].
	\end{equation}
	Integrating out $(\mu_i,\sigma_i)$ yields the above as an unbiased estimator of $\mathbb{E}_{G_0}[\hat{q}_{0.1,\text{tr},i} - q_{0.1,i}]^2 $, and averaging over $i$ yields $\hat{R}_{G_0}(\hat{q}_{0.1,\text{tr}},q_{0.1})$ as an unbiased estimator of $R(\hat{q}_{0.1,\text{tr}},q_{0.1})$.
\end{proof}

    \end{appendix}

\end{document}